\documentclass[a4paper,12pt]{article}
\usepackage{amsmath}
\usepackage{amssymb,amsthm,graphicx}
\usepackage{mathtools}
\usepackage{enumitem}
\usepackage{color}
\usepackage{epsfig}
\usepackage{graphics}
\usepackage[font=small]{caption}
\usepackage{subcaption}
\usepackage[hang,flushmargin]{footmisc} 
\usepackage{placeins}
\usepackage{rotating}
\usepackage{float}
\usepackage{booktabs}
\usepackage[mathscr]{euscript}
\usepackage{natbib}
\usepackage{setspace}
\usepackage{mathrsfs}
\usepackage{multirow}
\usepackage{algorithm2e}
\usepackage[left=3cm,right=3cm,bottom=3cm,top=3cm]{geometry}
\usepackage{pdfsync} % activates backward search from skim
\usepackage{hyperref}
\hypersetup{
hidelinks,
}

\newif\ifdraftmode
\draftmodefalse
\newcommand{\colmath}[1]{\ifdraftmode\textcolor{red}{#1}\else{#1}\fi}

\newcommand{\pr}{{\colmath{\mathbb{P}}}}        % probability
\newcommand{\ex}{{\colmath{\mathbb{E}}}}        % expectation
\newcommand{\ind}{{\colmath{\boldsymbol{1}}}}   % indicator function
\newcommand{\idmat}{{\colmath{\boldsymbol{I}}}} % identity matrix

\DeclareMathOperator*{\argmin}{arg\,min}

\newcommand{\norm}[1]{\ensuremath{\| #1 \|}}
\newcommand{\normsup}[1]{\ensuremath{\| #1 \|_\infty}}

\newcommand{\normone}[1]{\ensuremath{\| #1 \|_1}}
\newcommand{\normtwo}[1]{\ensuremath{\| #1 \|_2}}
\newcommand{\normtwos}[1]{\ensuremath{\| #1 \|_2^2}}

\newcommand{\R}{\ensuremath{{\colmath{\mathbb{R}}}}}
\newcommand{\Rp}{\ensuremath{{\colmath{\mathbb{R}^p}}}}
\newcommand{\Rn}{\ensuremath{{\colmath{\mathbb{R}^n}}}}
\newcommand{\Rnp}{\ensuremath{{\colmath{\mathbb{R}^{n\times p}}}}}

\newcommand{\Y}{\ensuremath{{\colmath{Y}}}}                 % response vector
\newcommand{\Xmat}{\ensuremath{{\colmath{\boldsymbol{X}}}}} % design matrix
\newcommand{\Xvec}{\ensuremath{{\colmath{X}}}}              % columns of design matrix
\newcommand{\eps}{\ensuremath{{\colmath{\varepsilon}}}}     % noise vector
\newcommand{\uvec}{\ensuremath{{\colmath{u}}}}              % noise vector
\newcommand{\activeset}{\ensuremath{{\colmath{S}}}}         % active set
\newcommand{\criterion}{\ensuremath{{\colmath{\hat{Q}}}}}
\newcommand{\crithat}{\ensuremath{{\colmath{\hat{\Pi}}}}}     
\newcommand{\crit}{\ensuremath{{\colmath{\Pi}}}}
\newcommand{\critstar}{\ensuremath{\colmath{\Pi}^{\colmath{*}}}} 
\newcommand{\critstarstar}{\ensuremath{\colmath{\Pi}^{\colmath{**}}}} 
\newcommand{\critgauss}{\ensuremath{\colmath{\Pi}^{\colmath{G}}}}
\newcommand{\What}{\ensuremath{{\colmath{\hat{W}}}}}     
\newcommand{\W}{\ensuremath{{\colmath{W}}}}
\newcommand{\Wstar}{\ensuremath{\colmath{W}^{\colmath{*}}}}
\newcommand{\Wstarstar}{\ensuremath{\colmath{W}^{\colmath{**}}}}
\newcommand{\Wgauss}{\ensuremath{{\colmath{G}}}}
\newcommand{\remainder}{\ensuremath{{\colmath{R}}}}

\theoremstyle{plain}
\newtheorem{theorem}{Theorem}[section]

\newtheorem{prop}{Proposition}[section]
\newtheorem{lemma}{Lemma}[section]
\newtheorem{remark}{Remark}[section]

\newtheorem{lemmaA}{Lemma}[section]
\newtheorem{propA}{Proposition}[section]

\newtheorem{remarkA}{Remark}[section]

\makeatletter
\newcommand{\lefteqno}{\let\veqno\@@leqno}
\makeatother

\newcommand{\heading}[1]
{  \setcounter{page}{1}
   \begin{center}

   {\LARGE \textbf{#1}}
   \end{center}
}

\newcommand{\headingsupp}[2]
{  \setcounter{page}{1}
   \begin{center}

   {\LARGE \textbf{#1}}
   \vspace{0.3cm}

   {\LARGE \textbf{#2}}
   \end{center}
}

\newcommand{\authors}[4]
{  \parindent0pt
   \begin{center}
      \begin{minipage}[c][2cm][c]{5cm}
      \begin{center} 
      {\large #1} 
      #2 
      \end{center}
      \end{minipage}
      \begin{minipage}[c][2cm][c]{5cm}
      \begin{center} 
      {\large #3}
      #4 
      \end{center}
      \end{minipage}
   \end{center}
}

\renewcommand{\baselinestretch}{1.2}
\allowdisplaybreaks[1]

\usepackage{xcolor}
\usepackage{ifthen}
\newcommand{\comments}{0} % \comments = 0 to hide Johannes' and Michael's comments, 
\ifthenelse{\comments=1}{\newcommand{\jl}[1]{{\footnotesize\color{olive}#1}}}{\newcommand{\jl}[1]{}} % Johannes' comments
\ifthenelse{\comments=1}{\newcommand{\mv}[1]{{\footnotesize\color{orange}#1}}}{\newcommand{\mv}[1]{}} % Michael's comments

\begin{document}

\heading{Estimating the Lasso's Effective Noise}
\vspace{-0.6cm}

\authors{Johannes Lederer\renewcommand{\thefootnote}{1}\footnotemark[1]}{Ruhr-University Bochum}{Michael Vogt\renewcommand{\thefootnote}{2}\footnotemark[2]}{Ulm University}
\footnotetext[1]{Address: Department of Mathematics, Ruhr-University Bochum, 44801 Bochum, Germany. Email: \texttt{johannes.lederer@rub.de}.}
\footnotetext[2]{Corresponding author. Address: Institute of Statistics, Department of Mathematics and Economics, Ulm University, 89081 Ulm, Germany. Email: \texttt{m.vogt@uni-ulm.de}.}
\vspace{-0.8cm}

\renewcommand{\baselinestretch}{1.0}
\renewcommand{\abstractname}{}
\begin{abstract}
\noindent 
Much of the theory for the lasso in the linear model $\Y = \Xmat \beta^* + \eps$ hinges on the quantity $2\normsup{\Xmat^\top \eps}/n$, which we call the lasso's effective noise. 
Among other things, the effective noise plays an important role in finite-sample bounds for the lasso, the calibration of the lasso's tuning parameter, and inference on the parameter vector $\beta^*$.
In this paper, we develop a bootstrap-based estimator of the quantiles of the effective noise.
The estimator is fully data-driven, that is, does not require any additional tuning parameters.
We equip our estimator with finite-sample guarantees and apply it to tuning parameter calibration for the lasso and to high-dimensional inference on the parameter vector $\beta^*$.
\end{abstract}
\renewcommand{\baselinestretch}{1.2}



\setlength{\parindent}{15pt}
\numberwithin{equation}{section}

\section{Introduction}

Consider the high-dimensional linear model 
$\Y = \Xmat \beta^* + \eps$
with response vector $\Y \in \Rn$, 
design matrix $\Xmat \in \Rnp$, 
target vector $\beta^* \in \Rp$,
and random noise $\eps \in \Rn$.
We allow for a dimension~$p$ that is of the same order or even much larger than the sample size~$n$,
and we assume a target vector~$\beta^*$ that is sparse.
A popular estimator of $\beta^*$ in this framework is the lasso \citep{Tibshirani96} 
\begin{equation}
\hat{\beta}_\lambda  \in \argmin\limits_{\beta \in \Rp} \biggl\{ \frac{1}{n} \normtwos{\Y - \Xmat \beta} + \lambda \normone{\beta} \biggr\},
\end{equation}
where $\lambda \in [0,\infty)$ is a tuning parameter.
The lasso estimator satisfies the well-known prediction bound 
\begin{equation}\label{eq:predictionbound}
\lambda \geq \frac{2 \normsup{\Xmat^\top \eps}}{n} \quad \Longrightarrow \quad \frac{1}{n} \normtwos{\Xmat (\beta^* - \hat{\beta}_\lambda)} \leq 2 \lambda \normone{\beta^*},
\end{equation}
which is an immediate consequence of the basic inequality for the lasso \citep[Lemma~6.1]{Buhlmann11} and H\"older's inequality.
This simple bound highlights that a crucial quantity in the analysis of the lasso estimator is $2\normsup{\Xmat^\top \eps}/n$. We call this quantity henceforth the \textit{effective noise}.

The effective noise does not only play a central role in the stated prediction bound 
but rather in almost all known finite-sample bounds for the lasso.
Such bounds, called oracle inequalities, are generally of the form \citep{Buhlmann11,Giraud14,Hastie15}
\begin{equation}\label{eq:oracleinequality}
\lambda \geq (1+\delta) \, \frac{2\normsup{\Xmat^\top \eps}}{n} \quad \Longrightarrow \quad \norm{\beta^* - \hat{\beta}_\lambda} \leq \kappa \lambda
\end{equation}
with some constant $\delta \in [0,\infty)$, 
a factor $\kappa = \kappa(\beta^*)$ that may depend on $\beta^*$, 
and a (pseudo-)norm $\norm{\cdot}$. 
Oracle inequalities of the form \eqref{eq:oracleinequality} are closely related to tuning parameter calibration for the lasso:
they suggest to control the loss $L(\beta^*, \hat{\beta}_\lambda) = \norm{\beta^* - \hat{\beta}_\lambda}$ of the lasso estimator~$\hat{\beta}_\lambda$ 
by taking the smallest tuning parameter~$\lambda$ for which the bound $\norm{\beta^* - \hat{\beta}_\lambda} \leq \kappa \lambda$ holds 
with probability at least $1-\alpha$ for some given $\alpha \in (0,1)$.
Denoting the $(1-\alpha)$-quantile of the effective noise $2\normsup{\Xmat^\top \eps}/n$ by $\lambda_\alpha^*$,
we immediately derive from  the oracle inequality \eqref{eq:oracleinequality} that
\begin{equation}\label{eq:intro:finitesampleguarantee:oracle}
\pr \Bigl( \norm{\beta^* - \hat{\beta}_{(1+\delta)\lambda}} \leq \kappa (1+\delta) \lambda \Bigr) \ge 1 -\alpha
\end{equation}
for $\lambda \ge \lambda_\alpha^*$. 
Stated differently, $\lambda = (1+\delta)\lambda_\alpha^*$ is the smallest tuning parameter 
for which the oracle inequality \eqref{eq:oracleinequality} yields the finite-sample bound $\norm{\beta^* - \hat{\beta}_\lambda} \leq \kappa \lambda$ with probability at least $1-\alpha$. 
Importantly, the tuning parameter choice $\lambda = (1+\delta)\lambda_\alpha^*$ is not feasible in practice,
since the quantile $\lambda_\alpha^*$ of the effective noise is not observed. 
An immediate question is, therefore, whether the quantile~$\lambda_\alpha^*$ can be estimated.

The effective noise is also closely related to high-dimensional inference. 
To give an example, we consider testing the null hypothesis $H_0: \beta^* = 0$ against the alternative $H_1: \beta^* \ne 0$. 
Testing this hypothesis corresponds to an important question in practice: 
do the regressors in the model $\Y = \Xmat \beta^* + \eps$ have any effect on the response at all? 
A test statistic for the hypothesis $H_0$ is given by $T = 2 \normsup{\Xmat^\top \Y}/n$. 
Under $H_0$, it holds that $T = 2\normsup{\Xmat^\top \eps}/n$, that is, $T$ is the effective noise. 
A test based on the statistic $T$ can thus be defined as follows: 
reject $H_0$ at the significance level $\alpha$ if $T > \lambda_\alpha^*$.
Since the quantile $\lambda_\alpha^*$ is not observed, this test is not feasible in practice,
which brings us back to the question of whether the quantile~$\lambda_\alpha^*$ can be estimated.

In this paper, we devise a novel estimator of the quantile $\lambda_\alpha^*$ based on bootstrap.
Besides the level $\alpha \in (0,1)$, it does not depend on any free parameters,
which means that it is fully data-driven.
The estimator can be used to approach a number of statistical problems in the context of the lasso. 
We focus on two such problems: 
(i)~tuning parameter calibration for the lasso and 
(ii)~inference on the parameter vector~$\beta^*$. 
The idea of using an estimator of the quantile $\lambda_\alpha^*$ to approach statistical issues such as (i) and (ii) is very natural and by no means new. \cite{Belloni13b}, for example, choose the tuning parameter of the lasso based on an estimator of $\lambda_\alpha^*$. Similar procedures for the square-root lasso and the Dantzig selector are considered in \cite{Belloni11} and \cite{Chernozhukov2013}, respectively. 
However, these methods are quite limited 
as they presume that either the noise distribution or a good initial guess for the lasso's tuning parameter is known.
Even though our estimator builds on ideas from the aforementioned papers, it is markedly different from the methods considered there and goes beyond them in important aspects. We discuss this in detail in Section \ref{sec:estimateeffnoise} after introducing our estimator.

We now briefly summarize the main contributions of our paper with regards to the two statistical problems (i) and (ii).
\begin{enumerate}[leftmargin=0.75cm, label=(\roman*), font=\normalfont]

\item \textit{Tuning parameter calibration for the lasso.} 
Our estimator~$\hat{\lambda}_\alpha$ of the quantile~$\lambda_\alpha^*$ can be used to calibrate the lasso with essentially optimal finite-sample guarantees. 
Specifically, we derive finite-sample statements of the form 
\begin{equation}\label{eq:intro:finitesampleguarantee:estimator}
\pr \Big( \norm{\beta^* - \hat{\beta}_{(1+\delta) \hat{\lambda}_\alpha}} \leq \kappa (1+\delta) \lambda_{\alpha-\nu_n}^* \Big) \ge 1-\alpha-\eta_n, 
\end{equation}
where $0 < \nu_n \le C n^{-K}$ and $0 < \eta_n \le C n^{-K}$ for some positive constants $C$ and $K$.
Statement \eqref{eq:intro:finitesampleguarantee:estimator} shows that calibrating the lasso with the estimator $\hat{\lambda}_\alpha$ 
yields almost the same finite-sample bound on the loss $L(\beta^*,\beta) = \norm{\beta^* - \beta}$ as calibrating it with the oracle parameter $\lambda_\alpha^*$. 
In particular, \eqref{eq:intro:finitesampleguarantee:estimator} is almost as sharp as the oracle bound 
$\pr (\norm{\beta^* - \hat{\beta}_{(1+\delta) \lambda_\alpha^*}} \leq \kappa (1+\delta) \lambda_\alpha^*) \ge 1-\alpha$, 
which is obtained by plugging $\lambda = \lambda_\alpha^*$ into \eqref{eq:intro:finitesampleguarantee:oracle}. \\
Finite-sample guarantees for the practical calibration of the lasso's tuning parameter are scarce. 
Exceptions include finite-sample bounds for Adaptive Validation~(AV) \citep{Chichignoud16} and Cross-Validation~(CV) \citep{Chetverikov16}.
One advantage of our approach via the effective noise is that it yields finite-sample guarantees not only for a specific loss 
but for any loss for which an oracle inequality of the type \eqref{eq:oracleinequality} is available. 
Another advantage is that it does not depend on secondary tuning parameters that are difficult to choose in practice; 
the only parameter it depends on is the level $1-\alpha$, which plays a similar role as the significance level of a test and, 
therefore, can be chosen in the same vein in practice.

\item \textit{Inference on the parameter vector $\beta^*$.} 
Our estimator $\hat{\lambda}_\alpha$ of the quantile $\lambda_\alpha^*$ can also be used to test hypotheses on the parameter vector $\beta^*$ in the model $\Y = \Xmat \beta^* + \eps$. 
Consider again the problem of testing $H_0: \beta^* = 0$ against $H_1: \beta^* \ne 0$. 
Our approach motivates the following test:
reject $H_0$ at the significance level $\alpha$ if $T > \hat{\lambda}_\alpha$.
We prove under mild regularity conditions that this test has the correct level $\alpha$ under $H_0$ and is consistent against alternatives that are not too close to~$H_0$.
Moreover, we show that the test can be generalized readily to more complex hypotheses. \\
High-dimensional inference based on the lasso has turned out to be a very difficult problem.
Some of the few advances that have been made in recent years include tests for the significance of small, fixed groups of para\-meters 
\citep{Belloni13,Zhang14,vdGeer14b,Javanmard14,Gold19},
tests for the significance of parameters entering the lasso path \citep{Lockhart14},
rates for confidence balls for the entire parameter vector (and infeasibility thereof)~\citep{Nickl13, Cai18},
and methods for inference after model selection \citep{Lee16,Tibshirani16}. 
In stark contrast to most other methods for high-dimensional inference, 
our tests are completely free of tuning parameters and, therefore, 
dispense with any fine-tuning (such as the calibration of 
multiple lasso tuning parameters in the first group of papers cited above).
\end{enumerate}

The paper is organized as follows. 
In Section \ref{sec:model}, we detail the modeling framework.
Our estimator of the quantiles of the effective noise is developed in Section~\ref{sec:estimateeffnoise}.
In Section~\ref{sec:statapp}, we apply the estimator to tuning parameter calibration and inference for the lasso.
Our theoretical analysis is complemented by a simulation study in Section \ref{sec:sim}, which investigates the finite-sample performance of our methods.

\section{Model setting}\label{sec:model}

We consider the standard linear model 
\begin{equation}\label{model}
\Y = \Xmat \beta^* + \eps, 
\end{equation}
where $\Y = (\Y_1,\ldots,\Y_n)^\top \in \Rn$ is the response vector, $\Xmat = (\Xvec_1,\ldots,\Xvec_n)^\top \in \Rnp$ is the design matrix with the vectors $\Xvec_i = (\Xvec_{i1},\ldots,\Xvec_{ip})^\top$, $\beta^* = (\beta^*_1,\ldots,\beta^*_p)^\top \in \Rp$ is the parameter vector, 
and $\eps = (\eps_1,\ldots,\eps_n)^\top \in \Rn$ is the noise vector. 
We are particularly interested in high-dimensional versions of the model, that is, $p\approx n$ or even $p\gg n$.
Throughout the paper, we assume the design matrix~$\Xmat$ to be random,
but our results carry over readily to fixed design matrices. 
We impose the following regularity conditions on the model~\eqref{model}: 
\begin{enumerate}[label=(C\arabic*),leftmargin=1cm]
\item \label{C1} The random variables $(\Xvec_i,\eps_i)$ are independent across $i$. 
\item \label{C2} The covariates $\Xvec_{ij}$ have bounded support, that is, $|\Xvec_{ij}| \le C_\Xvec$ for all $i$, $j$ and some sufficiently large constant $C_\Xvec < \infty$. 
Moreover, $n^{-1} \sum_{i=1}^n \ex[\Xvec_{ij}^2] \ge c_\Xvec^2$ for some constant $c_\Xvec > 0$. 
\item \label{C3} The noise variables $\eps_i$ are such that $\ex[\eps_i|\Xvec_i] = 0$ and $\ex[|\eps_i|^\theta] \le C_\theta < \infty$ for some $\theta > 4$ and all $i$. 
Moreover, the conditional noise variance $\sigma^2(X_i) = \ex[\eps_i^2|\Xvec_i]$ satisfies $0 < c_\sigma^2 \le \sigma^2(\cdot) \le C_\sigma^2 < \infty$ with some suitable constants $c_\sigma$ and $C_\sigma$. 
\item \label{C4} It holds that $p \le C_r n^r$, where $r > 0$ is an arbitrarily large but fixed constant and $C_r > 0$. 
\item \label{C5} There exist a constant $C_\beta<\infty$ and some small $\delta_\beta > 0$ such that $\normone{\beta^*} \le C_\beta n^{1/2-\delta_\beta}$.
\end{enumerate}

Condition~\ref{C1} stipulates independence across the observations, but the observations need not be identically distributed.  
The assumption about the boundedness of the covariates~$\Xvec_{ij}$ in~\ref{C2} makes the derivations more lucid 
but can be relaxed to sufficiently strong moment conditions on the variables~$\Xvec_{ij}$. 
Assumption \ref{C3} on the moments of the noise terms~$\eps_i$ is quite mild: only a bit more than the first four moments are required to exist. 
Condition~\ref{C4} on the relationship between $n$ and $p$ is mild as well: 
$p$ is allowed to grow as any polynomial of~$n$. 
Condition~\ref{C5} imposes sparsity on the parameter vector $\beta^*$ in an $\ell_1$-sense.
One could also replace it by a similar assumption in terms of the $\ell_0$-norm. However, 
an advantage of the $\ell_1$-version is that it allows for approximate sparsity 
-- see e.g.\ Section~3.2 in \cite{vdGeer13} or Section~2.8 in \cite{van2016estimation}.

\section{Estimating the effective noise}\label{sec:estimateeffnoise}

\subsection{Definition of the estimator}

Let $\lambda_\alpha^*$ be the $(1-\alpha)$-quantile of the effective noise $2\normsup{\Xmat^\top \eps}/n$, 
which is formally defined as $\lambda_\alpha^* = \inf \{q: \pr(2\normsup{\Xmat^\top \eps}/n \le q) \ge 1-\alpha\}$. 
We estimate $\lambda_\alpha^*$ as follows:
for any $\lambda$, let $\hat{\eps}_{\lambda} = \Y - \Xmat \hat{\beta}_{\lambda}$ be the residual vector that results from fitting the lasso with the tuning parameter $\lambda$, and let $e = (e_1,\ldots,e_n)^\top$ be a standard normal random vector independent of the data $(\Xmat,\Y)$. 
Define the criterion function 
\[ \criterion(\lambda,e) = \max_{1 \le j \le p} \Big| \frac{2}{n} \sum_{i=1}^n \Xvec_{ij} \hat{\eps}_{\lambda, i} e_i \Big|, \]
and let $\hat{q}_\alpha(\lambda)$ be the $(1-\alpha)$-quantile of $\criterion(\lambda,e)$ conditionally on $\Xmat$ and $\Y$. 
Formally, $\hat{q}_\alpha(\lambda) = \inf \{ q: \pr_e (\criterion(\lambda,e) \le q) \ge 1-\alpha \}$, where we use the shorthand $\pr_e( \, \cdot \, ) = \pr( \, \cdot \, | \, \Xmat,\Y)$. 
Our estimator of $\lambda_\alpha^*$ is defined as 
\begin{equation}\label{eq:lambdahat}
\hat{\lambda}_\alpha = \inf \big\{ \lambda > 0\ :\  \hat{q}_\alpha(\lambda^\prime) \le \lambda^\prime \text{ for all } \lambda^\prime \ge \lambda \big\}. 
\end{equation}

In practice, $\hat{\lambda}_\alpha$ can be computed by the following algorithm:
\begin{itemize}[leftmargin=1.5cm]
\item[Step 1:] For some large natural number $M$, specify a grid of points $0 < \lambda_1 < \ldots < \lambda_M = \overline{\lambda}$, where $\overline{\lambda} = 2 \|\Xmat^\top \Y\|_{\infty}/n$ is the smallest tuning parameter $\lambda$ for which $\hat{\beta}_{\lambda}$ equals zero. 
Simulate $L$ samples $e^{(1)},\ldots,e^{(L)}$ of the standard normal random vector $e$.
\item[Step 2:] For each grid point $1 \le m \le M$, compute the values of the criterion function $\{ \criterion(\lambda_m, e^{(\ell)}): 1 \le \ell \le L \}$ and calculate the empirical $(1-\alpha)$-quantile $\hat{q}_{\alpha,\text{emp}}(\lambda_m)$ from them.
\item[Step 3:] Approximate $\hat{\lambda}_\alpha$ by $\hat{\lambda}_{\alpha,\text{emp}} := \hat{q}_{\alpha,\text{emp}}(\lambda_{\hat{m}})$, where
$\hat{m} = \min \{ m : \hat{q}_{\alpha,\text{emp}}(\lambda_{m^\prime}) \le \lambda_{m^\prime} \text{ for all } {m^\prime} \ge m \}$
if $\hat{q}_{\alpha,\text{emp}}(\lambda_M) \le \lambda_M$ and $\hat{m} = M$ otherwise.
\end{itemize}
The values of~$M$ and~$L$ in this algorithm can be chosen large without excessive load on the computations:
the dependence of the computational complexity on~$M$ can be reduced by computing the lasso with warm starts along the tuning parameter path;
the influence of~$L$ can be reduced through basic parallelization.
Hence, the algorithm is computationally feasible even when $n$ and $p$ are very large.

\subsection{Heuristic idea of the estimator}\label{subsec:estimateeffnoise:heuristics}

Before analyzing the estimator $\hat{\lambda}_\alpha$ mathematically, we describe the heuristic idea behind it: 
for every $\lambda \in (0,\infty)$, the criterion function $\criterion(\lambda,e)$ can be regarded as a multiplier bootstrap version of the effective noise $2\normsup{\Xmat^\top \eps}/n = 2 \max_{1 \le j \le p} \linebreak |\sum_{i=1}^n \Xvec_{ij} \eps_i|/n$, where $e$ is the vector of bootstrap multipliers. 
Consequently, $\hat{q}_\alpha(\lambda)$ can be interpreted as a bootstrap estimator of the $(1-\alpha)$-quantile $\lambda_\alpha^*$ of the effective noise. 
Since the quality of the estimator $\hat{q}_\alpha(\lambda)$ hinges on the choice of $\lambda$,
the question is how to select an estimator $\hat{q}_\alpha(\lambda)$ from the family $\{ \hat{q}_\alpha(\lambda): \lambda > 0 \}$ that is a good approximation of $\lambda_\alpha^*$. 
Our selection rule \eqref{eq:lambdahat} is motivated by the following two heuristic claims which are justified below: with high probability, it holds that
\begin{eqnarray}
\hat{q}_\alpha(\lambda) \approx \lambda & \text{for} & \lambda \in [\lambda_\alpha^* - \delta, \lambda_\alpha^* + \delta] \label{eq:heuristics1} \\
\hat{q}_\alpha(\lambda) < \lambda & \text{for} & \lambda > \lambda_{\alpha}^* + \delta \label{eq:heuristics2}
\end{eqnarray}
with some small $\delta > 0$. Equation \eqref{eq:heuristics1} suggests that the function $\lambda \mapsto \hat{q}_\alpha(\lambda)$ has a fixed point near $\lambda_\alpha^*$, whereas equation \eqref{eq:heuristics2} tells us that there should not be any fixed point for values $\lambda > \lambda_\alpha^* + \delta$. Taken together, \eqref{eq:heuristics1} and \eqref{eq:heuristics2} suggest approximating $\lambda_\alpha^*$ by solving the fixed point equation $\hat{q}_\alpha(\lambda) = \lambda$ and picking the largest such fixed point $\lambda = \hat{\lambda}_\alpha$. 
This is the heuristic idea which underlies the formal definition of our estimator $\hat{\lambda}_\alpha$ in \eqref{eq:lambdahat}. 
A graphical illustration is provided in Figure \ref{fig:estimator}.
\medskip

\begin{figure}[t!]
\phantom{Upper margin}
\vspace{-0.7cm}

\centering
\includegraphics[width=0.45\textwidth]{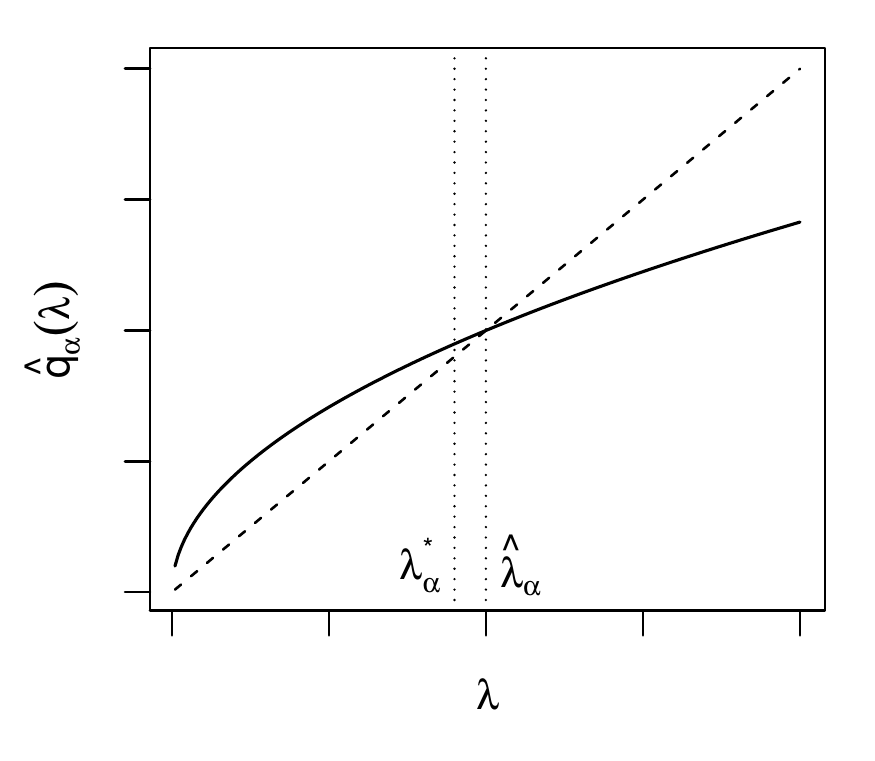}
\vspace{-0.25cm}

\caption{Graphical illustration of the estimator $\hat{\lambda}_\alpha$. 
The solid black line is the function $\lambda \mapsto \hat{q}_\alpha(\lambda)$, the dashed line is the 45-degree line, and the two vertical dotted lines indicate the values of $\lambda_\alpha^*$ and $\hat{\lambda}_\alpha$, respectively.} \label{fig:estimator} 
\end{figure}

\textit{Discussion of the heuristic claim \eqref{eq:heuristics1}.} 
To start with, we bound the criterion function $\criterion(\lambda,e)$ from below and above by 
\[ Q(e) - r_\lambda(e) \le \criterion(\lambda,e) \le Q(e) + r_\lambda(e), \]
where
\begin{align*}
Q(e) & = \max_{1 \le j \le p} \Big| \frac{2}{n} \sum_{i=1}^n \Xvec_{ij} \eps_i e_i \Big| \\
r_{\lambda}(e) & = \max_{1 \le j \le p} \Big| \frac{2}{n} \sum_{i=1}^n \Xvec_{ij} \Xvec_i^\top (\beta^* - \hat{\beta}_\lambda) e_i \Big|. 
\end{align*}
Here, $Q(e)$ is a multiplier bootstrap version of the effective noise which is based on the true noise terms $\eps$ rather than the residuals $\hat{\eps}_\lambda$. The term $r_{\lambda}(e)$ is a remainder that captures the estimation error $\hat{\eps}_\lambda-\eps$ produced by the lasso $\hat{\beta}_\lambda$. Let $\lambda_\alpha$ be the $(1-\alpha)$-quantile of $Q(e)$ conditionally on $\Xmat$ and $\Y$. Theory for the multiplier bootstrap in high dimensions \citep{Chernozhukov2013} suggests that the quantile $\lambda_\alpha$ gives a good approximation to $\lambda_\alpha^*$. If the remainder $r_{\lambda}(e)$ tends to be small for a certain choice of $\lambda$, then the criterion function $\criterion(\lambda,e)$ tends to be close to $Q(e)$, which in turn suggests that the quantile $\hat{q}_\alpha(\lambda)$ is close to $\lambda_\alpha$. Since $\lambda_\alpha$ gives a good approximation to $\lambda_\alpha^*$, we expect $\hat{q}_\alpha(\lambda)$ to be an accurate estimate of $\lambda_\alpha^*$ as well.

Standard prediction bounds for the lasso suggest that the tuning parameter choice $\lambda = \lambda_\alpha^*$ produces a precise model fit $\Xmat \hat{\beta}_{\lambda_\alpha^*}$. 
The prediction bound~\eqref{eq:predictionbound}, for example, implies that with pro\-bability at least $1-\alpha$, we have $\normtwos{\Xmat (\beta^* - \hat{\beta}_{\lambda_\alpha^*})}/n \le 2 \lambda_\alpha^* \normone{\beta^*}$, where $ 2 \lambda_\alpha^* \normone{\beta^*} = O (\normone{\beta^*} \sqrt{\log(p) / n}) = o(1)$ under our technical conditions. 
Hence, we expect the remainder term $r_{\lambda_\alpha^*}(e)$ to be small. 
%Hence, we expect the residual vector $\hat{\eps}_{\lambda_\alpha^*} = \Y - \Xmat \hat{\beta}_{\lambda_\alpha^*}$ to be a good proxy of the true noise vector $\eps$. This suggests that the remainder term $r_{\lambda_\alpha^*}(e)$ is small. 
From the considerations in the previous paragraph, it follows that $\hat{q}_\alpha(\lambda_\alpha^*)$ should be a suitable estimate of $\lambda_\alpha^*$, that is, $\hat{q}_\alpha(\lambda_\alpha^*) \approx \lambda_\alpha^*$. Since $\hat{q}_\alpha(\lambda) \approx \hat{q}_\alpha(\lambda_\alpha^*)$ for values of $\lambda$ close to $\lambda_\alpha^*$ (which is due to the continuity of the solution path of the lasso), we further expect that  
\begin{equation*}
\hat{q}_\alpha(\lambda) \approx \lambda \quad \text{for} \quad \lambda \in [\lambda_\alpha^* - \delta, \lambda_\alpha^* + \delta]
\end{equation*}
with some small $\delta > 0$, which is the heuristic claim \eqref{eq:heuristics1}.  
\medskip

\textit{Discussion of the heuristic claim \eqref{eq:heuristics2}.} 
As we gradually increase $\lambda$ from $\lambda_\alpha^*$ to larger values, the lasso estimator $\hat{\beta}_{\lambda}$ tends to become more biased towards zero, 
implying that the residual vector $\hat{\eps}_\lambda$ gets a less accurate proxy of the noise vector $\eps$.  
As a consequence, we expect the remainder term $r_\lambda(e)$ and thus the criterion function $\criterion(\lambda,e)$ to increase as $\lambda$ gets larger. This in turn suggests that the quantile $\hat{q}_\alpha(\lambda)$ gets larger with increasing $\lambda$, thus overestimating $\lambda_\alpha^*$ more and more strongly. 
On the other hand, one can formally prove that the remainder $r_\lambda(e)$ grows quite slowly with $\lambda$. In particular, one can show that with high probability, $r_{\lambda}(e) \le C \{ (\log n)^2 /  n^{1/4} \} \sqrt{\lambda}$ for all $\lambda \ge \lambda_\alpha^*$. A formalized version of this statement is given in Lemma \ref{lemmaA2} in the Appendix. Since $\criterion(\lambda,e) \le Q(e) + r_\lambda(e)$, this implies that the criterion function $\criterion(\lambda,e)$ and thus its $(1-\alpha)$-quantile $\hat{q}_\alpha(\lambda)$ grow fairly slowly with increasing $\lambda$. In particular, we expect $\hat{q}_\alpha(\lambda)$ to grow more slowly than $\lambda$, that is, 
\begin{equation*}
\hat{q}_\alpha(\lambda) < \lambda \quad \text{for} \quad \lambda > \lambda_{\alpha}^* + \delta
\end{equation*}
with some small $\delta > 0$. This is the heuristic claim \eqref{eq:heuristics2}.

\subsection{Theoretical analysis of the estimator}

We now analyze the theoretical properties of the estimator $\hat{\lambda}_\alpha$. 
To do so, we use the following notation. 
By $C_1$, $K_1$, $C_2$ and $K_2$, we denote positive real constants that depend only on the set of model parameters $\Theta = \{ c_\Xvec, C_\Xvec, c_\sigma, C_\sigma, C_\theta, \theta, C_r, r, C_\beta, \delta_\beta \}$ defined in \ref{C1}--\ref{C5}. 
The constants $C_1$, $K_1$, $C_2$ and $K_2$ are thus in particular independent of the sample size $n$ and the dimension $p$. 
Moreover, we let 
\[ \mathcal{T}_\lambda = \Big\{ \frac{2}{n} \normsup{\Xmat^\top \eps} \le \lambda \Big\} \]
be the event that the effective noise $2\normsup{\Xmat^\top \eps}/n$ is smaller than $\lambda$. 
The following theorem, which is the main result of the paper, formally relates the estimator $\hat{\lambda}_{\alpha}$ to the quantiles of the effective noise. 
In the sequel, we will use this theorem to derive results on optimal tuning parameter choice and inference for the lasso.

\pagebreak

\begin{theorem}\label{theo1}
Let \ref{C1}--\ref{C5} be satisfied. 
There exist an event $\mathcal{A}_n$ with $\pr(\mathcal{A}_n) \ge 1 - C_1 n^{-K_1}$ for some positive constants $C_1$ and $K_1$ and a sequence of real numbers $\nu_n$ with $0 < \nu_n \le C_2 n^{-K_2}$ for some positive constants $C_2$ and $K_2$ such that the following holds:
on the event $\mathcal{T}_{\lambda_{\alpha+\nu_n}^*} \cap \mathcal{A}_n$, 
\[ \lambda_{\alpha + \nu_n}^* \le \hat{\lambda}_{\alpha} \le \lambda_{\alpha-\nu_n}^* \]
for every $\alpha \in (a_n, 1-a_n)$ with $a_n = 2 \nu_n + (n \lor p)^{-1}$. 
\end{theorem}
\noindent The proof of Theorem \ref{theo1} is given in the Appendix. 
Precise definitions of $\mathcal{A}_n$ and $\nu_n$ are provided in equations \eqref{eq:def:eventA} and \eqref{eq:nu}, respectively. 
It is important to note that the bounds $\lambda_{\alpha + \nu_n}^*$ and $\lambda_{\alpha-\nu_n}^*$ in Theorem \ref{theo1} are design-specific, that is, they depend on the distribution of the design vectors $\Xvec_i$ (as well as on the distribution of the noise variables $\eps_i$). Among other things, the bounds tend to get smaller as the design gets more correlated, that is, as the correlation between the covariates $\Xvec_{ij}$ increases.

Since $\mathbb{P}(\mathcal{T}_{\lambda_{\alpha+\nu_n}^*} \cap \mathcal{A}_n) \ge 1 - \alpha - Cn^{-K}$ with some constants $C$ and $K$ that only depend on the model parameters $\Theta$, Theorem \ref{theo1} immediately implies that
\[ \pr \big( \lambda_{\alpha + \nu_n}^* \le \hat{\lambda}_{\alpha} \le \lambda_{\alpha-\nu_n}^* \big) \ge 1-\alpha - Cn^{-K}. \]
Hence, with probability at least $1-\alpha-Cn^{-K} = 1-\alpha-o(1)$, our estimator $\hat{\lambda}_{\alpha}$ gives a good approximation to $\lambda_\alpha^*$ in the sense that $\lambda_{\alpha + \nu_n}^* \le \hat{\lambda}_{\alpha} \le \lambda_{\alpha-\nu_n}^*$. 
Another immediate consequence of Theorem \ref{theo1} is that $|\hat{\lambda}_\alpha - \lambda_\alpha^*| \le \lambda_{\alpha-\nu_n}^* - \lambda_{\alpha+\nu_n}^*$ on the event $\mathcal{T}_{\lambda_{\alpha+\nu_n}^*} \cap \mathcal{A}_n$. From this, we obtain the deviation inequality
\begin{equation*}%\label{eq:dev-inequ-1}
\mathbb{P}\big( |\hat{\lambda}_\alpha - \lambda_\alpha^*| \le \rho_{n,\Xmat,\eps,\alpha} \big) \ge 1 - \alpha - Cn^{-K}, 
\end{equation*} 
where $\rho_{n,\Xmat,\eps,\alpha} = \lambda_{\alpha-\nu_n}^* - \lambda_{\alpha+\nu_n}^*$. 
With the help of Proposition \ref{propA1} from the Appendix, it is further possible to replace the bound $\rho_{n,\Xmat,\eps,\alpha}$ by some kind of Gaussian version: Let $\Wgauss = (\Wgauss_1,\ldots,\Wgauss_p)^\top$ be a Gaussian random vector with $\ex[\Wgauss_j] = 0$ for all $j$ and the covariances 
\begin{equation*}
\ex[\Wgauss_j \Wgauss_k] = \ex \left[ \Big( \frac{1}{\sqrt{n}} \sum_{i=1}^n \Xvec_{ij} \eps_i \Big) \Big( \frac{1}{\sqrt{n}} \sum_{i=1}^n \Xvec_{ik} \eps_i \Big) \right] 
\end{equation*}
for $1 \le j \le k \le p$. Moreover, let $\gamma_\alpha^\Wgauss$ be the $(1-\alpha)$-quantile of $\max_{1 \le j \le p} |\Wgauss_j|$. Then 
\begin{equation*}%\label{eq:dev-ineq-2}
\mathbb{P}\big( |\hat{\lambda}_\alpha - \lambda_\alpha^*| \le \rho_{n,\boldsymbol{X},\varepsilon,\alpha}^{\Wgauss} \big) \ge 1 - \alpha - Cn^{-K} 
\end{equation*} 
with $\rho_{n,\boldsymbol{X},\varepsilon,\alpha}^{\Wgauss} = 2 (\gamma_{\alpha-2\nu_n}^{\Wgauss} - \gamma_{\alpha+2\nu_n}^{\Wgauss}) / \sqrt{n}$.

\begin{remark}
An interesting question is the following: how large is the distance $\gamma_{\alpha-\nu_n}^{\Wgauss} - \gamma_{\alpha+\nu_n}^{\Wgauss}$ and thus the bound $\rho_{n,\Xmat,\eps,\alpha}^{\Wgauss}$? By definition, $\gamma_{\alpha}^{\Wgauss}$ is the $(1-\alpha)$-quantile of the maximum $\max_{1 \le j \le p} |\Wgauss_j|$ of $p$ Gaussian random variables with a general, potentially very complicated covariance structure. It is highly non-trivial to characterize the distribution of maxima of Gaussian random variables with a general correlation structure. Hence, finding precise bounds on the quantiles $\gamma_{\alpha}^{\Wgauss}$ (and thus their distance) is a hard problem in general. In some special cases, however, it is possible to obtain suitable bounds. 
Consider in particular the situation that the variables $\Wgauss_j$ are i.i.d., which occurs for example when (i) the design variables $\Xvec_{ij}$ are normalized such that $\ex[\Xvec_{ij}] = 0$ and $\ex[\Xvec_{ij}^2] = 1$ for all $i$ and $j$, (ii) the design is uncorrelated (i.e., $\ex[\Xvec_{ij} \Xvec_{ik}] = 0$ for all $j \ne k$), and (iii) the noise $\eps_i$ is homoskedastic (i.e., $\sigma^2(\Xvec_i) \equiv \text{const.}$). In this case, one can show that 
\begin{equation}\label{bound-EVT}
\gamma_{\alpha-2\nu_n}^{\Wgauss} - \gamma_{\alpha+2\nu_n}^{\Wgauss} \le \frac{C}{\sqrt{\log p}}
\end{equation}
via classic extreme value theory, where $C$ is a sufficiently large constant independent of $n$ and $p$. A brief sketch of the proof is included in the Supplementary Material for completeness. 
\end{remark}

\subsection{Relationship of the estimator to existing methods}

Roughly speaking, existing methods for estimating the quantiles of the effective noise fall into two categories:
\begin{enumerate}[label=(\Alph*),leftmargin=0.8cm]

\item If the distribution of the error vector $\eps$ is known, it is trivial to construct an approximation of $\lambda_\alpha^*$. To fix ideas, let $\eps \sim N(0,\sigma^2 \idmat)$ with known variance parameter $\sigma^2$ and consider a fixed design $\Xmat$ for simplicity. (Otherwise, assume that $\eps$ is independent of $\Xmat$ and condition on the latter.) In this case, the distribution of the effective noise $2\normsup{\Xmat^\top \eps}/n$ and thus its $(1-\alpha)$-quantile $\lambda_\alpha^*$ is known and can be computed by Monte Carlo simulations in practice. If the distribution of $\eps$ is not known exactly but is known to belong to a small family of distributions $\mathcal{F}$, it is further possible to compute (finite sample and asymptotic) upper bounds on $\lambda_\alpha^*$ under certain conditions as detailed in \cite{Belloni11}.

\item In the more interesting situation where the distribution of $\eps$ is unknown, a simple way to estimate $\lambda_\alpha^*$ is as follows: Let $\lambda^{[0]}$ be a preliminary choice of the lasso's tuning parameter. Plug $\lambda^{[0]}$ into $\hat{q}_\alpha(\cdot)$ and use the resulting value $\hat{q}_\alpha(\lambda^{[0]})$ as an estimator of $\lambda_\alpha^*$. -- This plug-in approach is not very satisfactory: The quality of the estimator $\hat{q}_\alpha(\lambda^{[0]})$ obviously hinges on the precise choice of $\lambda^{[0]}$. In particular, as discussed in our heuristic considerations of Section \ref{subsec:estimateeffnoise:heuristics}, we can expect the following: If $\lambda^{[0]}$ is close to $\lambda_\alpha^*$, then the estimator $\hat{q}_\alpha(\lambda^{[0]})$ will tend to be close to $\lambda_\alpha^*$ as well. In contrast, if $\lambda^{[0]}$ happens to be far away from $\lambda_\alpha^*$, then $\hat{q}_\alpha(\lambda^{[0]})$ may also be far off. Hence, simply plugging a preliminary choice $\lambda^{[0]}$ into $\hat{q}_\alpha(\cdot)$ does not solve the problem of estimating $\lambda_\alpha^*$ but merely shifts it to the choice of $\lambda^{[0]}$: if we want to make sure that $\hat{q}_\alpha(\lambda^{[0]})$ is a good estimator of $\lambda_\alpha^*$, we need to make sure that the same holds for the preliminary estimator $\lambda^{[0]}$.

\end{enumerate}
Estimators of category (A) are for example considered in \cite{Belloni11} and \cite{Belloni13b}, estimators of category (B) can be found in \cite{Chernozhukov2013} (in the context of the Dantzig selector rather than the lasso). It is important to emphasize that the problem of estimating the quantile $\lambda_\alpha^*$ is not the focus but only a very minor aspect of the aforementioned papers. This is presumably the reason why only the two simple approaches (A) and (B) have been considered there. Indeed, we are not aware of any article whose main focus is the estimation of the quantiles of the effective noise.

One way to improve on the plug-in method from (B) is to iterate it: Given some starting value $\lambda^{[0]}$, one computes the update $\lambda^{[r]} = \hat{q}_\alpha(\lambda^{[r-1]})$ for $r =1,2,\ldots$ until some convergence criterion is satisfied. The idea behind this iterative procedure is to find a fixed point $\lambda = \hat{q}_\alpha(\lambda)$ of the function $\hat{q}_\alpha(\cdot)$. Hence, it relies on the same heuristic as our method. The main contribution of our paper is (i) to devise an estimation approach which formalizes the fixed point heuristic and (ii) to derive finite sample theory for it. An important practical advantage of our fixed point method over the plug-in method from (B) is that it is free of tuning parameters: unlike the plug-in method, it does not require a preliminary estimator $\lambda^{[0]}$. Its only free parameter is the value $\alpha \in (0,1)$.

\section{Statistical applications}\label{sec:statapp}

\subsection{Tuning parameter choice}\label{subsec:tuning}

A major challenge when implementing the lasso estimator $\hat{\beta}_\lambda$ is to choose the regularization parameter~$\lambda$. 
As already discussed in the Introduction, the lasso satisfies the prediction bound \eqref{eq:predictionbound}, which can be rephrased as follows: 
\begin{equation}\label{eq:predictionbound'}
\text{On the event } \mathcal{T}_\lambda, \normtwos{\Xmat (\beta^* - \hat{\beta}_{\lambda^\prime})}/n \leq 2 \lambda^\prime \normone{\beta^*} \text{ for every } \lambda^\prime \ge \lambda. 
\end{equation}
To control the prediction error, we would like to choose the smallest tuning parameter $\lambda$ such that the bound $\normtwos{\Xmat (\beta^* - \hat{\beta}_\lambda)}/n \leq 2 \lambda \normone{\beta^*}$ holds with high probability. 
Formally speaking, we may consider
\[ \lambda^{\textnormal{oracle}}_\alpha = \inf \{ \lambda > 0: \pr(\mathcal{T}_\lambda) \ge 1 - \alpha \} \] 
with some $\alpha \in (0,1)$ as the optimal tuning parameter. 
We call $\lambda^{\textnormal{oracle}}_\alpha$ the oracle tuning parameter.
It immediately follows from~\eqref{eq:predictionbound'} that for every $\lambda \ge \lambda_\alpha^{\text{oracle}}$, 
\[ \pr \biggl( \frac{1}{n} \normtwos{\Xmat (\beta^* - \hat{\beta}_\lambda)} \leq 2 \lambda \normone{\beta^*} \biggl) \ge 1 - \alpha, \]
whereas this probability bound is not guaranteed for any other $\lambda < \lambda_\alpha^{\text{oracle}}$. 
Consequently, $\lambda_\alpha^{\text{oracle}}$ is the smallest tuning parameter for which the prediction bound~\eqref{eq:predictionbound'} yields the finite-sample guarantee 
\begin{equation}\label{eq:predictionoracle}
\frac{1}{n} \normtwos{\Xmat (\beta^* - \hat{\beta}_{\lambda^{\textnormal{oracle}}_\alpha})} \leq 2 \lambda^{\textnormal{oracle}}_\alpha \normone{\beta^*} 
\end{equation}
with probability at least $1-\alpha$. 
Importantly, the oracle tuning parameter $\lambda_\alpha^{\text{oracle}}$ is nothing else than the $(1-\alpha)$-quantile $\lambda_\alpha^*$ of the effective noise, that is, $\lambda^{\textnormal{oracle}}_\alpha = \lambda_\alpha^*$ for every $\alpha \in (0,1)$. 
Our estimator $\hat{\lambda}_\alpha$ can thus be interpreted as an approximation of the oracle parameter $\lambda^{\textnormal{oracle}}_{\alpha}$. 
With the help of Theorem \ref{theo1}, we can show that implementing $\hat{\beta}_\lambda$ with the estimator $\lambda = \hat{\lambda}_\alpha$ produces almost the same finite-sample guarantee as \eqref{eq:predictionoracle}.
\begin{prop}\label{prop:tuning:1}
Let the conditions of Theorem \ref{theo1} be satisfied. 
With probability $\ge 1 - \alpha - \nu_n - C_1 n^{-K_1} = 1 - \alpha + o(1)$, it holds that
\begin{equation*}
\frac{1}{n} \normtwos{\Xmat (\beta^* - \hat{\beta}_{\hat{\lambda}_\alpha})} \leq 2 \lambda^{\textnormal{oracle}}_{\alpha-\nu_n} \normone{\beta^*}.
\end{equation*}
\end{prop} 
\noindent For completeness, a short proof is provided in the Appendix. 
The upper bound  $2 \lambda^{\textnormal{oracle}}_{\alpha-\nu_n} \normone{\beta^*}$ in Proposition \ref{prop:tuning:1} is almost as sharp as the bound $2 \lambda^{\textnormal{oracle}}_\alpha \normone{\beta^*}$ in \eqref{eq:predictionoracle};
the only difference is that the $(1-\alpha)$-quantile $\lambda^{\textnormal{oracle}}_\alpha$ is replaced by the somewhat larger $(1-\{\alpha-\nu_n\})$-quantile $\lambda^{\textnormal{oracle}}_{\alpha-\nu_n}$. 
There are improved versions of the prediction bound~\eqref{eq:predictionbound} \citep{Lederer19} as well as other types of prediction bounds \citep{Dalalyan17,Hebiri12,vdGeer13} that can be treated in the same way.

Our method does not only allow us to obtain finite-sample bounds on the prediction loss. 
It can also be used to equip the lasso with finite-sample guarantees for other losses. 
We consider the $\ell_\infty$-loss $L_\infty(\beta^*,\beta) = \normsup{\beta^* - \beta}$ as an example.
Analogous considerations apply to any other loss for which an oracle inequality of the form \eqref{eq:oracleinequality} is available, such as the $\ell_1$- and $\ell_2$-losses.
Let $\activeset = \{ j : \beta^*_j \ne 0 \}$ be the active set of $\beta^*$. 
Moreover, for any vector $v = (v_1,\ldots,v_p)^\top \in \Rp$, let $v_S = (v_j \ind(j \in S))_{j=1}^p$ and define $v_{S^\complement}$ analogously with $S^\complement = \{1\ldots,p\} \setminus S$. The design matrix $\Xmat$ is said to fulfill the $\ell_\infty$-restricted eigenvalue condition \citep{Chichignoud16} with the constants $\phi > 0$ and $\delta > 0$ if 
\begin{equation}\label{eq:RE}
\frac{\normsup{\Xmat^\top \Xmat v}}{n} \ge \phi \normsup{v} \quad \text{for all } v \in \mathbb{C}_\delta(S),
\end{equation}
where $\mathbb{C}_\delta(S)$ is the double cone
\[ \mathbb{C}_\delta(S) = \biggl\{ v \in \Rp: \normone{v_{S^\complement}} \le \frac{2+\delta}{\delta} \normone{v_S} \biggr\}. \] 
Under condition \eqref{eq:RE}, we obtain the following oracle inequality, whose proof is provided in the Supplementary Material.
\begin{lemma}\label{lemma:oracleinequality}
Suppose that $\Xmat$ satisfies the restricted eigenvalue condition \eqref{eq:RE}. 
On the event $\mathcal{T}_\lambda$, it holds that 
\begin{equation}
\normsup{\hat{\beta}_{\lambda^\prime} - \beta^*} \le \kappa \lambda^\prime
\end{equation}
for every $\lambda^\prime \ge (1+\delta) \lambda$ with $\kappa = 2/\phi$. 
\end{lemma}
\noindent Whereas this $\ell_\infty$-oracle inequality is valid under condition \eqref{eq:RE}, different conditions are needed to obtain oracle inequalities for other losses -- see \citet{vandeGeer09} for a discussion of different assumptions. In the $\ell_2$-loss case, for instance, an $\ell_2$-restricted eigenvalue condition is usually imposed, which is somewhat different (and less restrictive) than \eqref{eq:RE}. Moreover, in the prediction loss case considered above, no conditions on the design (in particular, no restricted eigenvalue conditions) are needed at all. Hence, condition~\eqref{eq:RE} is not an assumption imposed by our method, it is rather inflicted by the oracle inequality of the $\ell_\infty$-loss.

Let $\mathcal{B}_n$ be the event that $\Xmat$ satisfies the restricted eigenvalue condition \eqref{eq:RE} and note that $\pr(\mathcal{B}_n) \rightarrow 1$ for certain classes of random design matrices $\Xmat$ \citep{vdGeer14}. 
The oracle inequality of Lemma \ref{lemma:oracleinequality} can be rephrased as follows: 
on the event $\mathcal{T}_\lambda \cap \mathcal{B}_n$, it holds that $\normsup{\hat{\beta}_{\lambda^\prime} - \beta^*} \le \kappa \lambda^\prime$ for any $\lambda^\prime \ge (1+\delta) \lambda$. 
The oracle parameter $\lambda_\alpha^{\textnormal{oracle}}$ yields the following finite-sample guarantee: 
on the event $\mathcal{T}_{\lambda_\alpha^{\textnormal{oracle}}} \cap \mathcal{B}_n$, that is, with probability $\ge 1 - \alpha - P(\mathcal{B}_n^\complement)$, it holds that 
\begin{equation}\label{eq:guarantee:infty}
\normsup{\hat{\beta}_{(1+\delta) \lambda_{\alpha}^{\textnormal{oracle}}} - \beta^*} \le (1+\delta) \kappa \lambda_{\alpha}^{\textnormal{oracle}}.
\end{equation}
Theorem \ref{theo1} implies that we can approximately recover this finite-sample guarantee when replacing the oracle parameter $\lambda_\alpha^{\textnormal{oracle}}$ with the estimator $\hat{\lambda}_\alpha$.
\pagebreak
\begin{prop}\label{prop:tuning:2}
Let the conditions of Theorem \ref{theo1} be satisfied.
With probability $\ge 1 - \alpha - \pr(\mathcal{B}_n^\complement) - \nu_n - C_1 n^{-K_1} = 1 - \alpha - \pr(\mathcal{B}_n^\complement) + o(1)$, it holds that
\begin{equation*}
\normsup{\hat{\beta}_{(1+\delta) \hat{\lambda}_{\alpha}} - \beta^*} \le (1+\delta) \kappa \lambda_{\alpha-\nu_n}^{\textnormal{oracle}}.
\end{equation*}
\end{prop} 
\noindent A proof of Proposition \ref{prop:tuning:2} can be found in the Appendix. 
It is important to note that the $\ell_\infty$-bound of Proposition \ref{prop:tuning:2} entails finite-sample guarantees for variable selection. 
Specifically, it implies that with probability $\ge 1 - \alpha - \pr(\mathcal{B}_n^\complement)+ o(1)$, the lasso estimator $\hat{\beta}_{(1+\delta)\hat{\lambda}_{\alpha}}$ recovers all non-zero components of $\beta^*$ that are larger in absolute value than $(1+\delta) \kappa \lambda_{\alpha-\nu_n}^{\textnormal{oracle}}$. 
From Lemma \ref{lemmaA3} and Proposition \ref{propA1} in the Appendix, it follows that $\lambda_{\alpha-\nu_n}^{\textnormal{oracle}} \le C \sqrt{\log(n \lor p)/n}$ with some sufficiently large constant $C$. 
Hence, with probability $\ge 1 - \alpha - \pr(\mathcal{B}_n^\complement) + o(1)$, the lasso estimator $\hat{\beta}_{(1+\delta)\hat{\lambda}_{\alpha}}$ in particular recovers all non-zero entries of $\beta^*$ that are of larger order than $O(\sqrt{\log(n \lor p)/n})$.

%\begin{remark}
%Recent studies refine the theoretical underpinning of the lasso's tuning parameter when predictors are correlated.
%For example, \citet{Dalalyan17} (see especially their proof of Theorem~4) show that the oracle tuning parameter can be based on the ``reduced'' effective noise
%${2\normsup{(\mathcal{P} \Xmat)^\top \eps}}/{n}$, where 
%$\mathcal{P} = \idmat - \Xmat_T (\Xmat_T^\top \Xmat_T)^{\dagger} \Xmat_T^\top$ is the projection on the columns of the design matrix that have indexes in an appropriately chosen set $T\subset\{1,\dots,p\}$ (the dagger ${}^\dagger$ denotes a pseudoinverse),
%rather than on the ``full'' effective noise ${2\normsup{\Xmat^\top \eps}}/{n}$,
%and they show that the reduced effective noise is considerably smaller than the full effective noise when the predictors are highly correlated.
%We do not go into details here to avoid digression,
%but we point out that the reduced effective noise can be estimated the same way (and with the same guarantees) as the full effective noise---see the following section.
%This observation highlights, once more, that our approach can cater to the specifics of the data at hand.
%\end{remark}

\subsection{Inference for the lasso}\label{subsec:inference}

Inference for the lasso is a notoriously difficult problem: the distribution of the lasso has a complicated limit and is hardly useful for statistical inference \citep{Knight00, Leeb05}. 
For this reason, inferential methods for the lasso are quite rare. 
Some exceptions are 
tests for the significance of small, fixed groups of para\-meters \citep{Belloni13,Zhang14,vdGeer14b,Javanmard14,Gold19}
and tests for the significance of parameters entering the lasso path \citep{Lockhart14}.
%rates for confidence balls for the entire parameter vector (and infeasibility thereof)~\citep{Nickl13, Cai18},
%and methods for inference after model selection~\citep{Belloni13,Tibshirani16}. 
In what follows, we show that our method enables us to construct tuning-parameter-free tests for certain hypotheses of interest.

We first consider  testing the null hypothesis $H_0: \beta^* = 0$ against the alternative $H_1: \beta^* \ne 0$, which was briefly discussed in the Introduction. 
Our test statistic of~$H_0$ is defined as 
\[ T = \frac{2\normsup{\Xmat^\top \Y}}{n}, \]
which implies that $T = 2 \normsup{\Xmat^\top \eps}/n$ under $H_0$, that is, $T$ is  the effective noise under $H_0$. 
This observation suggests to define a test of $H_0$ as follows: 
reject $H_0$ at the significance level $\alpha$ if $T > \hat{\lambda}_\alpha$, where $\hat{\lambda}_\alpha$ estimates the $(1-\alpha)$-quantile $\lambda_\alpha^*$ of $T$ under $H_0$. 
This test has the following theoretical properties.
\begin{prop}\label{prop:testing:1}
Let the conditions of Theorem \ref{theo1} be satisfied. 
Under the null hypothesis $H_0: \beta^* = 0$, 
it holds that
\[  \pr (T \le \hat{\lambda}_\alpha) \ge 1 - \alpha + o(1). \]
Moreover, under any alternative $\beta^* \ne 0$ that satisfies the condition $\pr(\normsup{\Xmat^\top\Xmat\beta^*}/n \linebreak \ge c \sqrt{\log(n \lor p)/n}) \rightarrow 1$ for every fixed $c > 0$, 
it holds that
\[ \pr ( T > \hat{\lambda}_\alpha) =  1 - o(1). \]
\end{prop}
\noindent The proof is deferred to the Appendix.
Proposition \ref{prop:testing:1} ensures that the proposed test is  of level $\alpha$ asymptotically and has asymptotic power $1$ against any alternative $\beta^* \ne 0$ that satisfies the condition $\pr(\normsup{\Xmat^\top\Xmat\beta^*}/n \ge c \sqrt{\log(n \lor p)/n}) \linebreak \rightarrow 1$ for every $c > 0$. 
Such a condition is inevitable:
in the model $\Y = \Xmat \beta^* + \eps$, it is not possible to distinguish between vectors $\beta^* \ne 0$ that satisfy $\Xmat \beta^* = 0$ and the null vector. 
Hence, a test can only have power against alternatives $\beta^* \ne 0$ that satisfy $\Xmat \beta^* \ne 0$, that is, against alternatives $\beta^* \ne 0$ that do not lie in the kernel $\text{Ker}(\Xmat)=\text{Ker}(\Xmat^\top \Xmat/n)$ of the linear mapping $\Xmat$. 
By imposing the condition $\pr(\normsup{\Xmat^\top\Xmat\beta^*}/n \ge c \sqrt{\log(n \lor p)/n}) \rightarrow 1$, we restrict attention to alternatives $\beta^* \ne 0$ that have enough signal outside the kernel of $\Xmat$.

We now generalize the discussed test procedure in a way that allows to handle more complex hypotheses. 
Specifically, we generalize it such that a low-dimensional linear model can be tested against a high-dimensional alternative. 
To do so, we partition the design matrix $\Xmat$ into two parts according to $\Xmat = (\Xmat_A, \Xmat_B)$, where $A \, \dot{\cup} \, B = \{1,\ldots,p\}$, 
$\Xmat_A$ is the part of the design matrix that contains the observations on the regressors in the set $A$, 
and $\Xmat_B$ contains the observations on the regressors in the set $B$. 
We also partition the parameter vector $\beta^*$ accordingly into two parts $\beta_A^* \in \R^{|A|}$ and $\beta_B^* \in \R^{|B|}$ such that $\beta^* = ( (\beta_A^*)^\top, (\beta_B^*)^\top )^\top$. 
The linear model \eqref{model} can then be written as
\begin{equation}\label{eq:partitionedmodel}
\Y = \Xmat_A \beta_A^* + \Xmat_B \beta_B^* + \eps. 
\end{equation} 
In practice, regression is often based on simple, low-dimensional models of the form $\Y = \Xmat_A \beta_A^* + w$, where $w$ is the error term, and the number of regressors $|A|$ is small. 
Quite frequently, however, the question arises whether important explanatory variables are missing from these simple models. 
This question can formally be checked by a statistical test of the low-dimensional model $\Y = \Xmat_A \beta_A^* + w$ against a high-dimensional alternative of the form \eqref{eq:partitionedmodel} that contains a large number $|B|$ of controls. 
More precisely speaking, a test of the null hypothesis $H_{0,B}: \beta_B^* = 0$ against the alternative $H_{1,B}: \beta_B^* \ne 0$ is required. 
Note that setting $A = \emptyset$ and $B = \{1,\ldots,p\}$ nests the previously discussed problem of testing $H_0$ against $H_1$ as a special case.

We construct a test of $H_{0,B}$ as follows: 
let $\mathcal{P} = \idmat - \Xmat_A (\Xmat_A^\top \Xmat_A)^{-1} \Xmat_A^\top$ be the projection matrix onto the orthogonal complement of the column space of $\Xmat_A$. 
Applying $\mathcal{P}$ to both sides of the model equation \eqref{eq:partitionedmodel} gives 
\begin{equation}\label{eq:projectedmodel}
\mathcal{P} \Y = \mathcal{P} \Xmat_B \beta_B^* + \uvec
\end{equation} 
with $\uvec = \mathcal{P} \eps$, which is itself a high-dimensional linear model with response $\mathcal{P} \Y$ and design matrix $\mathcal{P} \Xmat_B$. 
In order to test whether the parameter vector $\beta_B^*$ in model~\eqref{eq:projectedmodel} is equal to $0$, we use the same strategy as for the simpler problem of testing $H_0$:
our test statistic is given by 
\[ T_B = \frac{2\normsup{(\mathcal{P} \Xmat_B)^\top \mathcal{P} \Y}}{n}, \]
which implies that $T_B = 2\normsup{(\mathcal{P} \Xmat_B)^\top \uvec}  / n$ under $H_{0,B}$. 
The quantiles of the statistic $2 \normsup{(\mathcal{P} \Xmat_B)^\top \uvec}  / n$ can be approximated by our method developed in Section \ref{sec:estimateeffnoise}: 
define the criterion function
\[ \criterion_B(\lambda, e) = \max_{j \in B} \Big| \frac{2}{n} \sum_{i=1}^n (\mathcal{P} \Xmat_B)_{ij} \, \hat{\uvec}_{\lambda, i} \, e_i \Big|, \]
where $(\mathcal{P} \Xmat_B)_{ij}$ is the $(i,j)$-th element of the matrix $\mathcal{P} \Xmat_B$, $\hat{\uvec}_\lambda = \mathcal{P} \Y - \mathcal{P} \Xmat_B \hat{\beta}_{B,\lambda}$ is the residual vector which results from fitting the lasso with tuning parameter~$\lambda$ to the model \eqref{eq:projectedmodel}, and $e = (e_1,\ldots,e_n)^\top$ is a standard normal random vector independent of the data $(\Xmat,\Y)$. 
Moreover, let $\hat{q}_{\alpha,B}(\lambda)$ be the $(1-\alpha)$-quantile of $\criterion_B(\lambda,e)$ conditionally on $(\Xmat,\Y)$. 
As described in Section \ref{sec:estimateeffnoise}, we estimate the $(1-\alpha)$-quantile $\lambda_{\alpha,B}^*$ of $ 2 \normsup{(\mathcal{P} \Xmat_B)^\top \uvec}/n$ by 
\begin{equation*}
\hat{\lambda}_{\alpha,B} = \inf \big\{ \lambda > 0: \hat{q}_{\alpha,B}(\lambda^\prime) \le \lambda^\prime \text{ for all } \lambda^\prime \ge \lambda \big\}.
\end{equation*}
Our test of the hypothesis $H_{0,B}$ is now carried out as follows: 
reject $H_{0,B}$ at the significance level $\alpha$ if $T_B > \hat{\lambda}_{\alpha,B}$.

To derive the formal properties of the test, we define $\vartheta^{(j)} = \argmin_{\vartheta \in \R^{|A|}} \ex[ (X_{ij} - X_{i,A}^\top \vartheta)^2 ]$ with $X_{i,A} = (X_{ij}: j \in A)$. 
Put differently, we define $X_{i,A}^\top \vartheta^{(j)}$ to be the $L_2$-projection of $X_{ij}$ onto the linear subspace spanned by the elements of $X_{i,A}$. 
We assume that  $\min_{j \in B} \ex[ (X_{ij} - X_{i,A}^\top \vartheta^{(j)})^2 ] \ge c_\vartheta > 0$ for some constant $c_\vartheta$. 
Such an assumption is to be expected:
it essentially says  that the random variables $X_{ij}$ with $j \in B$ cannot be represented by a linear combination of the random variables $X_{ij}$ with $j \in A$.
The assumption is also mild;
 in particular, it is much weaker than irrepresentable-type conditions that are usually imposed in the context of variable selection for the lasso~\citep{vandeGeer09}. 
We can now summarize the formal properties of the test. 
\pagebreak
\begin{prop}\label{prop:testing:2}
Let the conditions of Theorem \ref{theo1} be satisfied, suppose for simplicity that the random variables $(\Xvec_i,\eps_i)$ are identically distributed across $i$, and let $|A|$ be a fixed number that does not grow with the sample size $n$. 
In addition, assume that the $|A| \times |A|$ matrix $\Psi_A = ( \ex[X_{ij} X_{ik}]: j, k \in A )$ is positive definite and that $\min_{j \in B} \ex[ (X_{ij} - X_{i,A}^\top \vartheta^{(j)})^2 ] \ge c_\vartheta > 0$. 
Under the null hypothesis $H_{0,B}: \beta_B^* = 0$, 
it holds that
\[  \pr (T_B \le \hat{\lambda}_{\alpha,B}) \ge 1 - \alpha + o(1). \]
Moreover, under any alternative $\beta_B^* \ne 0$ with the property that $\pr(\normsup{(\mathcal{P}\Xmat_B)^\top(\mathcal{P}\Xmat_B) \linebreak \beta_B^*}/n \ge c \sqrt{\log(n \lor p)/n}) \rightarrow 1$ for every $c > 0$, 
it holds that
\[ \pr ( T_B > \hat{\lambda}_{\alpha,B}) = 1 - o(1). \]
\end{prop}
\noindent This result shows that the proposed procedure is an asymptotic level-$\alpha$-test that has asymptotic power $1$ against any alternative $\beta^*_B \ne 0$ with the property that $\pr(\normsup{(\mathcal{P}\Xmat_B)^\top(\mathcal{P}\Xmat_B) \beta_B^*}/n \ge c \sqrt{\log(n \lor p)/n}) \rightarrow 1$ for any $c > 0$. 
The latter condition parallels the one in Proposition \ref{prop:testing:1}. 
The proof of Proposition \ref{prop:testing:2} is provided in the Appendix.

\section{Simulations}\label{sec:sim}

In this section, we corroborate our results through Monte Carlo experiments.
We simulate data from the linear regression model~\eqref{model} with sample size $n=500$ and dimension~$p \in \{250,500,1000\}$. 
The covariate vectors $\Xvec_i = (\Xvec_{i1},\ldots,\Xvec_{ip})^\top$ are independently sampled from a $p$-dimensional normal distribution with mean~$0$ and covariance matrix $(1-\kappa) \idmat + \kappa \boldsymbol{E}$, where $\idmat$ is the $p \times p$ identity matrix, $\boldsymbol{E} = (1,\ldots,1)^\top (1,\ldots,1) \in \R^{p \times p}$, and $\kappa\in[0,1)$ is the correlation between the entries of the covariate vector $\Xvec_i$. 
We show the simulation results for $\kappa=0.25$ unless indicated differently,
but we obtained similar results for other values of $\kappa$ as well.
The noise variables $\eps_i$ are drawn i.i.d.~from a normal distribution with mean $0$ and variance $\sigma^2 = 1$. 
The target vector $\beta^*$ has the form $\beta^* = (c,\ldots,c, 0,\ldots,0)^\top$, where the first $5$ entries are set to $c$ and the remaining ones to $0$. 
The value of $c$ is chosen such that one obtains a prespecified value for the signal-to-noise ratio $\text{SNR} = \sqrt{\normtwos{\Xmat\beta^*}/n} / \sigma = \sqrt{\normtwos{\Xmat\beta^*}/n}$.
We set $\text{SNR} = 1$ except when we analyze the  hypothesis tests from Section \ref{subsec:inference}: 
there, we consider the value $\text{SNR} = 0$, which corresponds to the null hypothesis, and the values $\text{SNR} \in \{0.1, 0.2\}$, which correspond to two different alternatives.
We implement our estimation method with $L=100$ bootstrap replicates, which seems sufficient across a wide variety of settings.
%Increasing the number of replicates appears to have no substantial effect on the performance of our method, at least not in the simulation scenarios considered here. 
The lasso paths are computed through \texttt{glmnet} \citep{glmnet} version $2.2.1$ with an equidistant grid of $\lambda$-values and $M=100$,
that is, $\lambda\in\{1\cdot 2\|\boldsymbol{X}^\top Y\|_\infty/(100n),2\cdot 2\|\boldsymbol{X}^\top Y\|_\infty/(100n),\dots\}$.
All Monte Carlo experiments are based on $N=1000$ simulation runs. 
The implementations are in \texttt{R} version $3.5.1$. 
%The code and all simulations can be found on \href{https://github.com/LedererLab/}{github.com/LedererLab/TBD}.

\subsection{Approximation quality}\label{subsec:sim:1}

\begin{figure}[t!]
\includegraphics[width=\textwidth]{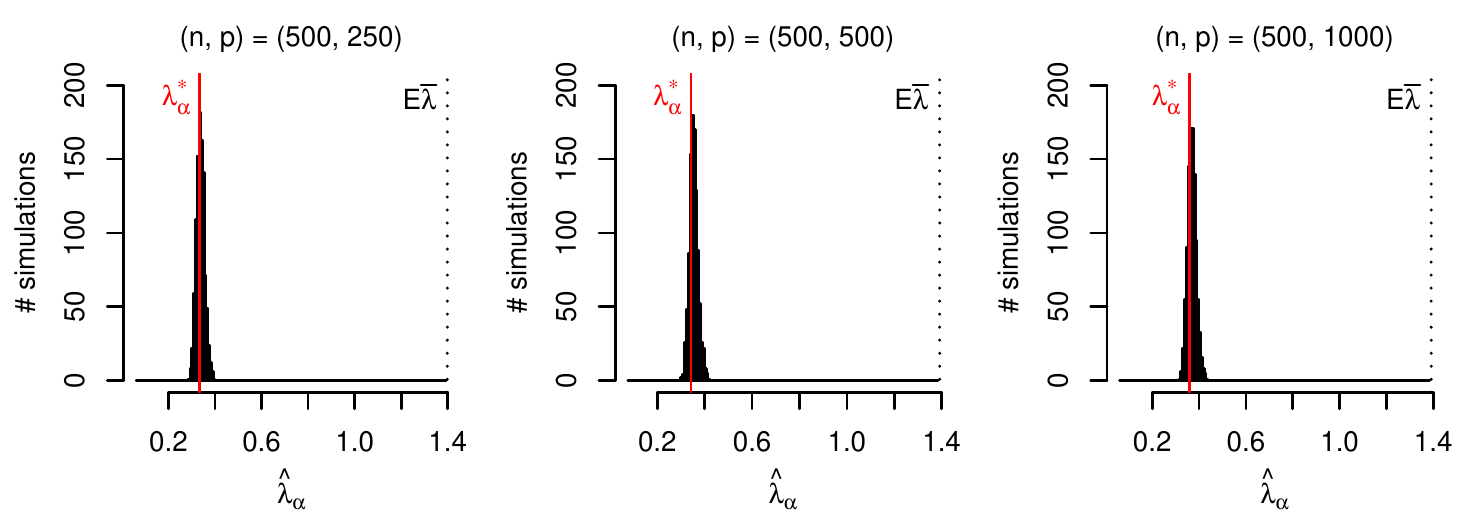}
\caption{Histograms of the estimates $\hat{\lambda}_\alpha$ for different values of $n$ and $p$. 
The red vertical lines indicate the values of the oracle parameter $\lambda_\alpha^*$; 
the dotted vertical lines give the values of $\ex[ \overline{\lambda} ]$, where $\overline{\lambda} = 2 \normsup{\Xmat^\top \Y}/n$ is the smallest $\lambda$ for which $\hat{\beta}_\lambda = 0$.}\label{fig:sim:histlambda}

\vspace{-0.45cm}
\end{figure}

We first examine how well the estimator $\hat{\lambda}_\alpha$ approximates the quantile $\lambda_\alpha^*$. 
Figure~\ref{fig:sim:histlambda} contains histograms of the $N=1000$ estimates of $\hat{\lambda}_\alpha$ for $\alpha=0.05$ and different values of $n$ and $p$. 
The red vertical line in each plot indicates the value of the quantile~$\lambda_\alpha^*$, which is computed by simulating $1000$ values of the effective noise $2 \normsup{\Xmat^\top \eps}/n$ and then taking their empirical $(1-\alpha)$-quantile. 
The $x$-axis covers the interval~$[0, \ex\overline{\lambda}]$ in each plot, where $\overline{\lambda} = 2 \normsup{\Xmat^\top \Y}/n$ is the smallest tuning parameter for which the lasso estimator is constantly equal to zero. 
This range is motivated as follows:
varying the tuning parameter $\lambda$ in the interval $[0,\overline{\lambda}]$ produces all possible lasso solutions.
It is thus natural to measure the approximation quality of $\hat{\lambda}_\alpha$ by the deviation $|\hat{\lambda}_\alpha - \lambda_\alpha^*|$ relative to the length of the interval $[0,\overline{\lambda}]$ rather than by the absolute deviation $|\hat{\lambda}_\alpha - \lambda_\alpha^*|$. 
This, in turn, suggests that the right scale to plot histograms of the estimates $\hat{\lambda}_\alpha$ is the interval $[0,\overline{\lambda}]$. 
Since this interval is stochastic, we let the $x$-axis of our plots span the interval $[0,\ex \overline{\lambda}]$ instead. 
The histogram plots of Figure~\ref{fig:sim:histlambda} can be regarded as an empirical illustration of the deviation inequalities derived after Theorem \ref{theo1}. They demonstrate that the estimates $\hat{\lambda}_\alpha$ approximate the oracle quantile $\lambda_\alpha^*$ accurately. 
%According to Figure~\ref{fig:sim:histlambda}, the estimates $\hat{\lambda}_\alpha$ approximate the oracle quantile $\lambda_\alpha^*$ accurately. 

\subsection{Tuning parameter calibration}\label{subsec:sim:2}

\begin{figure}[t!]
\centering
\begin{subfigure}[b]{\textwidth}
\centering
\includegraphics[width=\textwidth]{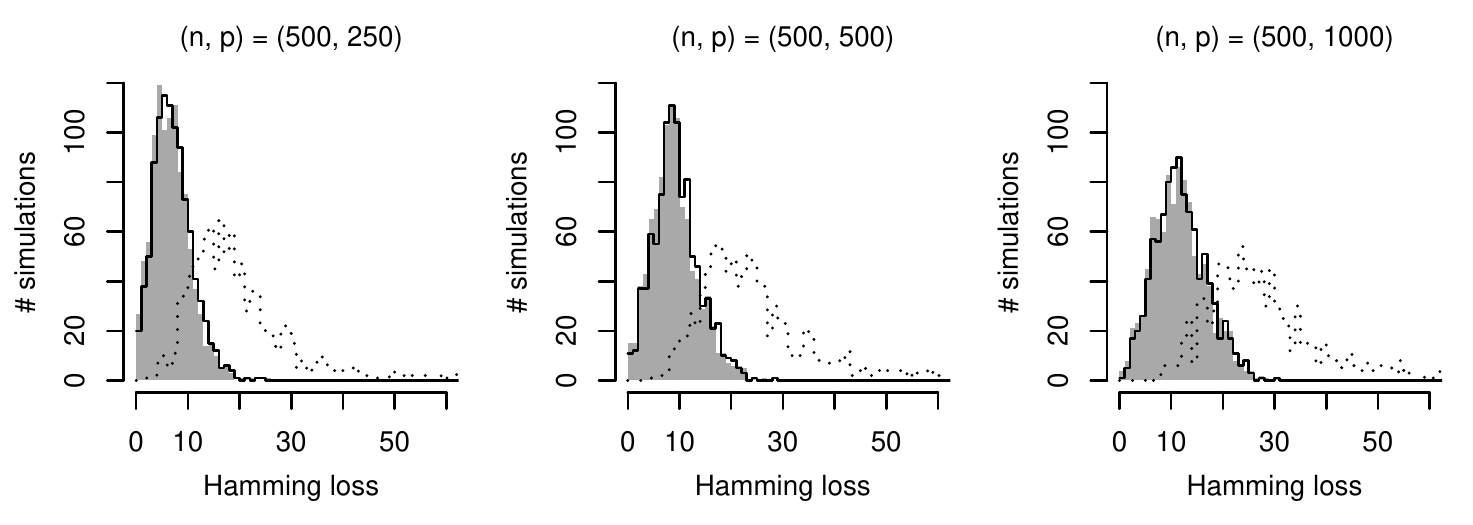}
\caption{Hamming distances for $\kappa = 0.25$}\label{subfig:a:sim:Hamming}
\end{subfigure} 
\vspace{-0.1cm}

\begin{subfigure}[b]{\textwidth}
\centering
\includegraphics[width=\textwidth]{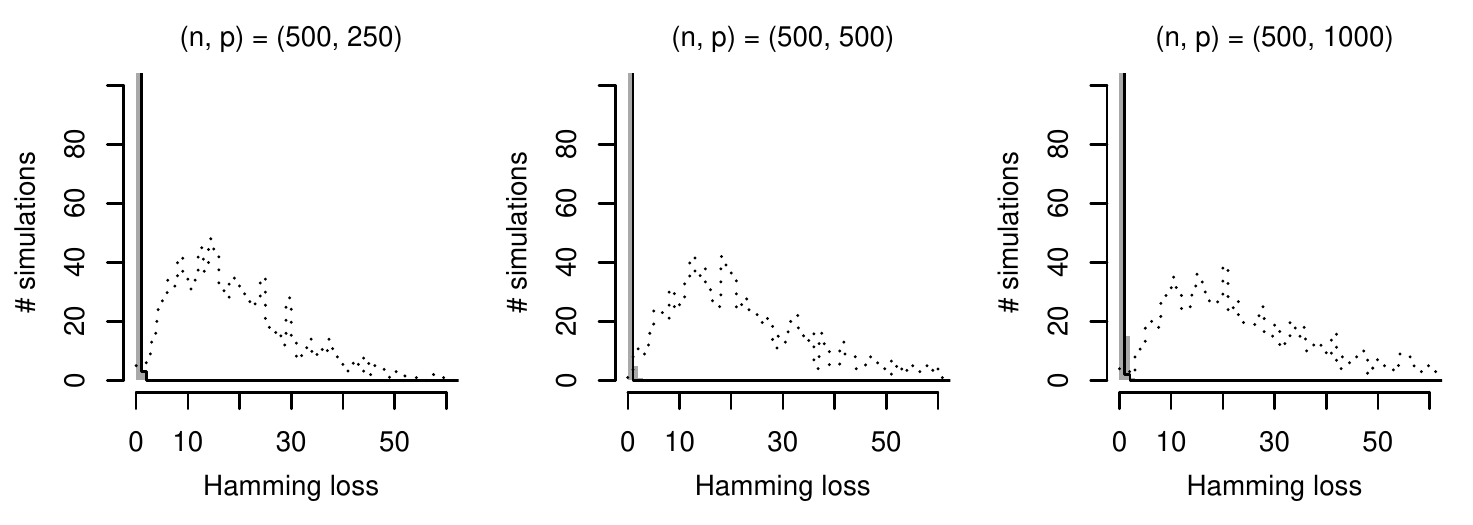}
\caption{Hamming distances for $\kappa = 0$}\label{subfig:b:sim:Hamming}
\end{subfigure} 
\caption{Histograms of the Hamming distances produced by the estimators $\hat{\beta}$, $\hat{\beta}_{\text{oracle}}$, and~$\hat{\beta}_{\text{CV}}$. 
The solid black lines indicate the histograms of~$\Delta_H(\hat{\beta},\beta^*)$, 
the gray-shaded areas indicate  the histograms of~$\Delta_H(\hat{\beta}_{\text{oracle}},\beta^*)$, 
and the dotted lines indicate the histograms of $\Delta_H(\hat{\beta}_{\text{CV}},\beta^*)$. 
The histograms of our estimator $\hat{\beta}$ and the oracle $\hat{\beta}_{\text{oracle}}$ in Subfigure~(b) essentially consist of only one bin at the value $0$ that goes up to almost $1000$ (which is the total number of simulation runs); to make the histograms of the cross-validated estimator visible, we cut the $y$-axis of the plots in Subfigure~(b) at the value $100$.
}\label{fig:sim:Hamming}

%\vspace{-0.45cm}
\end{figure}

We next investigate the performance of our method for calibrating the tuning parameter of the lasso. 
Our estimator of $\beta^*$ is defined as $\hat{\beta} := \hat{\beta}_{\hat{\lambda}_{\alpha}}$, 
where we use the estimator $\hat{\lambda}_\alpha$ with $\alpha = 0.05$ as the tuning parameter.
Our main interest is a comparison between~$\hat{\beta}$ and the oracle estimator $\hat{\beta}_{\text{oracle}} := \hat{\beta}_{\lambda_{\alpha}^*}$, 
which is tuned with the oracle parameter $\lambda_\alpha^*$ rather than its estimate~$\hat{\lambda}_\alpha$.
This comparison allows us to investigate whether~$\hat{\beta}$  is as accurate as suggested by our theory.
To highlight the practical performance of our estimator further,
we also compare~$\hat{\beta}$ to the lasso estimator $\hat{\beta}_{\text{CV}} := \hat{\beta}_{\hat{\lambda}_{\text{CV}}}$, 
where $\hat{\lambda}_{\text{CV}}$ is the tuning parameter chosen by $10$-fold cross-validation (which is performed on the same grid of $\lambda$-values as our method).
Of course, there are many other tuning parameter calibration schemes besides cross-validation,
but a comprehensive comparison of all calibration schemes is beyond the scope of this paper,
and, therefore, we focus on the arguably most popular representative.

\begin{figure}[t!]
\includegraphics[width=\textwidth]{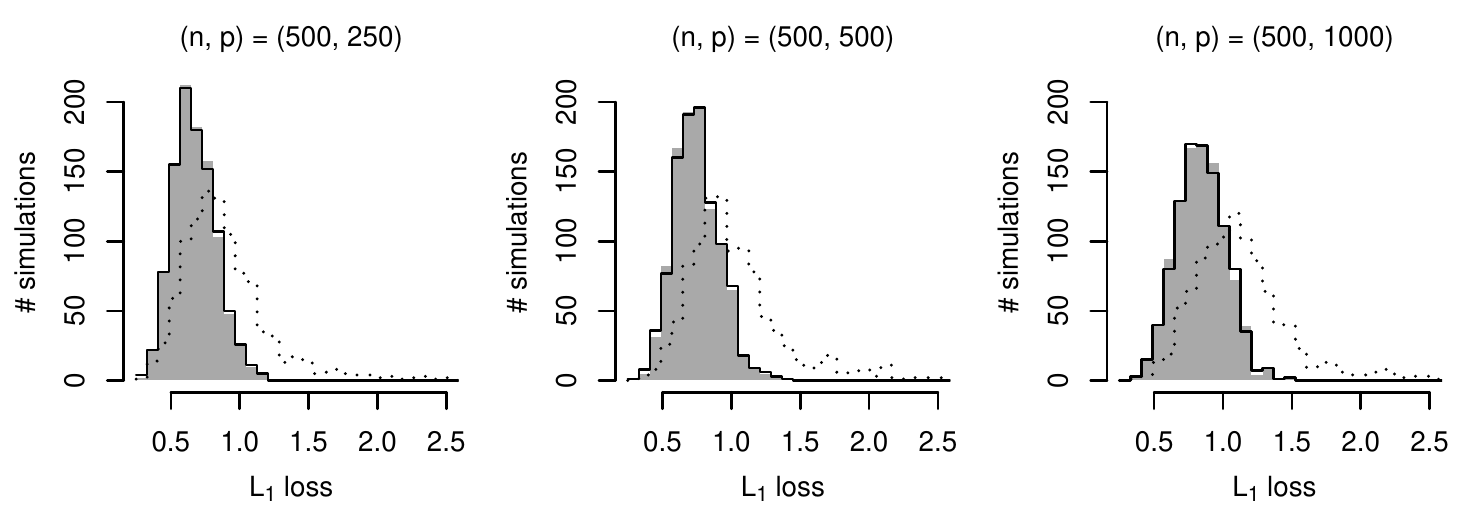}
\caption{Histograms of the $\ell_1$-loss produced by the estimators $\hat{\beta}$, $\hat{\beta}_{\text{oracle}}$, and $\hat{\beta}_{\text{CV}}$. 
The format of the plots is the same as in Figure \ref{fig:sim:Hamming}.}\label{fig:sim:L1}
\vspace{0.4cm}

\includegraphics[width=\textwidth]{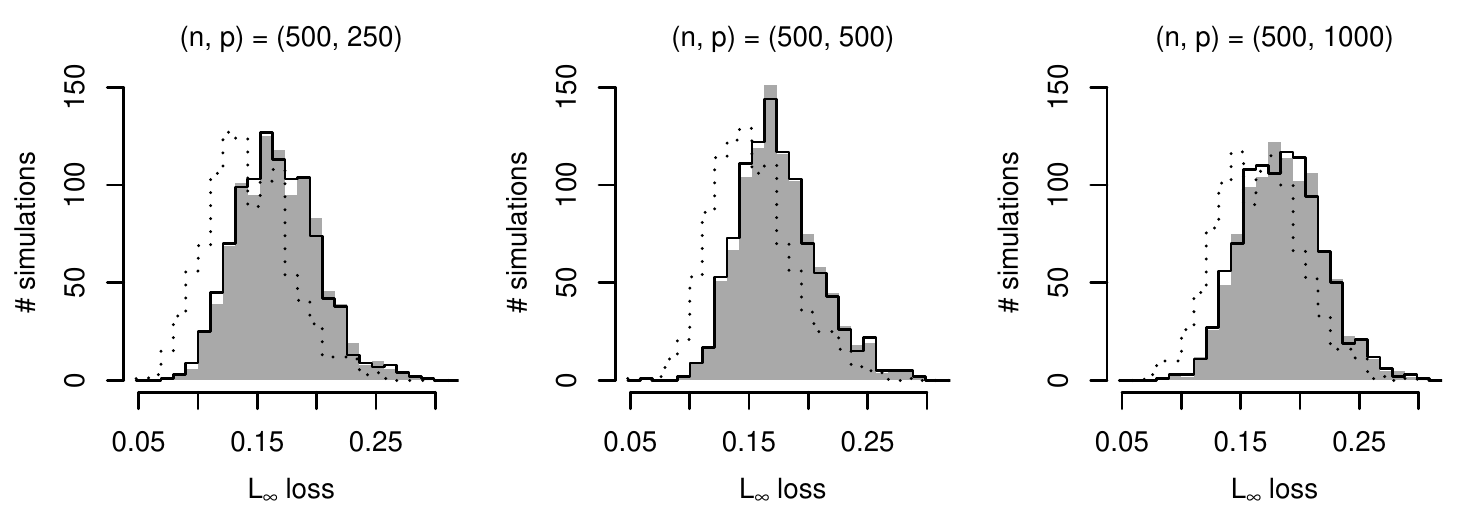}
\caption{Histograms of the $\ell_\infty$-loss produced by the estimators $\hat{\beta}$, $\hat{\beta}_{\text{oracle}}$, and $\hat{\beta}_{\text{CV}}$. 
The format of the plots is the same as in Figure \ref{fig:sim:Hamming}.}\label{fig:sim:sup}

\vspace{-0.75cm}
\end{figure}

We use four error measures to compare vectors~$\beta \in \Rp$ to~$\beta^*$: 
the Hamming distance $\Delta_H(\beta,\beta^*) = \sum_{j=1}^p | \ind (\beta_j = 0) - \ind (\beta_j^* = 0) |$,
the $\ell_1$-distance $\Delta_1(\beta, \beta^*) = \normone{\beta - \beta^*}$,
the $\ell_\infty$-distance $\Delta_\infty(\beta, \beta^*) = \normsup{\beta - \beta^*}$, and
the prediction error $\Delta_{\text{pr}}(\beta, \beta^*) = \normtwos{\Xmat(\beta - \beta^*)}/n$.
The Hamming distance allows us to investigate the variable selection properties of the estimators $\hat{\beta}$, $\hat{\beta}_{\text{oracle}}$ and $\hat{\beta}_{\text{CV}}$: 
the quantity $\Delta_H(\beta,\beta^*)$ counts the number of false-negative and false-positive entries in the vector $\beta$, 
where the entry $j$ is defined to be a false negative if $\beta_j^* \ne 0$ but $\beta_j = 0$ and a false positive if $\beta_j^* = 0$ but $\beta_j \ne 0$. 
The $\ell_p$-loss with $\ell \in \{1,\infty\}$ and the mean-squared prediction error $\Delta_{\text{pr}}$, on the other hand,
allow us to investigate the estimators' estimation and prediction properties, respectively.

\begin{figure}[t!]
\includegraphics[width=\textwidth]{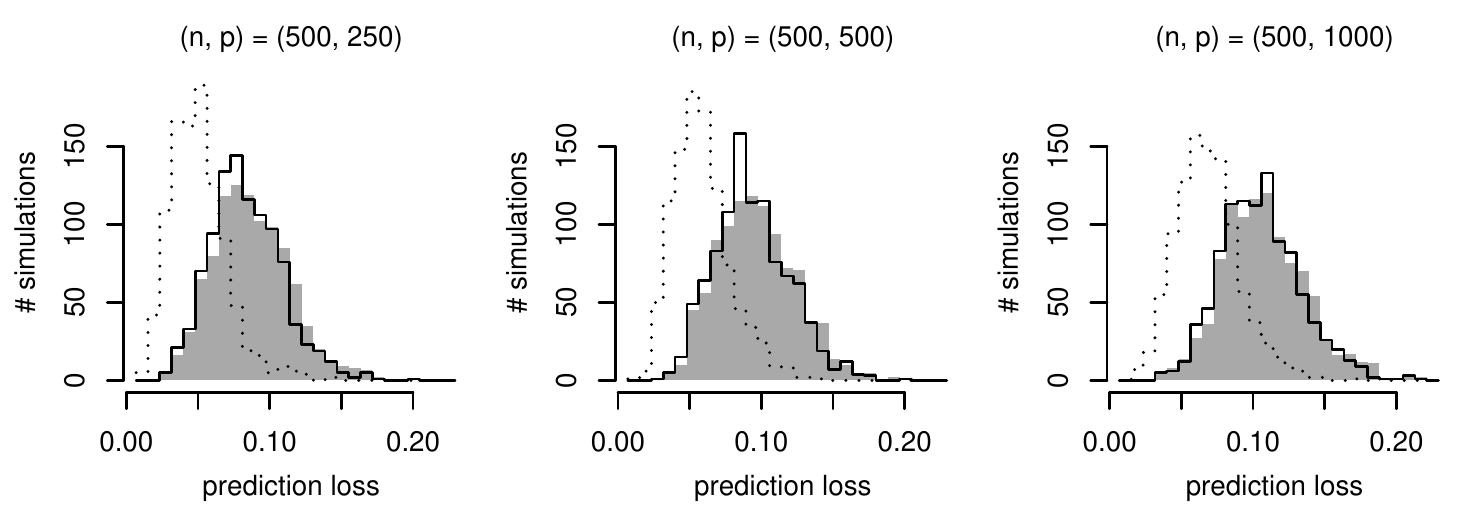}
\caption{Histograms of the prediction loss produced by the estimators $\hat{\beta}$, $\hat{\beta}_{\text{oracle}}$, and $\hat{\beta}_{\text{CV}}$. 
The format of the plots is the same as in Figure \ref{fig:sim:Hamming}.}\label{fig:sim:prediction}

%\vspace{-0.45cm}
\end{figure}

The simulation results for the Hamming distance are reported in Figure~\ref{subfig:a:sim:Hamming} for our usual value $\kappa = 0.25$ of the correlation and in Figure~\ref{subfig:b:sim:Hamming} for $\kappa = 0$. 
The black line in each plot depicts the histogram of the Hamming distances $\Delta_H(\hat{\beta},\beta^*)$ that are produced by our estimator $\hat{\beta}$ over the $N=1000$~simulation runs,
the gray-shaded area depicts the histogram of $\Delta_H(\hat{\beta}_{\text{oracle}},\beta^*)$ produced by the oracle $\hat{\beta}_{\text{oracle}}$,
and the dotted line depicts the histogram of $\Delta_H(\hat{\beta}_{\text{CV}},\beta^*)$ produced by the cross-validated estimator~$\hat{\beta}_{\text{CV}}$.

Comparing Figures~\ref{subfig:a:sim:Hamming} and \ref{subfig:b:sim:Hamming},
we find that both the oracle and our estimator provide more accurate variable selection for smaller correlations -- in line with theories for the lasso~\citep{Zhao06}.
We also find that both the oracle and our estimator provide more accurate variable selection than cross-validation 
-- in line with the well-known fact that cross-validation typically overselects.
Finally, we find that the histograms of our estimator are virtually the same as the ones of the oracle estimator -- in line with our theory.

The simulation results for the $\ell_1$-norm are reported in Figure \ref{fig:sim:L1} and for the $\ell_\infty$-norm in Figure \ref{fig:sim:sup}. 
We find again that for both the $\ell_1$- and the $\ell_\infty$-loss, the histograms produced by our estimator~$\hat{\beta}$ are extremely close to those of the oracle~$\hat{\beta}_{\text{oracle}}$, 
meaning that the performance of our procedure matches the performance of the  oracle.
We also find that our estimator improves on cross-validation in terms of the $\ell_1$-norm but slightly loses in terms of the $\ell_\infty$-norm. 
The reason for this difference is that  $\hat{\lambda}_{\text{CV}}$ tends to be much smaller than $\hat{\lambda}_\alpha$ and $\lambda_\alpha^*$; 
this induces an accumulation of small, spurious parameters, which affects the $\ell_1$-norm more than the $\ell_\infty$-norm.

The simulation results for the prediction error are reported in Figure \ref{fig:sim:prediction}. 
Once more, the histograms of our estimator are extremely close to those of the oracle.
Cross-validation performs best, which is no surprise in view of it being specifically designed for this task.

The two main conclusions from the simulations are that our method 
(i)~exhibits virtually the same performance as the oracle and
(ii)~rivals cross-validation in terms of variable selection and estimation but not necessarily prediction.

\subsection{Inference}\label{subsec:sim:3}

\begin{table}[!t] 
\setlength{\tabcolsep}{2pt}
\centering 
\caption{Empirical size under the null and power against different alternatives.}\label{table:sim:test}

{\small 
\begin{subtable}[b]{\textwidth}
\centering 
\caption{empirical size under $H_0: \beta^* = 0$}\label{subtable:sim:test:null} 
\begin{tabular}{@{\extracolsep{5pt}} lcccccccc} 
\\[-1.8ex]\hline 
\hline \\[-1.8ex] 
 & & \multicolumn{3}{c}{feasible test} & & \multicolumn{3}{c}{oracle test} \\
 & & $\alpha=0.01$ & $\alpha=0.05$ & $\alpha=0.1$ & & $\alpha=0.01$ & $\alpha=0.05$ & $\alpha=0.1$ \\[0.1cm]
\hline \\[-1.8ex] 
$(n, p) = (500, 250)$ & & $0.024$ & $0.057$ & $0.110$ & & $0.021$ & $0.056$ & $0.087$ \\ 
$(n, p) = (500, 500)$ & & $0.018$ & $0.050$ & $0.097$ & & $0.008$ & $0.064$ & $0.116$ \\ 
$(n, p) = (500, 1000)$ & & $0.015$ & $0.044$ & $0.082$ & & $0.010$ & $0.050$ & $0.095$ \\[0.1cm] 
\hline\\ 
\end{tabular} 
\end{subtable}

\begin{subtable}[b]{\textwidth}
\centering 
\caption{empirical power under the alternative with $\text{SNR} = 0.1$}\label{subtable:sim:test:alt01} 
\begin{tabular}{@{\extracolsep{5pt}} lcccccccc} 
\\[-1.8ex]\hline 
\hline \\[-1.8ex] 
 & & \multicolumn{3}{c}{feasible test} & & \multicolumn{3}{c}{oracle test} \\
 & & $\alpha=0.01$ & $\alpha=0.05$ & $\alpha=0.1$ & & $\alpha=0.01$ & $\alpha=0.05$ & $\alpha=0.1$ \\[0.1cm]
\hline \\[-1.8ex] 
$(n, p) = (500, 250)$ & & $0.151$ & $0.304$ & $0.440$ & & $0.118$ & $0.341$ & $0.458$ \\ 
$(n, p) = (500, 500)$ & & $0.148$ & $0.293$ & $0.433$ & & $0.089$ & $0.341$ & $0.456$ \\ 
$(n, p) = (500, 1000)$ & & $0.122$ & $0.284$ & $0.409$ & & $0.090$ & $0.293$ & $0.417$ \\[0.1cm] 
\hline \\
\end{tabular} 
\end{subtable}

\begin{subtable}[b]{\textwidth}
\centering 
\caption{empirical power under the alternative with $\text{SNR} = 0.2$}\label{subtable:sim:test:alt02}
\begin{tabular}{@{\extracolsep{5pt}} lcccccccc} 
\\[-1.8ex]\hline 
\hline \\[-1.8ex] 
 & & \multicolumn{3}{c}{feasible test} & & \multicolumn{3}{c}{oracle test} \\
 & & $\alpha=0.01$ & $\alpha=0.05$ & $\alpha=0.1$ & & $\alpha=0.01$ & $\alpha=0.05$ & $\alpha=0.1$ \\[0.1cm]
\hline \\[-1.8ex] 
$(n, p) = (500, 250)$ & & $0.644$ & $0.850$ & $0.923$ & & $0.664$ & $0.890$ & $0.940$ \\ 
$(n, p) = (500, 500)$ & & $0.631$ & $0.840$ & $0.909$ & & $0.579$ & $0.880$ & $0.926$ \\ 
$(n, p) = (500, 1000)$ & & $0.599$ & $0.811$ & $0.904$ & & $0.600$ & $0.867$ & $0.918$ \\[0.1cm] 
\hline \\[-1.8ex] 
\end{tabular} 
\end{subtable}}

%\vspace{-0.45cm}
\end{table}

We finally explore the empirical performance of the tests developed in Section~\ref{subsec:inference}. 
We focus on the simpler test $H_0: \beta^* = 0$ against $H_1: \beta^* \ne 0$,
where we reject $H_0$ at the significance level $\alpha$ if $T=2 \normsup{\Xmat^\top \Y}/n > \hat{\lambda}_\alpha$.
We compare this test with an oracle version that rejects $H_0$ if $T > \lambda_\alpha^*$.
Similarly as before, this comparison allows us to investigate if our practical test matches its theoretical (and in practice infeasible) analog as suggested by our theory. 
The simulation setup is as described before, including the mentioned variations over the signal-to-noise ratio $\text{SNR}$:
the value $\text{SNR} = 0$ specifies the null hypothesis $H_0: \beta^* = 0$;
the values $\text{SNR}\in\{0.1,0.2\}$ specify the alternative (the larger $\text{SNR}$, the further the setup deviates from the null).

Table \ref{subtable:sim:test:null} reports the empirical size of our feasible test and of its oracle version under the null for different values of the nominal size $\alpha$, sample size $n$, and dimension~$p$. 
The empirical size is defined as the number of rejections divided by the total number of simulation runs.
We find that the size both of our feasible test and of the oracle test is close to the target $\alpha$ in all considered scenarios.

Tables \ref{subtable:sim:test:alt01} and \ref{subtable:sim:test:alt02} report the empirical power of the tests for the signal-to-noise ratios $\text{SNR} = 0.1$ and $\text{SNR}=0.2$, respectively. 
The empirical power is again defined as the number of rejections divided by the total number of simulation runs.
We find that the power increases when the signal-to-noise ratio $\text{SNR}$ goes up, as expected.
We further find that the power of our test is very similar to the one of the oracle test.
Moreover, the power can be seen to be quite substantial despite the small signal-to-noise ratios.

We conclude that our test has (i)~similar performance as its oracle version, (ii)~sizes close to the nominal ones, and (iii)~considerable power against alternatives.

\subsection{Robustness checks}

In Section  \ref{sec:supp:robustness} of the Supplementary Material, we carry out some robustness checks. We in particular examine how our simulation results are affected by different distributions of the noise $\eps_i$ and the design $\Xvec_i$, how our method for tuning parameter calibration is influenced by the choice of $\alpha$, and what is the effect of the number of bootstrap iterations $L$ and the grid size $M$ on the simulation results.

\section*{Acknowledgements}
We thank the editor and the reviewers for their insightful comments.

\allowdisplaybreaks[3]

\section*{Appendix}
\def\theequation{A.\arabic{equation}}

In what follows, we prove the main theoretical results of the paper. We assume throughout that the technical conditions \ref{C1}--\ref{C5} are fulfilled.

\subsection*{Notation}

Throughout the Appendix, the symbols $B$, $c$, $C$, $D$ and $K$ denote generic constants that may take a different value on each occurrence. Moreover, the symbols $B_j$, $c_j$, $C_j$, $D_j$ and $K_j$ with subscript $j$ (which may be either a natural number or a letter) are specific constants that are defined in the course of the Appendix. Unless stated differently, the constants $B$, $c$, $C$, $D$, $K$, $B_j$, $c_j$, $C_j$, $D_j$ and $K_j$ depend neither on the sample size $n$ nor on the dimension $p$. For ease of notation, we let $\Theta = \{ c_\Xvec, C_\Xvec, c_\sigma, C_\sigma, C_\theta, \theta, C_r, r, C_\beta, \delta_\beta \}$ be the list of model parameters specified in \ref{C1}--\ref{C5}. For $a$, $b \in \R$, we write $a \lor b = \max\{a,b\}$. The random variables $\Xmat$, $\eps$ and $e$ are assumed to be defined on the same probability space $(\Omega, \mathcal{A}, \pr)$ for all $n \ge 1$. We write $\pr_e(\cdot) = \pr(\, \cdot \, | \Xmat, \eps)$ and $\ex_e[ \, \cdot \, ] = \ex[ \, \cdot \, | \Xmat,\eps]$ to denote the probability and expectation conditionally on $\Xmat$ and $\eps$.

To derive the theoretical results of the paper, it is convenient to reformulate the estimator $\hat{\lambda}_\alpha$ as follows: define $\crithat(\gamma,e) = \max_{1\le j \le p} |\What_j(\gamma,e)|$, where
\[ \What(\gamma,e) = \big(\What_1(\gamma,e),\ldots,\What_p(\gamma,e)\big)^\top \quad \text{with} \quad \What_j(\gamma,e) = \frac{1}{\sqrt{n}} \sum\limits_{i=1}^n \Xvec_{ij} \hat{\eps}_{\frac{2}{\sqrt{n}}\gamma,i} e_i. \]
Moreover, let $\hat{\pi}_\alpha(\gamma)$ be the $(1-\alpha)$-quantile of $\crithat(\gamma,e)$ conditionally on $\Xmat$ and $\Y$, that is, conditionally on $\Xmat$ and $\eps$, which is formally defined as $\hat{\pi}_\alpha(\gamma) = \inf \{ q: \pr_e (\crithat(\gamma,e) \le q) \ge 1-\alpha \}$, and set 
\[ \hat{\gamma}_\alpha = \inf \big\{ \gamma > 0: \hat{\pi}_\alpha(\gamma^\prime) \le \gamma^\prime \text{ for all } \gamma^\prime \ge \gamma \big\}. \]
The quantities $\crithat(\gamma,e)$, $\hat{\pi}_\alpha(\gamma)$ and $\hat{\gamma}_\alpha$ are related to $\criterion(\lambda,e)$, $\hat{q}_\alpha(\lambda)$ and $\hat{\lambda}_\alpha$ by the equations
\begin{equation*}
\crithat\Big(\frac{\sqrt{n}}{2}\lambda,e\Big) = \frac{\sqrt{n}}{2} \criterion(\lambda,e), \quad
\hat{\pi}_\alpha\Big(\frac{\sqrt{n}}{2}\lambda\Big) = \frac{\sqrt{n}}{2} \hat{q}_\alpha(\lambda), \quad 
\hat{\gamma}_\alpha = \frac{\sqrt{n}}{2} \hat{\lambda}_\alpha.
\end{equation*}
Hence, $\hat{\gamma}_\alpha$ is a rescaled version of the estimator $\hat{\lambda}_\alpha$. In particular, we can reformulate $\hat{\lambda}_\alpha$ in terms of $\hat{\gamma}_\alpha$ as $\hat{\lambda}_\alpha = 2 \hat{\gamma}_\alpha/\sqrt{n}$.

For our proof strategy, we require some auxiliary statistics which are closely related to $\crithat(\gamma,e)$, $\hat{\pi}_\alpha(\gamma)$ and $\hat{\gamma}_\alpha$. To start with, we define $\crit(e) = \max_{1\le j \le p} |\W_j(e)|$, where
\[ \W(e) = \big(\W_1(e),\ldots,\W_p(e)\big)^\top \quad \text{with} \quad \W_j(e) = \frac{1}{\sqrt{n}} \sum\limits_{i=1}^n \Xvec_{ij} \eps_i e_i, \]
and let $\gamma_\alpha$ be the $(1-\alpha)$-quantile of $\crit(e)$ conditionally on $\Xmat$ and $\eps$. Moreover, we set $\critstar = \max_{1\le j \le p} |\Wstar_j|$, where 
\[ \Wstar = (\Wstar_1,\ldots,\Wstar_p)^\top \quad \text{with} \quad \Wstar_j = \frac{1}{\sqrt{n}} \sum\limits_{i=1}^n \Xvec_{ij} \eps_i, \] 
and let $\gamma_\alpha^*$ be the $(1-\alpha)$-quantile of $\critstar$. Notice that $\gamma_\alpha^*$ is a rescaled version of $\lambda_\alpha^*$, in particular, $\gamma_\alpha^* = \sqrt{n} \lambda_\alpha^* / 2$. Finally, we define $\critgauss = \max_{1 \le j \le p} |\Wgauss_j|$, where $\Wgauss = (\Wgauss_1,\ldots,\Wgauss_p)^\top$ is a Gaussian random vector with the same covariance structure as $\Wstar$, that is, $\ex[\Wgauss] = \ex[\Wstar] = 0$ and $\ex[\Wgauss \Wgauss^\top] = \ex[\Wstar (\Wstar)^\top]$, and we denote the $(1-\alpha)$-quantile of $\critgauss$ by $\gamma_\alpha^{\Wgauss}$.

\pagebreak

In order to relate the criterion function $\crithat(\gamma,e)$ to the term $\crit(e)$, we make use of the simple bound 
\begin{equation}\label{eq:critbound}
\crithat(\gamma,e) 
\begin{cases}
\le \crit(e) + \remainder(\gamma,e) \\
\ge \crit(e) - \remainder(\gamma,e), 
\end{cases}
\end{equation}
where 
\[ \remainder(\gamma,e) = \max_{1 \le j \le p} \biggl| \frac{1}{\sqrt{n}} \sum_{i=1}^n \Xvec_{ij} \Xvec_i^\top \big(\beta^* - \hat{\beta}_{\frac{2}{\sqrt{n}}\gamma}\big) e_i \biggr|. \]
For our technical arguments, we further define the expression
\begin{equation*}
\Delta = \max_{1 \le j,k \le p} \big| \Sigma_{jk} - \Sigma_{jk}^* \big| = \max_{1 \le j,k \le p} \biggl| \frac{1}{n} \sum\limits_{i=1}^n \big( \Xvec_{ij} \Xvec_{ik} \eps_i^2 - \ex[ \Xvec_{ij} \Xvec_{ik} \eps_i^2 ] \big) \biggr|, 
\end{equation*}
where $\Sigma = (\Sigma_{jk}: 1 \le j,k \le p) = \ex_e [\W(e) \W(e)^\top]$ is the covariance matrix of $\W(e)$ conditionally on $\Xmat$ and $\eps$, and $\Sigma^* = (\Sigma_{jk}^*: 1 \le j,k \le p) = \ex[ \Wstar (\Wstar)^\top]$ is the covariance matrix of $\Wstar$. We finally introduce the event
\[ \mathcal{S}_\gamma = \Big\{ \frac{1}{\sqrt{n}} \normsup{\Xmat^\top \eps} \le \gamma \Big\}, \] 
which relates to $\mathcal{T}_\lambda$ by the equation $\mathcal{S}_{\sqrt{n} \lambda / 2} = \mathcal{T}_\lambda$, as well as the event 
\begin{equation}\label{eq:def:eventA}
\mathcal{A}_n = \big\{ \Delta \le B_\Delta \sqrt{\log(n \lor p)/n} \big\}, 
\end{equation}
where the constant $B_\Delta$ is defined in Lemma \ref{lemmaA1} below.

\subsection*{Auxiliary results}

Before we prove the main results of the paper, we derive some auxiliary lemmas which are needed later on. Their proofs can be found in the Supplementary Material.

\begin{lemmaA}\label{lemmaA1}
There exist positive constants $B_\Delta$, $C_\Delta$ and $K_\Delta$ that depend only on the model parameters $\Theta$ such that 
\[ \pr \big( \Delta > B_\Delta \sqrt{\log(n \lor p) / n} \big) \le C_\Delta n^{-K_\Delta}. \]
In particular, $K_\Delta$ can be chosen to be any positive constant with $K_\Delta < (\theta-4)/4$, where $\theta > 4$ is defined in \ref{C3}.  
\end{lemmaA}

\begin{lemmaA}\label{lemmaA2}
On the event $\mathcal{S}_\gamma$, it holds that 
\[ \pr_e \biggl( \remainder(\gamma^\prime,e) > \frac{B_\remainder (\log n)^2 \sqrt{\normone{\beta^*} \gamma^\prime}}{n^{1/4}} \biggr) \le C_\remainder \, n^{-K_{\remainder}} \]
for any $\gamma^\prime \ge \gamma$, where the constants $B_\remainder$, $C_\remainder$ and $K_{\remainder}$ depend only on the model parameters $\Theta$, and $K_{\remainder}$ can be chosen as large as desired by picking $C_\remainder$ large enough. 
\end{lemmaA}

\begin{lemmaA}\label{lemmaA3}
For every $\alpha > 1/(n \lor p)$, it holds that
\begin{equation*}
\gamma_\alpha^{\Wgauss} \le C_\Xvec C_\sigma \bigl[ \sqrt{2 \log (2 p)} + \sqrt{2\log (n \lor p)} \bigr],
\end{equation*}
where the constants $C_\Xvec$ and $C_\sigma$ are defined in \ref{C2} and \ref{C3}, respectively.
\end{lemmaA}

%\begin{lemmaA}\label{lemmaA3:lowerbound}
%For every $\delta_0 > 0$, there exists a constant $c_{\min} = c_{\min}(\delta_0) > 0$ that depends only on $\delta_0$ and the parameters $\Theta$ such that $\gamma_\alpha^\Wgauss \ge c_{\min}$ for all $\alpha \in (0,1-\delta_0]$. 
%\end{lemmaA} 

In addition to Lemmas \ref{lemmaA1}--\ref{lemmaA3}, we state some results on high-dimensional Gaussian approximations and anti-concentration bounds for Gaussian random vectors from \cite{Chernozhukov2013} and \cite{Chernozhukov2015} that are required for the proofs in the sequel. The first result is an anti-concentration bound which is taken from \cite{Chernozhukov2015} -- see their Theorem 3 and Corollary 1. 
\begin{lemmaA}\label{lemmaA5}
Let $(V_1,\ldots,V_p)^\top$ be a centered Gaussian random vector in $\Rp$. Suppose that there are constants $0 < c_3 < C_3 < \infty$ with $c_3 \le \sigma_j \le C_3$, where $\sigma_j^2 = \ex[V_j^2]$ for $1 \le j \le p$. Then for every $\delta > 0$, 
\[ \sup_{t \in \R} \pr \biggl( \Big| \max_{1 \le j \le p} V_j - t \Big| \le \delta \biggr) \le C \delta \sqrt{1 \lor \log(p/\delta)}, \]
where $C > 0$ depends only on $c_3$ and $C_3$. 
\end{lemmaA}

\noindent The next two results correspond to Theorem 2 in \cite{Chernozhukov2015} and Corollary 2.1 in \cite{Chernozhukov2013}, respectively. 
\begin{lemmaA}\label{lemmaA6}
Let $V = (V_1,\ldots,V_p)^\top$ and $V^\prime = (V_1^\prime,\ldots,V_p^\prime)^\top$ be centered Gaussian random vectors in $\Rp$ with covariance matrices $\Sigma^V = (\Sigma_{jk}^V: 1 \le j,k \le p)$ and $\Sigma^{V^\prime}= (\Sigma_{jk}^{V^\prime}: 1 \le j,k \le p)$, respectively, and define $\delta = \max_{1 \le j,k \le p} |\Sigma_{jk}^V - \Sigma_{jk}^{V^\prime}|$. Suppose that there are constants $0 < c_4 < C_4 < \infty$ with $c_4 \le \Sigma_{jj}^V \le C_4$ for $1 \le j \le p$. Then there exists a constant $C > 0$ that depends only on $c_4$ and $C_4$ such that
\begin{align*}
\sup_{t \in \R} \biggl| \pr \Big( \max_{1 \le j \le p} V_j \le t \Big) - & \, \pr \Big( \max_{1 \le j \le p} V_j^\prime \le t \Big) \biggr| \\ & \le C \delta^{1/3} \big\{ 1 \lor 2 \log p \lor \log(1/\delta) \big\}^{1/3} (\log p)^{1/3}. 
\end{align*}   
\end{lemmaA}

\begin{lemmaA}\label{lemmaA7} 
Let $Z_i = (Z_{i1},\ldots,Z_{ip})^\top$ be independent $\Rp$-valued random vectors for $1 \le i \le n$ with mean zero and the following properties: $c_5 \le n^{-1} \sum_{i=1}^n \ex[Z_{ij}^2] \le C_5$ and $\max_{k=1,2} \{ n^{-1} \sum_{i=1}^n \ex[|Z_{ij}|^{2+k} / D_n^k] \} + \ex[(\max_{1 \le j \le p} |Z_{ij}|/D_n)^4] \le 4$, where $c_5 > 0$, $C_5 > 0$ and $D_n \ge 1$ is such that $D_n^4 (\log (pn))^7 / n \le C_6 n^{-c_6}$ for some constants $c_6 > 0$ and $C_6 > 0$. Define
\[ W = (W_1,\ldots,W_p)^\top \quad \text{with} \quad W_j = \frac{1}{\sqrt{n}} \sum\limits_{i=1}^n Z_{ij} \]
and let $V = (V_1,\ldots,V_p)^\top$ be a Gaussian random vector with the same mean and covariance as $W$, that is, $\ex[V] = \ex[W] = 0$ and $\ex[V V^\top] = \ex[W W^\top]$. Then there exist constants $C>0$ and $K>0$ that depend only on $c_5$, $C_5$, $c_6$ and $C_6$ such that 
\[ \sup_{t \in \R} \biggl| \pr\Big( \max_{1 \le j \le p} W_j \le t \Big) - \pr\Big( \max_{1 \le j \le p} V_j \le t \Big) \biggr| \le C n^{-K}. \]
\end{lemmaA}

The final lemma of this section concerns the quantiles of Gaussian maxima.
\begin{lemmaA}\label{lemmaA8}
Let $(V_1,\ldots,V_p)^\top$ be a centered Gaussian random vector in $\Rp$ which fulfills the conditions of Lemma \ref{lemmaA5}. Moreover, let $\gamma_\alpha^V$ be the $(1-\alpha)$-quantile of $\max_{1 \le j \le p} V_j$, which is formally defined as $\gamma_\alpha^V = \inf \{ q: \pr( \max_{1 \le j \le p} V_j \le q) \ge 1 - \alpha \}$. It holds that 
\[ \pr \Big( \max_{1 \le j \le p} V_j \le \gamma_\alpha^V \Big) = 1 - \alpha \]
for every $\alpha \in (0,1)$.  
\end{lemmaA}

\begin{remarkA}\label{remarklemmasA5678}
\textnormal{Note that Lemmas \ref{lemmaA5}--\ref{lemmaA8} continue to hold for maxima of the form $\max_{1 \le j \le p} |V_j|$, $\max_{1 \le j \le p} |V_j^\prime|$ and $\max_{1 \le j \le p} |W_j|$. This follows from the fact that $\max_{1 \le j \le p}|V_j| = \max_{1 \le j \le 2p} U_j$ with $U_j = V_j$ and $U_{p+j} = -V_j$ for $1 \le j \le p$.} 
\end{remarkA}

\subsection*{Proof of Theorem \ref{theo1}}

The proof proceeds in several steps. To start with, we formally relate the quantiles $\gamma_\alpha^*$, $\gamma_\alpha$ and $\gamma_\alpha^{\Wgauss}$ to each other.

\begin{propA}\label{propA1}
There exist positive constants $C$ and $K$ that depend only on the model parameters $\Theta$ such that \\[-1cm]
\begin{align*}
 & \gamma_{\alpha+\kappa_n}^* \le \gamma_\alpha^{\Wgauss} \le \gamma_{\alpha-\kappa_n}^* \\ 
 & \gamma_{\alpha+\kappa_n}^{\Wgauss} \le \gamma_\alpha^* \le \gamma_{\alpha-\kappa_n}^{\Wgauss} 
\end{align*}
for any $\alpha \in (\kappa_n,1-\kappa_n)$ with $\kappa_n = C n^{-K}$.
\end{propA}
\begin{proof}[\textnormal{\textbf{Proof of Proposition \ref{propA1}.}}] 
From \ref{C2} and \ref{C3}, it immediately follows that $0 < c_5 \le n^{-1} \sum_{i=1}^n \ex[(\Xvec_{ij} \eps_i)^2] \le C_5 < \infty$ and $\max_{k=1,2} \{ n^{-1} \sum_{i=1}^n \ex[|\Xvec_{ij} \eps_i|^{2+k}/D^k] \} + \ex[ (\max_{1 \le j \le p} |\Xvec_{ij} \eps_i| / D)^4] \le 4$ for some appropriately chosen constants $c_5$, $C_5$ and $D$ that depend only on the parameters $\Theta$. Since $D$ does not depend on $n$, it further holds that $D^4 (\log (pn))^7 / n \le C_6 n^{-c_6}$, where $c_6$ can be chosen to be any positive constant strictly smaller than $1$ provided that $C_6$ is picked sufficiently large. Hence, we can apply Lemma \ref{lemmaA7} to the terms $\critstar = \max_{1\le j \le p} |\Wstar_j|$ and  $\critgauss = \max_{1 \le j \le p} |\Wgauss_j|$ to obtain the following: there exist constants $C > 0$ and $K > 0$ depending only on $c_5$, $C_5$, $c_6$ and $C_6$ such that 
\begin{equation}\label{eq:propA1:1}
\sup_{t \in \R} \big| \pr \big( \critstar \le t \big) - \pr \big( \critgauss \le t \big) \big| < C n^{-K}. 
\end{equation}
By Lemma \ref{lemmaA8}, it holds that $\pr(\critgauss \le \gamma_\alpha^\Wgauss) = 1-\alpha$. Using this identity together with \eqref{eq:propA1:1} yields that
\[  1 - (\alpha + C n^{-K}) < \pr \big( \critstar \le \gamma_\alpha^\Wgauss \big) < 1 - (\alpha - C n^{-K}), \]
which in turn implies that $\gamma_{\alpha+C n^{-K}}^* \le \gamma_{\alpha}^{\Wgauss} \le \gamma_{\alpha-C n^{-K}}^*$ for any $\alpha \in (Cn^{-K}, 1 - Cn^{-K})$. This is the first statement of Proposition \ref{propA1}. The second statement is an immediate consequence thereof. 
\end{proof}

\begin{propA}\label{propA2}
There exist positive constants $C$ and $K$ that depend only on the model parameters $\Theta$ such that on the event $\mathcal{A}_n$, 
\begin{align*}
 & \gamma_{\alpha+\xi_n}^* \le \gamma_\alpha \le \gamma_{\alpha-\xi_n}^* \\ 
 & \gamma_{\alpha+\xi_n} \le \gamma_\alpha^* \le \gamma_{\alpha-\xi_n} 
\end{align*}
for any $\alpha \in (\xi_n,1-\xi_n)$ with $\xi_n = C n^{-K}$.
\end{propA}
\begin{proof}[\textnormal{\textbf{Proof of Proposition \ref{propA2}.}}] 
Conditionally on $\Xmat$ and $\eps$, $\W(e)$ is a Gaussian random vector with the covariance matrix $\Sigma = (\Sigma_{jk}: 1 \le j,k \le p)$, where $\Sigma_{jk} = n^{-1} \sum\nolimits_{i=1}^n \Xvec_{ij} \Xvec_{ik} \eps_i^2$. Moreover, by definition, $\Wgauss$ is a Gaussian random vector with the covariance matrix $\Sigma^* = (\Sigma_{jk}^*: 1 \le j,k \le p)$, where $\Sigma_{jk}^* = n^{-1} \sum\nolimits_{i=1}^n \ex[\Xvec_{ij} \Xvec_{ik} \eps_i^2]$. It is straightforward to verify that under \ref{C2} and \ref{C3}, $c_4 \le \Sigma_{jj}^* \le C_4$ with some constants $0 < c_4 \le C_4 < \infty$ that depend only on the parameters $\Theta$. Hence, by Lemma \ref{lemmaA6}, the distribution of $\critgauss = \max_{1 \le j \le p} |\Wgauss_j|$ is close to the conditional distribution of $\crit(e) = \max_{1 \le j \le p} |\W_j(e)|$ in the following sense: 
\begin{equation}\label{eq:propA2:1}
\sup_{t \in \R} \big| \pr_e \big( \crit(e) \le t \big) - \pr \big( \critgauss \le t \big) \big| \le \pi(\Delta),  
\end{equation}
where $\pi(\Delta) = C \Delta^{1/3} \{ 1 \lor 2 \log (2p) \lor \log(1/\Delta) \}^{1/3} \{\log (2p)\}^{1/3}$ with $C$ depending only on $c_4$ and $C_4$. Notice that the logarithm in the expression $\pi(\Delta)$ takes the argument $2p$ rather than $p$ as in the formulation of Lemma \ref{lemmaA6}. This is due to the fact that $\crit(e)$ and $\critgauss$ are maxima over absolute values as discussed in Remark \ref{remarklemmasA5678}. From \eqref{eq:propA2:1}, it immediately follows that on the event $\mathcal{A}_n$, 
\begin{equation}\label{eq:propA2:2}
\sup_{t \in \R} \big| \pr_e \big( \crit(e) \le t \big) - \pr \big( \critgauss \le t \big) \big| < \pi_n,  
\end{equation}
where we let $\pi_n$ be such that $\pi(B_\Delta \sqrt{\log(n \lor p)/n}) < \pi_n \le C n^{-K}$ with some positive constants $C$ and $K$. With the help of \eqref{eq:propA2:2} and analogous arguments as in the proof of Proposition \ref{propA1}, we can infer that on the event $\mathcal{A}_n$, 
\begin{equation*}
\begin{split} 
 & \gamma_{\alpha+\pi_n} \le \gamma_{\alpha}^{\Wgauss} \le \gamma_{\alpha-\pi_n} \\
 & \gamma_{\alpha+\pi_n}^{\Wgauss} \le \gamma_{\alpha} \le \gamma_{\alpha-\pi_n}^{\Wgauss}. 
\end{split}
\end{equation*}
Combining this with the statement of Proposition \ref{propA1}, we finally get that on $\mathcal{A}_n$,
\begin{align*}
 & \gamma_{\alpha+\xi_n}^* \le \gamma_{\alpha} \le \gamma_{\alpha-\xi_n}^* \\
 & \gamma_{\alpha+\xi_n} \le \gamma_{\alpha}^* \le \gamma_{\alpha-\xi_n} 
\end{align*} 
with $\xi_n = \kappa_n + \pi_n$, which completes the proof.
\end{proof}

We now turn to the main part of the proof of Theorem \ref{theo1}. To make the notation more compact, we introduce the shorthands $\rho_n = B_\remainder (\log n)^2 \{ C_\beta n^{1/2-\delta_\beta} \}^{1/2} / n^{1/4}$ and $\psi_n = C_{\Xvec} C_\sigma [ \sqrt{2 \log (2 p)} + \sqrt{2\log (n \lor p)} ]$. Moreover, we let $\{ \nu_n \}$ be any null sequence with 
\begin{equation} 
\nu_n > \xi_n + C_\remainder n^{-K_\remainder} + C_7 \max \biggl\{ g_n \sqrt{1 \lor \log \biggl( \frac{2p}{g_n}\biggr)}, h_n \sqrt{1 \lor \log \biggl( \frac{2p}{h_n}\biggr)} \biggr\}, \label{eq:nu}
\end{equation}
where $g_n = \rho_n(1+\psi_n)$, $h_n = \{ \rho_n + \rho_n(1+\rho_n) \} \psi_n + \rho_n$, and $C_7$ is a positive constant that depends only on the parameters $\Theta$ and that is specified below.

Our aim is to prove that on the event $\mathcal{T}_{\lambda_{\alpha+\nu_n}^*} \cap \mathcal{A}_n$, $\lambda_{\alpha+\nu_n}^* \le \hat{\lambda}_\alpha \le \lambda_{\alpha-\nu_n}^*$. This is equivalent to the following statement: on the event $\mathcal{S}_{\gamma_{\alpha+\nu_n}^*} \cap \mathcal{A}_n$, it holds that $\gamma_{\alpha+\nu_n}^* \le \hat{\gamma}_\alpha \le \gamma_{\alpha-\nu_n}^*$. By definition of $\hat{\gamma}_\alpha$, 
\begin{align}
\hat{\pi}_\alpha(\gamma) \le \gamma \text{ for all } \gamma \ge \gamma_{\alpha-\nu_n}^* & \quad \Longrightarrow \quad \hat{\gamma}_\alpha \le \gamma_{\alpha-\nu_n}^* \label{eq:theo1:rmk1} \\
\hat{\pi}_{\alpha}(\gamma) > \gamma \text{ for some } \gamma > \gamma_{\alpha+\nu_n}^* & \quad \Longrightarrow \quad \hat{\gamma}_\alpha > \gamma_{\alpha+\nu_n}^*, \label{eq:theo1:rmk2}
\end{align}
and by definition of $\hat{\pi}_\alpha(\gamma)$, 
\begin{align}
\pr_e( \crithat(\gamma,e) \le \gamma) > 1-\alpha & \quad \Longrightarrow \quad \hat{\pi}_\alpha(\gamma) \le \gamma \label{eq:theo1:rmk3} \\
\pr_e( \crithat(\gamma,e) \le \gamma) < 1-\alpha & \quad \Longrightarrow \quad \hat{\pi}_\alpha(\gamma) > \gamma. \label{eq:theo1:rmk4}
\end{align}
Hence, it suffices to prove the following two statements: 
\begin{enumerate}[label=(\Roman*), leftmargin=0.8cm]
\item \label{claim1} On the event $\mathcal{S}_{\gamma_{\alpha+\nu_n}^*} \cap \mathcal{A}_n$, $\pr_e( \crithat(\gamma,e) \le \gamma) > 1-\alpha$ for all $\gamma \ge \gamma_{\alpha-\nu_n}^*$. 
\item \label{claim2} On the event $\mathcal{S}_{\gamma_{\alpha+\nu_n}^*} \cap \mathcal{A}_n$, $\pr_e( \crithat(\gamma,e) \le \gamma) < 1-\alpha$ for some $\gamma > \gamma_{\alpha+\nu_n}^*$. 
\end{enumerate}

\pagebreak

\begin{proof}[\textnormal{\textbf{Proof of \ref{claim1}.}}\nopunct]  
Suppose we are on the event $\mathcal{S}_{\gamma_{\alpha+\nu_n}^*} \cap \mathcal{A}_n$ and let $\gamma \ge \gamma_{\alpha-\nu_n}^*$.  Using the simple bound \eqref{eq:critbound}, we obtain that
\begin{align}
\pr_e \big( \crithat(\gamma,e) \le \gamma \big) 
 & \ge \pr_e \big( \crit(e) + \remainder(\gamma,e) \le \gamma \big) \nonumber \\
 & \ge \pr_e \big( \crit(e) + \remainder(\gamma,e) \le \gamma, \remainder(\gamma,e) \le \rho_n \sqrt{\gamma} \big) \nonumber \\
 & \ge \pr_e \big( \crit(e) + \rho_n \sqrt{\gamma} \le \gamma, \remainder(\gamma,e) \le \rho_n \sqrt{\gamma} \big) \nonumber \\
 & \ge \pr_e \big( \crit(e) \le \gamma - \rho_n \sqrt{\gamma} \big) - \pr_e \big( \remainder(\gamma,e) > \rho_n \sqrt{\gamma} \big) \nonumber \\
 & \ge \pr_e \big( \crit(e) \le \gamma - \rho_n \sqrt{\gamma} \big) - C_\remainder n^{-K_\remainder}, \label{eq:claim1:eq1}
\end{align} 
where the last inequality is by Lemma \ref{lemmaA2}. Since $\gamma - \rho_n \sqrt{\gamma} \ge \gamma - \rho_n (1 + \gamma) = (1- \rho_n) \gamma - \rho_n$ and $\gamma \ge \gamma_{\alpha-\nu_n}^*$, we further get that
\begin{align*}
\pr_e \big( \crit(e) \le \gamma - \rho_n \sqrt{\gamma} \big) 
 & \ge \pr_e \big( \crit(e) \le (1- \rho_n) \gamma - \rho_n \big) \\
 & \ge \pr_e \big( \crit(e) \le (1- \rho_n) \gamma_{\alpha-\nu_n}^* - \rho_n \big) \\
 & = \pr_e \big( \crit(e) \le \gamma_{\alpha-\nu_n}^* - \rho_n (1 + \gamma_{\alpha-\nu_n}^*) \big). 
\end{align*} 
Moreover, since $\gamma_{\alpha-\nu_n}^* \ge \gamma_{\alpha+\xi_n-\nu_n}$ on the event $\mathcal{A}_n$ by Proposition \ref{propA2} and $\gamma_{\alpha-\nu_n}^* \le \gamma_{\alpha-\kappa_n-\nu_n}^\Wgauss \le \psi_n$ by Proposition \ref{propA1} and Lemma \ref{lemmaA3}, it follows that 
\begin{align}
\pr_e \big( \crit(e) \le \gamma - \rho_n \sqrt{\gamma} \big) \nonumber 
 & \ge \pr_e \big( \crit(e) \le \gamma_{\alpha+\xi_n-\nu_n} -  \rho_n (1 + \psi_n) \big) \nonumber \\
 & = \pr_e \big( \crit(e) \le \gamma_{\alpha+\xi_n-\nu_n} \big) \nonumber \\ & \quad - \pr_e \big( \gamma_{\alpha+\xi_n-\nu_n} - \rho_n (1 + \psi_n) < \crit(e) \le \gamma_{\alpha+\xi_n-\nu_n} \big). \label{eq:claim1:eq2}
\end{align}
On the event $\mathcal{A}_n$, we have that $c_\Xvec^2 c_\sigma^2 - B_\Delta \sqrt{\log(n \lor p) / n} \le \ex_e[\W_j^2(e)] \le C_\Xvec^2 C_\sigma^2 + B_\Delta \sqrt{\log(n \lor p) / n}$. Hence, we can apply Lemma \ref{lemmaA5} to get that 
\begin{align}
\pr_e & \big( \gamma_{\alpha+\xi_n-\nu_n} -  \rho_n (1 + \psi_n) < \crit(e) \le \gamma_{\alpha+\xi_n-\nu_n} \big) \nonumber \\
 & \le \sup_{t \in \R} \pr_e \big( |\crit(e)-t| \le  \rho_n (1 + \psi_n) \big) \nonumber \\
 & \le C_7  \rho_n (1 + \psi_n) \sqrt{1 \lor \log \biggl( \frac{2p}{ \rho_n (1 + \psi_n)} \biggr)} \label{eq:claim1:eq3}
\end{align} 
with $C_7$ depending only on $\Theta$. By Lemma \ref{lemmaA8}, it further holds that $\pr_e (\crit(e) \le \gamma_{\alpha+\xi_n-\nu_n}) = 1 - (\alpha+\xi_n-\nu_n)$. Plugging this identity and \eqref{eq:claim1:eq3} into \eqref{eq:claim1:eq2} yields that 
\begin{equation*}
\pr_e \big( \crit(e) \le \gamma - \rho_n \sqrt{\gamma} \big) \ge 1 - \alpha + \nu_n - \xi_n - C_7 g_n \sqrt{1 \lor \log \biggl( \frac{2p}{g_n} \biggr)}
\end{equation*}
with $g_n = \rho_n (1 + \psi_n)$. Inserting this into \eqref{eq:claim1:eq1}, we finally arrive at 
\begin{equation*}
\pr_e \big( \crithat(\gamma,e) \le \gamma \big) \ge 1 - \alpha + \nu_n - \xi_n - C_\remainder n^{-K_\remainder} - C_7 g_n \sqrt{1 \lor \log \biggl( \frac{2p}{g_n} \biggr)} > 1 - \alpha,
\end{equation*}
where the last inequality follows from the definition of $\nu_n$ in \eqref{eq:nu}. 
\end{proof}

\begin{proof}[\textnormal{\textbf{Proof of \ref{claim2}.}}\nopunct]
Suppose we are on the event $\mathcal{S}_{\gamma_{\alpha+\nu_n}^*} \cap \mathcal{A}_n$ and let $\gamma = (1+\phi_n) \gamma_{\alpha+\nu_n}^*$, where $\{\phi_n\}$ is a null sequence of positive numbers with $\phi_n \le Cn^{-K}$ for some constants $C$ and $K$. For convenience, we set $\phi_n = \rho_n$, but we could also work with any other choice of $\phi_n$ that satisfies the conditions mentioned in the previous sentence. With the bound \eqref{eq:critbound}, we get that
\begin{align}
\pr_e \big( \crithat(\gamma,e) > \gamma \big) 
 & \ge \pr_e \big( \crit(e) - \remainder(\gamma,e) > \gamma \big) \nonumber \\
 & \ge \pr_e \big( \crit(e) - \remainder(\gamma,e) > \gamma, \remainder(\gamma,e) \le \rho_n \sqrt{\gamma} \big) \nonumber \\
 & \ge \pr_e \big( \crit(e) - \rho_n \sqrt{\gamma} > \gamma, \remainder(\gamma,e) \le \rho_n \sqrt{\gamma} \big) \nonumber \\
 & \ge \pr_e \big( \crit(e) > \gamma + \rho_n \sqrt{\gamma} \big) - \pr_e \big( \remainder(\gamma,e) > \rho_n \sqrt{\gamma} \big) \nonumber \\
 & \ge \pr_e \big( \crit(e) > \gamma + \rho_n \sqrt{\gamma} \big) - C_\remainder n^{-K_\remainder}, \label{eq:claim2:eq1}
\end{align} 
where the final inequality is a direct consequence of Lemma \ref{lemmaA2}. Analogous arguments as in the proof of \ref{claim1} yield that 
\begin{align*}
 & \pr_e \big( \crit(e) > \gamma + \rho_n \sqrt{\gamma} \big) \\
 & \ge \pr_e \big( \crit(e) > \gamma + \rho_n (1 + \gamma) \big) \\ 
 & = \pr_e \big( \crit(e) > (1+\phi_n) \gamma_{\alpha+\nu_n}^* + \rho_n (1 + (1+\phi_n) \gamma_{\alpha+\nu_n}^*) \big) \\
 & = \pr_e \big( \crit(e) > \gamma_{\alpha+\nu_n}^* + \{ \phi_n + \rho_n (1+\phi_n) \} \gamma_{\alpha+\nu_n}^* + \rho_n \big) \\
 & \ge \pr_e \big( \crit(e) > \gamma_{\alpha-\xi_n+\nu_n} + h_n \big) \\
 & = \pr_e \big( \crit(e) > \gamma_{\alpha-\xi_n+\nu_n} \big) - \pr_e \big( \gamma_{\alpha-\xi_n+\nu_n} < \crit(e) \le \gamma_{\alpha-\xi_n+\nu_n} + h_n \big) \\
 & \ge \alpha + \nu_n - \xi_n - C_7 h_n \sqrt{1 \lor \log (2p / h_n)},
\end{align*}
where $h_n = \{ \phi_n + \rho_n (1+\phi_n) \} \psi_n + \rho_n = \{ \rho_n + \rho_n (1+\rho_n) \} \psi_n + \rho_n$ under the assumption that $\phi_n = \rho_n$. Inserting this into \eqref{eq:claim2:eq1}, we arrive that 
\begin{equation*}
\pr_e \big( \crithat(\gamma,e) > \gamma \big) \ge \alpha + \nu_n - \xi_n - C_\remainder n^{-K_\remainder} - C_7 h_n \sqrt{1 \lor \log \biggl( \frac{2p}{h_n} \biggr)} > \alpha, 
\end{equation*}
which is equivalent to the statement that $\pr_e \big( \crithat(\gamma,e) \le \gamma \big) < 1-\alpha$. 
\end{proof}

\subsection*{Proof of Proposition \ref{prop:tuning:1}}

From \eqref{eq:predictionbound}, it follows that $\normtwos{\Xmat (\beta^* - \hat{\beta}_\lambda)}/n \leq 2 \lambda \normone{\beta^*}$ for every $\lambda \ge \lambda_{\alpha+\nu_n}^*$ on the event $\mathcal{T}_{\lambda_{\alpha+\nu_n}^*}$. Moreover, by Theorem \ref{theo1}, $\lambda_{\alpha + \nu_n}^* \le \hat{\lambda}_{\alpha} \le \lambda_{\alpha-\nu_n}^*$ on the event $\mathcal{T}_{\lambda_{\alpha+\nu_n}^*} \cap \mathcal{A}_n$. Hence, we can infer that 
\begin{equation*}
\frac{1}{n} \normtwos{\Xmat (\beta^* - \hat{\beta}_{\hat{\lambda}_\alpha})} \leq 2 \hat{\lambda}_\alpha \normone{\beta^*} \le 2 \lambda_{\alpha-\nu_n}^* \normone{\beta^*} 
\end{equation*}
on the event $\mathcal{T}_{\lambda_{\alpha+\nu_n}^*} \cap \mathcal{A}_n$, which occurs with probability $\pr(\mathcal{T}_{\lambda_{\alpha+\nu_n}^*} \cap \mathcal{A}_n) \ge 1 - \alpha - \nu_n - C_1 n^{-K_1}$.

\subsection*{Proof of Proposition \ref{prop:tuning:2}}

From Lemma \ref{lemma:oracleinequality}, we know that on the event $\mathcal{T}_{\lambda_{\alpha+\nu_n}^*} \cap \mathcal{B}_n$, $\normsup{\hat{\beta}_\lambda - \beta^*} \le \kappa \lambda$ for every $\lambda \ge (1+\delta) \lambda_{\alpha+\nu_n}^*$.  Moreover, on the event $\mathcal{T}_{\lambda_{\alpha+\nu_n}^*} \cap \mathcal{A}_n$, it holds that $\lambda_{\alpha + \nu_n}^* \le \hat{\lambda}_{\alpha} \le \lambda_{\alpha-\nu_n}^*$ by Theorem \ref{theo1}. Hence, we can infer that 
\begin{equation*}
\normsup{\hat{\beta}_{(1+\delta)\hat{\lambda}_\alpha} - \beta^*} \le (1+\delta) \kappa \hat{\lambda}_\alpha \le (1+\delta) \kappa \lambda_{\alpha-\nu_n}^*
\end{equation*}
on the event $\mathcal{T}_{\lambda_{\alpha+\nu_n}^*} \cap \mathcal{A}_n \cap \mathcal{B}_n$, which occurs with probability $\pr(\mathcal{T}_{\lambda_{\alpha+\nu_n}^*} \cap \mathcal{A}_n \cap \mathcal{B}_n) \ge 1 - \alpha - \pr(\mathcal{B}_n^\complement) - \nu_n - C_1 n^{-K_1}$.

\subsection*{Proof of Proposition \ref{prop:testing:1}}

For the proof, we reformulate the test of $H_0$ as follows: slightly abusing notation, we redefine the test statistic as $T = \normsup{\Xmat^\top \Y}/\sqrt{n}$. As above, we further let $\gamma_\alpha^* = \sqrt{n} \lambda_\alpha^*/2$ be the $(1-\alpha)$-quantile of $\normsup{\Xmat^\top \eps}/\sqrt{n}$ and define the estimator $\hat{\gamma}_\alpha = \sqrt{n} \hat{\lambda}_\alpha/2$. Our test of $H_0$ can now be expressed as follows: reject $H_0$ at the significance level $\alpha$ if $T > \hat{\gamma}_\alpha$.

We first prove that $\pr(T \le \hat{\gamma}_\alpha) \ge 1 - \alpha + o(1)$ under $H_0$. With the help of Theorem \ref{theo1}, we obtain that under $H_0$,  
\begin{align*}
\pr (T \le \hat{\gamma}_\alpha) 
 & \ge \pr (T \le \hat{\gamma}_\alpha, \mathcal{T}_{\lambda_{\alpha+\nu_n}^*} \cap \mathcal{A}_n) 
   \ge \pr (T \le \gamma_{\alpha+\nu_n}^*, \mathcal{T}_{\lambda_{\alpha+\nu_n}^*} \cap \mathcal{A}_n) \\
 & = \pr (\mathcal{T}_{\lambda_{\alpha+\nu_n}^*} \cap \mathcal{A}_n) \ge 1 - \alpha - \nu_n - C_1 n^{-K_1},
\end{align*}
where the equality in the last line follows from the fact that the two events $\mathcal{T}_{\lambda_{\alpha+\nu_n}^*}$ and $\{T \le \gamma_{\alpha+\nu_n}^*\}$ are identical. As a result, we get that $\pr (T \le \hat{\gamma}_\alpha) \ge 1 - \alpha + o(1)$ under $H_0$.

We next prove that $\pr(T > \hat{\gamma}_\alpha) = 1 - o(1)$ under any alternative $H_1: \beta^* \ne 0$ that fulfills the conditions of Proposition \ref{prop:testing:1}. Suppose we are on such an alternative and let $\{ \alpha_n \}$ be a null sequence with $2 \nu_n + (n \lor p)^{-1} < \alpha_n < \alpha$. It holds that
\begin{align}
\pr (T > \hat{\gamma}_\alpha) \ge \pr (T > \hat{\gamma}_{\alpha_n}) 
 & = \pr \Big( \frac{1}{\sqrt{n}} \normsup{\Xmat^\top\Xmat \beta^* + \Xmat^\top \eps} > \hat{\gamma}_{\alpha_n} \Big) \nonumber \\
 & \ge \pr \Big( \frac{\normsup{\Xmat^\top\Xmat \beta^*}}{\sqrt{n}} > \hat{\gamma}_{\alpha_n} + \frac{\normsup{\Xmat^\top \eps}}{\sqrt{n}} \Big), \label{eq:prop:testing:1:alt1} 
\end{align}
where the last line is due to the triangle inequality. Applying Theorem \ref{theo1}, Proposition \ref{propA1} and Lemma \ref{lemmaA3}, the term $\hat{\gamma}_{\alpha_n}$ can be bounded by 
\begin{equation}\label{eq:prop:testing:1:alt2} 
\hat{\gamma}_{\alpha_n} \le \gamma_{\alpha_n-\nu_n}^* \le \gamma_{\alpha_n-\nu_n-\kappa_n}^\Wgauss \le \psi_n
\end{equation}
on the event $\mathcal{T}_{\lambda_{\alpha_n+\nu_n}^*} \cap \mathcal{A}_n$, where $\psi_n = C_{\Xvec} C_\sigma [ \sqrt{2 \log (2 p)} + \sqrt{2\log (n \lor p)} ]$. Moreover, for the term $\critstar = \normsup{\Xmat^\top\eps}/\sqrt{n}$, we have that
\begin{align} 
\pr \Big( \frac{\normsup{\Xmat^\top \eps}}{\sqrt{n}} \le \psi_n \Big) = \pr (\critstar \le \psi_n) 
 & = \pr (\critgauss \le \psi_n) + \big[ \pr (\critstar \le \psi_n) - \pr (\critgauss \le \psi_n) \big] \nonumber \\
 & \ge \pr (\critgauss \le \psi_n) - Cn^{-K} \nonumber \\
 & \ge \pr (\critgauss \le \gamma_{2/(n \lor p)}^\Wgauss) - Cn^{-K} \nonumber \\
 & = 1 - \frac{2}{n \lor p} - Cn^{-K} = 1 - o(1) \label{eq:prop:testing:1:alt3}
\end{align}
with some positive constants $C$ and $K$, where the first inequality follows from \eqref{eq:propA1:1} and the second one from Lemma \ref{lemmaA3}. Using \eqref{eq:prop:testing:1:alt2} and \eqref{eq:prop:testing:1:alt3} in the right-hand side of equation \eqref{eq:prop:testing:1:alt1}, we can infer that 
\begin{align}
 & \pr \Big( \frac{\normsup{\Xmat^\top\Xmat \beta^*}}{\sqrt{n}} > \hat{\gamma}_{\alpha_n} + \frac{\normsup{\Xmat^\top \eps}}{\sqrt{n}} \Big) \nonumber \\
 & \ge \pr \Big( \frac{\normsup{\Xmat^\top\Xmat \beta^*}}{\sqrt{n}} > \hat{\gamma}_{\alpha_n} + \frac{\normsup{\Xmat^\top \eps}}{\sqrt{n}}, \, \frac{\normsup{\Xmat^\top \eps}}{\sqrt{n}} \le \psi_n, \, \mathcal{T}_{\lambda_{\alpha_n+\nu_n}^*} \cap \mathcal{A}_n \Big) \nonumber \\
 & \ge \pr \Big( \frac{\normsup{\Xmat^\top\Xmat \beta^*}}{\sqrt{n}} > 2 \psi_n, \, \frac{\normsup{\Xmat^\top \eps}}{\sqrt{n}} \le \psi_n, \, \mathcal{T}_{\lambda_{\alpha_n+\nu_n}^*} \cap \mathcal{A}_n \Big) \nonumber \\
 & = \pr \Big( \frac{\normsup{\Xmat^\top\Xmat \beta^*}}{\sqrt{n}} > 2 \psi_n \Big) - o(1) = 1 - o(1), \label{eq:prop:testing:1:alt4} 
\end{align}
the last equality following from the assumption that $\pr(\normsup{\Xmat^\top\Xmat\beta^*}/n \linebreak \ge c \sqrt{\log(n \lor p)/n}) \rightarrow 1$ for every $c > 0$. Combining \eqref{eq:prop:testing:1:alt4} with \eqref{eq:prop:testing:1:alt1} yields that $\pr (T > \hat{\gamma}_\alpha) = 1 - o(1)$ under the alternative, which completes the proof.

\subsection*{Proof of Proposition \ref{prop:testing:2}}

Similarly as in the proof of Proposition \ref{prop:testing:1}, we reformulate the test of $H_{0,B}: \beta^*_B = 0$. Slightly abusing notation, we redefine the test statistic as 
\[ T_B = \frac{\normsup{(\mathcal{P} \Xmat_B)^\top \mathcal{P} \Y}}{\sqrt{n}}. \] 
Moreover, we let $\gamma_{\alpha,B}^* = \sqrt{n} \lambda_{\alpha,B}^*/2$ be the $(1-\alpha)$-quantile of $\normsup{(\mathcal{P} \Xmat_B)^\top \uvec}/\sqrt{n}$ and set $\hat{\gamma}_{\alpha,B} = \sqrt{n} \hat{\lambda}_{\alpha,B}/2$. Our test of $H_{0,B}$ can now be formulated as follows: reject $H_{0,B}$ at the significance level $\alpha$ if $T_B > \hat{\gamma}_{\alpha,B}$. This test has the same structure as the test of the simpler hypothesis $H_0: \beta^* = 0$. The only difference is that it is based on the transformed model $\mathcal{P} \Y = \mathcal{P} \Xmat_B \beta^*_B + \uvec$ rather than on the original model $\Y = \Xmat \beta^* + \eps$. Even though a minor detail at first sight, this change of model brings about some technical complications. The issue is that the entries of the noise vector $\uvec$ are in general not independent, whereas those of $\eps$ are. Similarly, the rows of the design matrix $\mathcal{P} \Xmat_B$ are in general not independent in contrast to those of $\Xmat$. As a consequence, the central result of our theory, Theorem \ref{theo1}, cannot be applied to the estimator $\hat{\gamma}_{\alpha,B}$ directly. To adapt it to the present situation, we define the event 
\[ \mathcal{S}_\gamma^\prime = \Big\{ \frac{1}{\sqrt{n}} \normsup{(\mathcal{P} \Xmat_B)^\top \uvec} \le \gamma \Big\} \]
and let $C_1^\prime$, $K_1^\prime$,  $C_2^\prime$ and $K_2^\prime$ be positive constants that depend only on the model parameters $\Theta^\prime = \Theta \cup \{ c_\vartheta, |A|, \normtwo{\Psi_A^{-1}} \}$. With this notation at hand, we can prove the following. 
\begin{propA}\label{propA3}
There exist an event $\mathcal{A}_n^\prime$ with $\pr(\mathcal{A}_n^\prime) \ge 1 - C_1^\prime n^{-K_1^\prime}$ for some positive constants $C_1^\prime$ and $K_1^\prime$ and a sequence of real numbers $\nu_n^\prime$ with $0 < \nu_n^\prime \le C_2^\prime n^{-K_2^\prime}$ for some positive constants $C_2^\prime$ and $K_2^\prime$ such that the following holds: on the event $\mathcal{S}_{\gamma_{\alpha+\nu_n^\prime,B}^*}^\prime \cap \mathcal{A}_n^\prime$, 
\[ \gamma_{\alpha + \nu_n^\prime,B}^* \le \hat{\gamma}_{\alpha,B} \le \gamma_{\alpha-\nu_n^\prime,B}^* \]
for any $\alpha \in (a_n, 1-a_n)$ with $a_n = 2 \nu_n^\prime + (n \lor p)^{-1}$. 
\end{propA} 
\noindent The overall strategy to prove Proposition \ref{propA3} is the same as the one for Theorem \ref{theo1}. There are some complications, however, that stem from the fact that the entries of $\uvec$ and the rows of $\mathcal{P} \Xmat_B$ are not independent. We provide the proof of Proposition \ref{propA3} in the Supplementary Material, where we highlight the main differences to the proof of Theorem \ref{theo1}.

With Proposition \ref{propA3} in place, the proof of Proposition \ref{prop:testing:2} proceeds analogously to the one of Proposition \ref{prop:testing:1}. For this reason, we only give a brief summary of the main steps. First suppose that the null hypothesis $H_{0,B}$ holds true. With the help of Proposition \ref{propA3}, we get that  
\begin{align*}
\pr (T_B \le \hat{\gamma}_{\alpha,B}) 
 & \ge \pr \big( T_B \le \hat{\gamma}_{\alpha,B}, \mathcal{S}_{\gamma_{\alpha+\nu_n^\prime,B}^*}^\prime \cap \mathcal{A}_n^\prime \big) \\
 & \ge \pr \big( T_B \le \gamma_{\alpha+\nu_n^\prime,B}^*, \mathcal{S}_{\gamma_{\alpha+\nu_n^\prime,B}^*}^\prime \cap \mathcal{A}_n^\prime \big) \\
 & = \pr \big( \mathcal{S}_{\gamma_{\alpha+\nu_n^\prime,B}^*}^\prime \cap \mathcal{A}_n^\prime \big) \ge 1 - \alpha - \nu_n^\prime - C_1^\prime n^{-K_1^\prime},
\end{align*}
which implies that $\pr (T_B \le \hat{\gamma}_{\alpha,B}) \ge 1 - \alpha + o(1)$ under $H_{0,B}$.

Next suppose we are on an alternative $H_{1,B}: \beta^*_B \ne 0$ that satisfies the conditions of Proposition \ref{prop:testing:2} and let $\{ \alpha_n \}$ be a null sequence with $2 \nu_n^\prime + (n \lor p)^{-1} < \alpha_n < \alpha$. Similarly as in the proof of Proposition \ref{prop:testing:1}, we can establish the bound 
\begin{equation}\label{eq:prop:testing:2:alt1} 
\pr (T_B > \hat{\gamma}_{\alpha,B}) \ge \pr \Big( \frac{\normsup{(\mathcal{P}\Xmat_B)^\top \mathcal{P} \Xmat_B \beta^*_B}}{\sqrt{n}} > \hat{\gamma}_{\alpha_n,B} + \frac{\normsup{(\mathcal{P}\Xmat_B)^\top \uvec}}{\sqrt{n}} \Big) 
\end{equation}
and verify the following: (i) $\hat{\gamma}_{\alpha_n,B} \le \psi_n^\prime$ on the event $\mathcal{S}_{\gamma_{\alpha_n+\nu_n^\prime,B}^*}^\prime \cap \mathcal{A}_n^\prime$, where $\psi_n^\prime = C [ \sqrt{2 \log (2 p)} + \sqrt{2\log (n \lor p)} ]$ with some sufficiently large constant $C$ that depends only on $\Theta^\prime$, and (ii) $\pr( \normsup{(\mathcal{P} \Xmat_B)^\top \uvec} / \sqrt{n} \le \psi_n^\prime) = 1 - o(1)$. Applying (i) and (ii) to the right-hand side of \eqref{eq:prop:testing:2:alt1} yields that  
\begin{align*}
\pr \Big( & \frac{\normsup{(\mathcal{P}\Xmat_B)^\top \mathcal{P} \Xmat_B \beta^*_B}}{\sqrt{n}} > \hat{\gamma}_{\alpha_n,B} + \frac{\normsup{(\mathcal{P}\Xmat_B)^\top \uvec}}{\sqrt{n}} \Big) \\*
 & \ge \pr \Big( \frac{\normsup{(\mathcal{P} \Xmat_B)^\top \mathcal{P} \Xmat_B \beta^*_B}}{\sqrt{n}} > 2 \psi_n^\prime \Big) - o(1) = 1 - o(1), 
\end{align*}
which in turn implies that $\pr (T_B > \hat{\gamma}_{\alpha,B}) = 1 - o(1)$ under the alternative.

\bibliographystyle{ims}
{\small
\setlength{\bibsep}{0.55em}
\bibliography{bibliography}}

\newpage
\allowdisplaybreaks[3]

\headingsupp{Supplementary Material to}{``Estimating the Lasso's Effective Noise''}
\authors{Johannes Lederer}{Ruhr-University Bochum}{Michael Vogt}{Ulm University}
\setlength{\parindent}{15pt}
\def\theequation{S.\arabic{equation}}
\def\thesection{S.\arabic{section}}
\def\thefigure{S.\arabic{figure}}
\def\thetable{S.\arabic{table}}
\setcounter{equation}{0}
\setcounter{section}{0}
\setcounter{figure}{0}
\setcounter{table}{0}

\section{Technical details}

In what follows, we provide the technical details and proofs that are omitted in the paper.

\subsection*{Proof of Lemma \ref{lemma:oracleinequality}}

To show the result, we slightly generalize the proof of Lemma 5 in \cite{Chichignoud16}. Standard arguments from the lasso literature \citep{Buhlmann11} show that on the event $\mathcal{T}_\lambda$, 
\[ \normone{\hat{\beta}_{\lambda^\prime,S^\complement} - \beta_{S^\complement}^*} \le \frac{2+\delta}{\delta} \normone{\hat{\beta}_{\lambda^\prime,S} - \beta_S^*}, \]
that is, $\hat{\beta}_{\lambda^\prime} - \beta^* \in \mathbb{C}_\delta(S)$ for every $\lambda^\prime \ge (1+\delta)\lambda$. Under the $\ell_\infty$-restricted eigenvalue condition \eqref{eq:RE}, we thus obtain that on $\mathcal{T}_\lambda$, 
\begin{equation}\label{eq:lemma:oracleinequality:1}
\phi \normsup{\hat{\beta}_{\lambda^\prime} - \beta^*} \le \frac{\normsup{\Xmat^\top \Xmat (\hat{\beta}_{\lambda^\prime} - \beta^*)}}{n} 
\end{equation}
for every $\lambda^\prime \ge (1+\delta)\lambda$. Moreover, since the lasso satisfies the zero-subgradient condition $2 \Xmat^\top (\Xmat \hat{\beta}_{\lambda^\prime} - \Y)/n + \lambda^\prime \hat{z} = 0$ with $\hat{z} \in \Rp$ belonging to the subdifferential of the function $f(\beta) = \normone{\beta}$, it holds that
\[ \frac{2 \Xmat^\top \Xmat}{n} (\hat{\beta}_{\lambda^\prime} - \beta^*) = - \lambda^\prime \hat{z} + \frac{2 \Xmat^\top \eps}{n}. \]
Taking the supremum norm on both sides of this equation and taking into account that $2 \normsup{\Xmat^\top\eps}/n \le \lambda$ on the event $\mathcal{T}_\lambda$, we obtain that 
\begin{equation}\label{eq:lemma:oracleinequality:2} 
\frac{2 \normsup{\Xmat^\top \Xmat (\hat{\beta}_{\lambda^\prime} - \beta^*)}}{n} \le \lambda^\prime + \frac{2 \normsup{\Xmat^\top\eps}}{n} \le 2 \lambda^\prime 
\end{equation}
for every $\lambda^\prime \ge (1+\delta)\lambda$ on $\mathcal{T}_\lambda$. The statement of Lemma \ref{lemma:oracleinequality} follows upon combining \eqref{eq:lemma:oracleinequality:1} and \eqref{eq:lemma:oracleinequality:2}.

\subsection*{Proof of Lemma \ref{lemmaA1}}

Let $\delta$ be a small positive constant with $0 < \delta < (\theta-4)/\theta$ and $\theta > 4$ defined in \ref{C3}. Define $Z_{ijk} = \Xvec_{ij} \Xvec_{ik} \eps_i^2$ along with $Z_{ijk} = Z_{ijk}^\le + Z_{ijk}^>$, where
\begin{equation*}
Z_{ijk}^\le = Z_{ijk} \, \ind \big( |\eps_i| \le n^{\frac{1-\delta}{4}} \big) \qquad \text{and} \qquad Z_{ijk}^> = Z_{ijk} \, \ind \big( |\eps_i| > n^{\frac{1-\delta}{4}} \big),
\end{equation*}   
and write $\Delta \le \Delta^\le + \Delta^>$ with
\begin{align*}
\Delta^\le & = \max_{1 \le j,k \le p} \biggl| \frac{1}{n} \sum\limits_{i=1}^n ( Z_{ijk}^\le - \ex Z_{ijk}^\le ) \biggr| \\
\Delta^> & = \max_{1 \le j,k \le p} \biggl| \frac{1}{n} \sum\limits_{i=1}^n ( Z_{ijk}^> - \ex Z_{ijk}^> ) \biggr|. 
\end{align*}
In what follows, we prove that 
\begin{align}
\pr \big( \Delta^\le > B \sqrt{\log(n \lor p) / n} \big) & \le C n^{-K} \label{eq:lemma1:part1} \\
\pr \big( \Delta^> > B \sqrt{\log(n \lor p) / n} \big) & \le C n^{1 - (\frac{1-\delta}{4})\theta}, \label{eq:lemma1:part2}
\end{align}
where $B$, $C$ and $K$ are positive constants depending only on the parameters $\Theta$, and $K$ can be made as large as desired by choosing $B$ and $C$ large enough. Lemma \ref{lemmaA1} is a direct consequence of the two statements \eqref{eq:lemma1:part1} and \eqref{eq:lemma1:part2}.

We start with the proof of \eqref{eq:lemma1:part1}. A simple union bound yields that
\begin{equation}\label{eq:lemma1:part1:eq1}
\pr \big( \Delta^\le > B \sqrt{\log(n \lor p) / n} \big) \le \sum\limits_{j,k=1}^p P_{jk}^\le,
\end{equation}
where 
\[ P_{jk}^\le = \pr \biggl( \Big| \frac{1}{\sqrt{n}} \sum\limits_{i=1}^n U_{ijk} \Big| > B \sqrt{\log(n \lor p)} \biggr) \]
with $U_{ijk} = Z_{ijk}^\le - \ex Z_{ijk}^\le$. Using Markov's inequality, $P_{jk}^\le$ can be bounded by  
\begin{align}
P_{jk}^\le 
 & \le \exp \big( -\mu B \sqrt{\log(n \lor p)} \big) \ \ex \biggl[ \exp\biggl( \mu \Big| \frac{1}{\sqrt{n}} \sum\limits_{i=1}^n U_{ijk} \Big| \biggr) \biggr] \nonumber \\ 
 & \le \exp \big( -\mu B \sqrt{\log(n \lor p)} \big) \biggl\{ \ex \biggl[ \exp\biggl( \frac{\mu}{\sqrt{n}} \sum\limits_{i=1}^n U_{ijk} \biggr) \biggr] \nonumber \\*
 & \phantom{le \exp \biggl( -\mu B \sqrt{\frac{\log(n \lor p)}{n}} \biggr) \biggl\{ } + \ex \biggl[ \exp\biggl( -\frac{\mu}{\sqrt{n}} \sum\limits_{i=1}^n U_{ijk} \biggr) \biggr] \biggr\} \label{eq:lemma1:part1:eq2}
\end{align}
with an arbitrary constant $\mu > 0$. We now choose $\mu = \sqrt{\log(n \lor p)}/C_\mu$, where the constant $C_\mu > 0$ is picked so large that $\mu |U_{ijk}|/\sqrt{n} \le 1/2$ for all $n$. With this choice of $\mu$, we obtain that 
\begin{align*} 
\ex \biggl[ \exp\biggl( \pm \frac{\mu}{\sqrt{n}} \sum\limits_{i=1}^n U_{ijk} \biggr) \biggr] 
 & = \prod\limits_{i=1}^n \ex \biggl[ \exp \biggl( \pm \frac{\mu}{\sqrt{n}} U_{ijk} \biggr) \biggr] 
   \le \prod\limits_{i=1}^n \biggl( 1 + \frac{\mu^2}{n} \ex[U_{ijk}^2] \biggr) \\
 & \le \prod\limits_{i=1}^n \exp \biggl( \frac{\mu^2}{n} \ex[U_{ijk}^2] \biggr) \le \exp(C_U \mu^2), 
\end{align*}
where the first inequality follows from the fact that $\exp(x) \le 1 + x + x^2$ for $|x| \le 1/2$ and $C_U < \infty$ is an upper bound on $\ex[U_{ijk}^2]$. Plugging this into \eqref{eq:lemma1:part1:eq2} gives
\begin{align*}
P_{jk}^\le 
 & \le 2 \exp \big( -\mu B \sqrt{\log(n \lor p)} + C_U \mu^2 \big) \\
 & \le 2 \exp \biggl( - \Big\{ \frac{B}{C_\mu} - \frac{C_U}{C_\mu^2} \Big\} \log(n \lor p) \biggr) = 2 (n \lor p)^{\frac{C_U}{C_\mu^2}-\frac{B}{C_\mu}}. 
\end{align*}
Inserting this bound into \eqref{eq:lemma1:part1:eq1}, we finally obtain that 
\[ \pr \big( \Delta^\le > B \sqrt{\log(n \lor p) / n} \big) \le 2 p^2 (n \lor p)^{\frac{C_U}{C_\mu^2}-\frac{B}{C_\mu}} \le C n^{-K}, \]
where $K > 0$ can be chosen as large as desired by picking $B$ sufficiently large. This completes the proof of \eqref{eq:lemma1:part1}.

We next turn to the proof of \eqref{eq:lemma1:part2}. It holds that 
\[ \pr \big( \Delta^> > B \sqrt{\log(n \lor p) / n} \big) \le P_1^> + P_2^>, \]
where
\begin{align}
P_1^> := & \ \pr \biggl( \max_{1 \le j,k \le p} \Big| \frac{1}{n} \sum\limits_{i=1}^n Z_{ijk}^> \Big| > \frac{B}{2} \sqrt{\frac{\log(n \lor p)}{n}} \biggr) \nonumber \\
     \le & \ \pr \big( |\eps_i| > n^{\frac{1-\delta}{4}} \text{ for some } 1 \le i \le n \big) \nonumber \\
     \le & \ \sum\limits_{i=1}^n \pr \big( |\eps_i| > n^{\frac{1-\delta}{4}} \big) \le \sum\limits_{i=1}^n \ex \big[ |\eps_i|^\theta \big] / n^{(\frac{1-\delta}{4})\theta} \nonumber \\ 
     \le & \ C_\theta n^{1-(\frac{1-\delta}{4})\theta} \label{eq:lemma1:part2:eq1}
\end{align}
and 
\begin{equation}\label{eq:lemma1:part2:eq2}
P_2^> := \pr \biggl( \max_{1 \le j,k \le p} \Big| \frac{1}{n} \sum\limits_{i=1}^n \ex Z_{ijk}^> \Big| > \frac{B}{2} \sqrt{\frac{\log(n \lor p)}{n}} \biggr) = 0 
\end{equation}
for sufficiently large $n$, since 
\begin{align*}
\max_{1 \le j,k \le p} \Big| \frac{1}{n} \sum\limits_{i=1}^n \ex Z_{ijk}^> \Big| 
 & \le C_\Xvec^2 \max_{1 \le i \le n} \ex \Big[ \eps_i^2 \, \ind(|\eps_i| > n^{\frac{1-\delta}{4}}) \Big] \\
 & \le C_\Xvec^2 \max_{1 \le i \le n} \ex \Big[ |\eps_i|^\theta \big/ n^{\frac{(\theta-2)(1-\delta)}{4}} \Big] \\
 & \le C_\Xvec^2 C_\theta n^{-\frac{(\theta-2)(1-\delta)}{4}} = o \biggl( \sqrt{\frac{\log(n \lor p)}{n}} \biggr).
\end{align*}
\eqref{eq:lemma1:part2} follows upon combining \eqref{eq:lemma1:part2:eq1} and \eqref{eq:lemma1:part2:eq2}.

\subsection*{Proof of Lemma \ref{lemmaA2}}

Suppose we are on the event $\mathcal{S}_\gamma$ and let $\gamma^\prime \ge \gamma$. In the case that $\beta^* = 0$, it holds that $\hat{\beta}_{2\gamma^\prime/\sqrt{n}} = 0$ for all $\gamma^\prime \ge \gamma$, implying that $\remainder(\gamma^\prime,e) = 0$. Hence, Lemma \ref{lemmaA2} trivially holds true if $\beta^* = 0$. We can thus restrict attention to the case that $\beta^* \ne 0$. Define $a_n = B (\log n)^2 \sqrt{\normone{\beta^*}}$ with some $B > 0$ and write $e_i = e_i^\le + e_i^>$ with
\begin{align*}
e_i^\le & = e_i \, \ind(|e_i| \le \log n) - \ex[e_i \, \ind(|e_i| \le \log n)] \\
e_i^> & = e_i \, \ind(|e_i| > \log n) - \ex[e_i \, \ind(|e_i| > \log n)]. 
\end{align*}
With this notation, we get that 
\begin{align}
 & \pr_e \biggl( \remainder(\gamma^\prime,e) > \frac{a_n \sqrt{\gamma^\prime}}{n^{1/4}} \biggr) \nonumber \\ 
 & = \pr_e \biggl( \max_{1 \le j \le p} \Big| \frac{1}{\sqrt{n}} \sum_{i=1}^n \Xvec_{ij} \Xvec_i^\top \big(\beta^* - \hat{\beta}_{\frac{2}{\sqrt{n}}\gamma^\prime}\big) e_i \Big| > \frac{a_n \sqrt{\gamma^\prime}}{n^{1/4}} \biggr) \nonumber \\
 & \le \sum\limits_{j=1}^p \pr_e \biggl( \Big| \frac{1}{\sqrt{n}} \sum_{i=1}^n \Xvec_{ij} \Xvec_i^\top \big(\beta^* - \hat{\beta}_{\frac{2}{\sqrt{n}}\gamma^\prime}\big) e_i \Big| > \frac{a_n \sqrt{\gamma^\prime}}{n^{1/4}} \biggr) \nonumber \\
 & \le \sum\limits_{j=1}^p \big\{ P_{e,j}^\le + P_{e,j}^> \big\}, \label{eq:lemmaA2:mainbound}
\end{align}
where
\begin{align*} 
P_{e,j}^\le & = \pr_e \biggl( \Big| \frac{1}{\sqrt{n}} \sum_{i=1}^n \Xvec_{ij} \Xvec_i^\top \big(\beta^* - \hat{\beta}_{\frac{2}{\sqrt{n}}\gamma^\prime}\big) e_i^\le \Big| > \frac{a_n \sqrt{\gamma^\prime}}{2 n^{1/4}} \biggr) \\
P_{e,j}^> & = \pr_e \biggl( \Big| \frac{1}{\sqrt{n}} \sum_{i=1}^n \Xvec_{ij} \Xvec_i^\top \big(\beta^* - \hat{\beta}_{\frac{2}{\sqrt{n}}\gamma^\prime}\big) e_i^> \Big| > \frac{a_n \sqrt{\gamma^\prime}}{2 n^{1/4}} \biggr).
\end{align*}
In what follows, we prove that for every $j \in \{1,\ldots,p\}$, 
\begin{equation} 
P_{e,j}^\le \le C n^{-K} \qquad \text{and} \qquad P_{e,j}^> \le C n^{-K}, \label{eq:lemmaA2:claim}
\end{equation}
where the constants $C$ and $K$ depend only on the parameters $\Theta$, and $K$ can be chosen as large as desired by picking $C$ large enough. Plugging this into \eqref{eq:lemmaA2:mainbound} immediately yields the statement of Lemma \ref{lemmaA2}.

We first show that $P_{e,j}^\le \le C n^{-K}$. To do so, we make use of the prediction bound \eqref{eq:predictionbound} which implies that 
\begin{equation}\label{eq:lemmaA2:predictionbound}
\frac{1}{n} \sum\limits_{i=1}^n \big\{ \Xvec_i^\top ( \beta^* - \hat{\beta}_{\frac{2}{\sqrt{n}}\gamma^\prime} ) \big\}^2 \le \frac{4 \gamma^\prime \normone{\beta^*}}{\sqrt{n}} 
\end{equation}
for any $\gamma^\prime \ge \gamma$ on the event $\mathcal{S}_\gamma$. From this, it immediately follows that on $\mathcal{S}_\gamma$,  
\begin{equation}\label{eq:lemmaA2:predbound}
\biggl| \frac{\Xvec_i^\top (\beta^*-\hat{\beta}_{\frac{2}{\sqrt{n}}\gamma^\prime})}{\sqrt{n}} \biggr| \le \frac{2 \sqrt{\gamma^\prime \normone{\beta^*}}}{n^{1/4}} 
\end{equation}
for all $i$. Using Markov's inequality, $P_{e,j}^\le$ can be bounded by 
\begin{align} 
P_{e,j}^\le 
 & = \pr_e \biggl( \Big| \frac{1}{n^{1/4}} \sum_{i=1}^n \Xvec_{ij} \Xvec_i^\top \big(\beta^* - \hat{\beta}_{\frac{2}{\sqrt{n}}\gamma^\prime}\big) e_i^\le \Big| > \frac{a_n \sqrt{\gamma^\prime}}{2} \biggr) \nonumber \\
 & \le \ex_e \exp \biggl( \mu \Big| \frac{1}{n^{1/4}} \sum_{i=1}^n \Xvec_{ij} \Xvec_i^\top \big(\beta^* - \hat{\beta}_{\frac{2}{\sqrt{n}}\gamma^\prime}\big) e_i^\le \Big| \biggr) \Big/ \exp\biggl(\frac{\mu a_n \sqrt{\gamma^\prime}}{2}\biggr) \nonumber \\
 & \le \ex_e \exp \biggl( \frac{\mu}{n^{1/4}} \sum_{i=1}^n \Xvec_{ij} \Xvec_i^\top \big(\beta^* - \hat{\beta}_{\frac{2}{\sqrt{n}}\gamma^\prime}\big) e_i^\le \biggr) \Big/ \exp\biggl(\frac{\mu a_n \sqrt{\gamma^\prime}}{2}\biggr) \nonumber \\
 & \quad + \ex_e \exp \biggl( - \frac{\mu}{n^{1/4}} \sum_{i=1}^n \Xvec_{ij} \Xvec_i^\top \big(\beta^* - \hat{\beta}_{\frac{2}{\sqrt{n}}\gamma^\prime}\big) e_i^\le \biggr) \Big/ \exp\biggl(\frac{\mu a_n \sqrt{\gamma^\prime}}{2}\biggr) \label{eq:lemmaA2:expbound1}
\end{align}
with any $\mu > 0$. We make use of this bound with the particular choice $\mu = (4 C_\Xvec \sqrt{\gamma^\prime \normone{\beta^*}} \log n)^{-1}$. Since $|\mu \Xvec_{ij} \Xvec_i^\top (\beta^* - \hat{\beta}_{2\gamma^\prime/\sqrt{n}}) e_i^\le / n^{1/4} | \le 1/2$ by condition \ref{C2} and \eqref{eq:lemmaA2:predbound} and since $\exp(x) \le 1 + x + x^2$ for any $|x| \le 1/2$, we obtain that
\begin{align}
 & \ex_e \exp \biggl( \pm \frac{\mu}{n^{1/4}} \sum_{i=1}^n \Xvec_{ij} \Xvec_i^\top \big(\beta^* - \hat{\beta}_{\frac{2}{\sqrt{n}}\gamma^\prime}\big) e_i^\le \biggr) \nonumber \\
 & = \prod\limits_{i=1}^n \ex_e \exp \biggl( \pm \mu \Xvec_{ij} \, \frac{\Xvec_i^\top (\beta^* - \hat{\beta}_{\frac{2}{\sqrt{n}}\gamma^\prime})}{n^{1/4}} \, e_i^\le \biggr) \nonumber \\
 & \le \prod\limits_{i=1}^n \biggl\{ 1 + \mu^2 \Xvec_{ij}^2 \biggl( \frac{\Xvec_i^\top (\beta^* - \hat{\beta}_{\frac{2}{\sqrt{n}}\gamma^\prime})}{n^{1/4}} \biggr)^2 \ex (e_i^\le)^2 \biggr\} \nonumber \\
 & \le \prod\limits_{i=1}^n \exp \biggl(  \mu^2 \Xvec_{ij}^2 \biggl( \frac{\Xvec_i^\top (\beta^* - \hat{\beta}_{\frac{2}{\sqrt{n}}\gamma^\prime})}{n^{1/4}} \biggr)^2 \ex (e_i^\le)^2 \biggr) \nonumber \\
 & \le \exp \biggl( \frac{c \mu^2}{\sqrt{n}} \sum\limits_{i=1}^n \big\{ \Xvec_i^\top (\beta^* - \hat{\beta}_{\frac{2}{\sqrt{n}}\gamma^\prime}) \big\}^2 \biggr) \label{eq:lemmaA2:expbound2}
\end{align}
with a sufficiently large $c > 0$. Plugging \eqref{eq:lemmaA2:expbound2} into \eqref{eq:lemmaA2:expbound1} and using \eqref{eq:lemmaA2:predictionbound} along with the definition of $\mu$, we arrive at 
\begin{align*}
P_{e,j}^\le 
 & \le 2 \exp \biggl( \frac{c \mu^2}{\sqrt{n}} \sum\limits_{i=1}^n \big\{ \Xvec_i^\top (\beta^* - \hat{\beta}_{\frac{2}{\sqrt{n}}\gamma^\prime}) \big\}^2 - \frac{\mu a_n \sqrt{\gamma^\prime}}{2} \biggr) \\
 & \le 2 \exp \biggl( 4 c \mu^2 \gamma^\prime \normone{\beta^*} - \frac{\mu a_n \sqrt{\gamma^\prime}}{2} \biggr) \\
 & \le 2 \exp \biggl( \frac{c}{4 C_\Xvec^2 (\log n)^2} - \frac{B \log n}{8 C_\Xvec} \biggr) \le C n^{-K},
\end{align*}
where $K$ can be chosen as large as desired by picking $C$ large enough.

We next verify that $P_{e,j}^> \le C n^{-K}$. The term $P_{e,j}^>$ can be bounded by $P_{e,j}^> \le P_{e,j,1}^> + P_{e,j,2}^>$, where
\begin{align*}
P_{e,j,1}^> & = \pr_e \biggl( \Big| \frac{1}{\sqrt{n}} \sum_{i=1}^n \Xvec_{ij} \Xvec_i^\top (\beta^* - \hat{\beta}_{\frac{2}{\sqrt{n}}\gamma^\prime}) e_i \, \ind(|e_i| > \log n) \Big| > \frac{a_n\sqrt{\gamma^\prime}}{4 n^{1/4}} \biggr) \\  
P_{e,j,2}^> & = \pr_e \biggl( \Big| \frac{1}{\sqrt{n}} \sum_{i=1}^n \Xvec_{ij} \Xvec_i^\top (\beta^* - \hat{\beta}_{\frac{2}{\sqrt{n}}\gamma^\prime}) \ex[e_i \, \ind(|e_i| > \log n)] \Big| > \frac{a_n\sqrt{\gamma^\prime}}{4 n^{1/4}} \biggr).   
\end{align*}
Since the variables $e_i$ are standard normal, it holds that 
\begin{align} 
P_{e,j,1}^> 
 & \le \pr \big( |e_i| > \log n \text{ for some } 1 \le i \le n \big) \nonumber \\
 & \le \sum\limits_{i=1}^n \pr \big( |e_i| > \log n \big) \le \frac{2 n}{\sqrt{2\pi} \log n} \exp \biggl( - \frac{(\log n)^2}{2} \biggr) \le C n^{-K} \label{eq:lemmaA2:tailbound1}  
\end{align} 
for any $n > 1$, where $K > 0$ can be chosen as large as desired. Moreover, with the help of condition \ref{C2} and \eqref{eq:lemmaA2:predbound}, we get that  
\begin{align}
P_{e,j,2}^> 
 & \le \pr_e \biggl( \sum_{i=1}^n |\Xvec_{ij}| \biggl| \frac{\Xvec_i^\top (\beta^* - \hat{\beta}_{\frac{2}{\sqrt{n}}\gamma^\prime})}{\sqrt{n}} \biggr| \ex[|e_i| \, \ind(|e_i| > \log n)] > \frac{a_n\sqrt{\gamma^\prime}}{4 n^{1/4}} \biggr) \nonumber \\
 & \le \pr_e \biggl(  C_\Xvec \frac{2 \sqrt{\gamma^\prime \normone{\beta^*}}}{n^{1/4}} \, \sum_{i=1}^n \ex[|e_i| \, \ind(|e_i| > \log n)] > \frac{a_n\sqrt{\gamma^\prime}}{4 n^{1/4}} \biggr) \nonumber \\
 & \le \pr_e \biggl(  \sum_{i=1}^n \ex[|e_i| \, \ind(|e_i| > \log n)] > \frac{B (\log n)^2}{8 C_\Xvec} \biggr) = 0 \label{eq:lemmaA2:tailbound2}  
\end{align} 
for $n$ large enough, where the last equality follows from the fact that for any $c > 1$, 
\begin{align*}
\sum_{i=1}^n \ex[|e_i| \, \ind(|e_i| > \log n)] 
 & \le \sum_{i=1}^n \ex\biggl[\frac{|e_i| \exp(c|e_i|)}{\exp(c\log n)} \, \ind(|e_i| > \log n)\biggr] \\
 & \le \frac{n \ex[|e_i| \exp(c|e_i|)]}{\exp(c\log n)} = o(1). 
\end{align*}
Combining \eqref{eq:lemmaA2:tailbound1} and \eqref{eq:lemmaA2:tailbound2}, we can conclude that $P_{e,j}^> \le C n^{-K}$, where $K$ can be picked as large as desired.

\subsection*{Proof of Lemma \ref{lemmaA3}}

The proof is based on standard concentration and maximal inequalities. According to the Gaussian concentration inequality stated in Theorem 7.1 of \cite{Ledoux2001} (see also Lemma 7 in \cite{Chernozhukov2015}), it holds that
\begin{equation}\label{eq:lemmaA3:exponential1}
\pr \Bigl( \max_{1 \le j \le p} \big| \Wgauss_j/\sigma_j \big| \ge \ex \Bigl[ \max_{1 \le j \le p} \big| \Wgauss_j/\sigma_j \big| \Bigr] + \sqrt{2 \log (n \lor p)} \biggr) \le \frac{1}{n \lor p},
\end{equation}
where we use the notation $\sigma_j^2 = \ex[\Wgauss_j^2]$. Combining \eqref{eq:lemmaA3:exponential1} with the maximal inequality $\ex[ \max_{1 \le j \le p} | \Wgauss_j/\sigma_j |] \leq \sqrt{2 \log (2 p)}$ (see e.g.\ Proposition 1.1.3 in \cite{Talagrand2003}) and multiplying each term inside the probability of \eqref{eq:lemmaA3:exponential1} with $C_{\Wgauss} = C_\Xvec C_\sigma$ yields
\begin{equation}\label{eq:lemmaA3:exponential2}
\pr \Bigl( C_{\Wgauss} \max_{1 \le j \le p} \big| \Wgauss_j/\sigma_j \big| \ge C_{\Wgauss} \bigl[ \sqrt{2 \log (2 p)} + \sqrt{2 \log (n \lor p)} \bigr] \biggr) \le \frac{1}{n \lor p}.
\end{equation}
Since $\sigma_j \le C_{\Wgauss}$ for any $j$, it holds that $C_{\Wgauss} \max_{1 \le j \le p} |\Wgauss_j/\sigma_j| \ge \max_{1 \le j \le p} |\Wgauss_j|$. Plugging this into \eqref{eq:lemmaA3:exponential2}, we arrive at 
\begin{equation*}
\pr \Bigl( \max_{1 \le j \le p} |\Wgauss_j| \ge C_{\Wgauss} \bigl[ \sqrt{2 \log (2 p)} + \sqrt{2\log (n \lor p)} \bigr] \biggr) \le \frac{1}{n \lor p},
\end{equation*}
which implies that $\gamma_\alpha^{\Wgauss} \le C_{\Wgauss} [ \sqrt{2 \log (2 p)} + \sqrt{2\log (n \lor p)} ]$ for any $\alpha > 1/(n \lor p)$.

%\subsection*{Proof of Lemma \ref{lemmaA3:lowerbound}}

%Since $\max_{1 \le j \le p} |\Wgauss_j| \ge |\Wgauss_1|$, it trivially holds that $\gamma_{\alpha}^{\Wgauss} \ge \gamma_{\alpha}^{\Wgauss_1}$, where $\gamma_{\alpha}^{\Wgauss_1}$ is the $(1-\alpha)$-quantile of $|\Wgauss_1|$. Moreover, it is straightforward to verify that under assumptions \ref{C2} and \ref{C3}, $\var(\Wgauss_1) \ge c_\Xvec^2 c_\sigma^2$. We thus obtain that for any $\alpha \in (0,1-\delta_0]$, $\gamma_{\alpha}^{\Wgauss} \ge \gamma_{\alpha}^{\Wgauss_1} \ge \gamma_{1-\delta_0}^{\Wgauss_1} \ge c_{\min}$, where $c_{\min}$ is the $\delta_0$-quantile of a random variable $|V|$ with $V \sim \normal(0, c_\Xvec^2 c_\sigma^2)$.

\subsection*{Proof of Lemma \ref{lemmaA8}}

The proof is by contradiction. Suppose that $\pr(\max_{1 \le j \le p} V_j \le \gamma_\alpha^V) > 1 - \alpha$, in particular, $\pr(\max_{1 \le j \le p} V_j \le \gamma_\alpha^V) = 1 - \alpha + \eta$ with some $\eta > 0$. By Lemma \ref{lemmaA5}, %it holds that 
\[ \sup_{t \in \R} \pr \biggl( \Big| \max_{1 \le j \le p} V_j - t \Big| \le \delta \biggr) \le b(\delta) := C \delta \sqrt{1 \lor \log(p/\delta)} \]
for any $\delta > 0$, which implies that
\begin{align*} 
\pr \Big( \max_{1 \le j \le p} V_j \le \gamma_\alpha^V - \delta \Big) 
 & = \pr \Big( \max_{1 \le j \le p} V_j \le \gamma_\alpha^V \Big) - \pr \Big( \gamma_\alpha^V - \delta < \max_{1 \le j \le p} V_j \le \gamma_\alpha^V \Big) \\
 & \ge \pr \Big( \max_{1 \le j \le p} V_j \le \gamma_\alpha^V \Big) - \sup_{t \in \R} \pr \biggl( \Big| \max_{1 \le j \le p} V_j - t \Big| \le \delta \biggr) \\
 & \ge 1 - \alpha + \eta - b(\delta).
\end{align*}
Since $b(\delta) \rightarrow 0$ as $\delta \rightarrow 0$, we can find a specific $\delta > 0$ with $b(\delta) < \eta$. For this specific $\delta$, we get that $\pr ( \max_{1 \le j \le p} V_j \le \gamma_\alpha^V - \delta ) > 1-\alpha$, which contradicts the definition of the quantile $\gamma_\alpha^V$ according to which $\gamma_\alpha^V = \inf\{q: \pr(\max_{1 \le j \le p} V_j \le q) \ge 1  - \alpha\}$.

\subsection*{Proof of Proposition \ref{propA3}}

We first have a closer look at the statistic $\critstar_B := \normsup{(\mathcal{P} \Xmat_B)^\top \uvec} / \sqrt{n}$. Without loss of generality, we let $A = \{1, \ldots,p_A\}$ and $B = \{p_A+1, \ldots,p_A+p_B\}$ with $p_A + p_B = p$, and we write $\Xvec_{i,A} = (\Xvec_{i1},\ldots,\Xvec_{ip_A})^\top$ to shorten notation. Moreover, we define $\hat{\psi}_{jk} = n^{-1} \sum_{i=1}^n X_{ij} X_{ik}$ and set $\hat{\psi}_{j,A} = (\hat{\psi}_{j1},\ldots,\hat{\psi}_{j p_A})^\top \in \R^{p_A}$ along with $\hat{\Psi}_A = ( \hat{\psi}_{jk}: 1 \le j,k \le p_A ) \in \R^{p_A \times p_A}$. Similarly, we let $\psi_{jk} = \ex[X_{ij} X_{ik}]$, $\psi_{j,A} = (\psi_{j1},\ldots,\psi_{j p_A})^\top$ and $\Psi_A = ( \psi_{jk}: 1 \le j,k \le p_A )$. With this notation, the statistic $\critstar_B = \normsup{(\mathcal{P} \Xmat_B)^\top \uvec} / \sqrt{n} = \normsup{(\mathcal{P} \Xmat_B)^\top \mathcal{P} \eps} / \sqrt{n} = \normsup{(\mathcal{P} \Xmat_B)^\top \eps} / \sqrt{n}$ can be rewritten as $\critstar_B = \max_{j \in B} |\Wstar_{j,B}|$, where 
\[ \Wstar_{j,B} = \frac{1}{\sqrt{n}} \sum_{i=1}^n \hat{Z}_{ij} \eps_i \quad \text{with} \quad \hat{Z}_{ij} = \Xvec_{ij} - \hat{\psi}_{j,A}^\top \hat{\Psi}_A^{-1} \Xvec_{i,A}, \]
and $\Wstar_B = (\Wstar_{j,B}: j \in B)$ is the vector with the elements $\Wstar_{j,B}$. In contrast to $\Xvec_i$, the random vectors $\hat{Z}_i = (\hat{Z}_{ij}: j \in B)$ are not independent across $i$ in general. In order to deal with this complication, we introduce the auxiliary statistic $\critstarstar_B = \max_{j \in B} |\Wstarstar_{j,B}|$, where $\Wstarstar_B = (\Wstarstar_{j,B}: j \in B)$ and 
\[ \Wstarstar_{j,B} =  \frac{1}{\sqrt{n}} \sum_{i=1}^n Z_{ij} \eps_i \quad \text{with} \quad Z_{ij} = \Xvec_{ij} - \psi_{j,A}^\top \Psi_A^{-1} \Xvec_{i,A}. \]
The random vectors $Z_i = (Z_{ij}: j \in B)$ have the following properties: 
(i) Unlike $\hat{Z}_i$, they are independent across $i$. 
(ii) Since $|\Xvec_{ij}| \le C_\Xvec$ by \ref{C2} and $\Psi_A$ is positive definite by assumption, $|Z_{ij}| \le C_Z < \infty$ with a constant $C_Z$ that depends only on the model parameters $\Theta^\prime$.
(iii) Since $Z_{ij}$ can be expressed as $Z_{ij} = X_{ij} - X_{i,A}^\top \vartheta^{(j)}$ with $\vartheta^{(j)}$ introduced before the formulation of Proposition \ref{prop:testing:2}, it holds that $\ex[Z_{ij}^2] \ge c_Z^2 > 0$ with $c_Z^2 = c_\vartheta$. 
We denote the $(1-\alpha)$-quantile of $\critstarstar_B$ by $\gamma_{\alpha,B}^{**}$. In the course of the proof, we will establish that $\gamma_{\alpha,B}^{**}$ is close to the quantile $\gamma_{\alpha,B}^*$ of the statistic $\critstar_B$ in a suitable sense.

In addition to the above quantities, we introduce some auxiliary statistics that parallel those defined in the proof of Theorem \ref{theo1}. To start with, let $\crithat_B(\gamma,e) = \max_{j \in B} |\What_{j,B}(\gamma,e)|$, where $\What_B(\gamma,e) = (\What_{j,B}(\gamma,e): j \in B)$ with 
\[ \What_{j,B}(\gamma,e) = \frac{1}{\sqrt{n}} \sum\limits_{i=1}^n \hat{Z}_{ij} \hat{\uvec}_{\frac{2}{\sqrt{n}}\gamma,i} e_i, \]
and let $\hat{\pi}_{\alpha,B}(\gamma)$ be the $(1-\alpha)$-quantile of $\crithat_B(\gamma,e)$ conditionally on $\Xmat$ and $\eps$. With this notation, the estimator $\hat{\gamma}_{\alpha,B}$ can be expressed as
\[ \hat{\gamma}_{\alpha,B} = \inf \big\{ \gamma > 0: \hat{\pi}_{\alpha,B}(\gamma^\prime) \le \gamma^\prime \text{ for all } \gamma^\prime \ge \gamma \big\}. \]
Moreover, let $\critgauss_B = \max_{j \in B} |\Wgauss_j|$, where $\Wgauss_B = (\Wgauss_j: j \in B)$ is a Gaussian random vector with $\ex[\Wgauss_B] = \ex[\Wstarstar_B] = 0$ and $\ex[\Wgauss_B \Wgauss_B^\top] = \ex[\Wstarstar_B (\Wstarstar_B)^\top]$, and let $\gamma_{\alpha,B}^\Wgauss$ denote the $(1-\alpha)$-quantile of $\critgauss_B$. Finally, define the statistic $\crit_B(e) = \max_{j \in B} |\W_{j,B}(e)|$, where $\W_B(e) = (\W_{j,B}(e): j \in B)$ with
\[ \W_{j,B}(e) =  \frac{1}{\sqrt{n}} \sum_{i=1}^n Z_{ij} \eps_i e_i, \]
and let $\gamma_{\alpha,B}$ be the $(1-\alpha)$-quantile of $\crit_B(e)$ conditionally on $\Xmat$ and $\eps$.

We next define some expressions which play a similar role as the quantity $\Delta$ in the proof of Theorem \ref{theo1}. In particular, we let $\Delta_1 = \normtwo{n^{-1} \sum_{i=1}^n \Xvec_{i,A} \eps_i}$ along with
\begin{align*}
\Delta_2 & = \max_{j \in A} \biggl| \frac{1}{n} \sum\limits_{i=1}^n \big\{ \Xvec_{ij}^2 \eps_i^2 - \ex[ \Xvec_{ij}^2 \eps_i^2] \big\} \biggr| \\
\Delta_3 & = \max_{1 \le j,k \le p} \biggl| \frac{1}{n} \sum\limits_{i=1}^n \big\{ \Xvec_{ij} \Xvec_{ik} - \ex[ \Xvec_{ij} \Xvec_{ik}] \big\} \biggr| \\
\Delta_4 & = \max_{j,k \in B} \biggl| \frac{1}{n} \sum\limits_{i=1}^n \big\{ Z_{ij} Z_{ik} \eps_i^2 - \ex[ Z_{ij} Z_{ik} \eps_i^2] \big\} \biggr|.
\end{align*}
Applying Markov's inequality, we obtain that 
\begin{align}
\pr \big( \Delta_1 > n^{-\frac{1}{2}+\rho} \big) & \le Cn^{-2\rho} \label{eq:probbound:Delta1} \\
\pr \big( \Delta_2 > n^{-\frac{1}{2}+\rho} \big) & \le Cn^{-2\rho}, \label{eq:probbound:Delta2}
\end{align}
where we choose $\rho$ to be a fixed constant with $\rho \in (0,1/2)$ and $C$ depends only on $\Theta^\prime$. Moreover, noticing that $|Z_{ij}| \le C_Z < \infty$ and $\ex[Z_{ij}^2] \ge c_Z^2 > 0$ under the conditions of Proposition \ref{prop:testing:2}, the same arguments as for Lemma \ref{lemmaA1} yield the following: there exist positive constants $C$, $D$ and $K$ depending only on $\Theta^\prime$ such that 
\begin{align}
\pr \big( \Delta_3 > D \sqrt{\log(n \lor p) /n} \big) & \le Cn^{-K} \label{eq:probbound:Delta3} \\
\pr \big( \Delta_4 > D \sqrt{\log(n \lor p) /n} \big) & \le Cn^{-K}. \label{eq:probbound:Delta4}
\end{align}
Taken together, \eqref{eq:probbound:Delta1}--\eqref{eq:probbound:Delta4} imply that the event 
\[ \mathcal{A}_n^\prime := \big\{ (\Delta_1 \lor \Delta_2) \le n^{-\frac{1}{2}+\rho} \text{ and } (\Delta_3 \lor \Delta_4) \le D \sqrt{\log(n \lor p) /n} \big\} \]
occurs with probability at least $1 - O(n^{-K} \lor n^{-2\rho})$.

With the above notation at hand, we now turn to the proof of Proposition \ref{propA3}. In a first step, we show that the quantiles of the statistic $\critstar_B$ are close to those of the auxiliary statistic $\critstarstar_B$ in the following sense: there exist positive constants $C$ and $K$ depending only on $\Theta^\prime$ such that 
\begin{equation}\label{eq:propA1:B}
\begin{split}
\gamma_{\alpha+\zeta_n,B}^* \le \gamma_{\alpha,B}^{**} \le \gamma_{\alpha-\zeta_n,B}^* \\ 
 \gamma_{\alpha+\zeta_n,B}^{**} \le \gamma_{\alpha,B}^* \le \gamma_{\alpha-\zeta_n,B}^{**} 
\end{split}
\end{equation}
for any $\alpha \in (\zeta_n,1-\zeta_n)$ with $\zeta_n = C n^{-K}$. The proof of \eqref{eq:propA1:B} is postponed until the arguments for Proposition \ref{propA3} are complete. In the second step, we relate the quantiles $\gamma_{\alpha,B}^{**}$ of $\critstarstar_B$ to the quantiles $\gamma_{\alpha,B}$ of $\crit_B(e)$. Arguments completely analogous to those for Proposition \ref{propA2} yield the following: there exist positive constants $C$ and $K$ depending only on $\Theta^\prime$ such that on the event $\mathcal{A}_n^\prime$, 
\begin{equation}\label{eq:propA2:B}
\begin{split}
 & \gamma_{\alpha+\xi_n^\prime,B} \le \gamma_{\alpha,B}^{**} \le \gamma_{\alpha-\xi_n^\prime,B} \\ 
 & \gamma_{\alpha+\xi_n^\prime,B}^{**} \le \gamma_{\alpha,B} \le \gamma_{\alpha-\xi_n^\prime,B}^{**} 
\end{split}
\end{equation}
for any $\alpha \in (\xi_n^\prime,1-\xi_n^\prime)$ with $\xi_n^\prime = C n^{-K}$. In the third step, we relate the auxiliary statistic $\crit_B(e)$ to the criterion function $\crithat_B(\gamma,e)$, which underlies the estimator $\hat{\gamma}_{\alpha,B}$. Straightforward calculations show that 
\begin{equation}\label{eq:critboundB}
\crithat_B(\gamma,e) 
\begin{cases}
\le \crit_B(e) + \remainder_B(\gamma,e) \\
\ge \crit_B(e) - \remainder_B(\gamma,e), 
\end{cases}
\end{equation}
where $\remainder_B(\gamma,e) = \remainder_{B,1}(\gamma,e) + \remainder_{B,2}(e) + \remainder_{B,3}(e)$ with
\begin{align*}
\remainder_{B,1}(\gamma,e) & = \max_{j \in B} \biggl| \frac{1}{\sqrt{n}} \sum_{i=1}^n \hat{Z}_{ij} \big\{ \mathcal{P} \Xmat_B \big(\beta_B^* - \hat{\beta}_{B,\frac{2}{\sqrt{n}}\gamma}\big) \big\}_i \, e_i \biggr| \\
\remainder_{B,2}(e) & = \max_{j \in B} \biggl| \frac{1}{\sqrt{n}} \sum\limits_{i=1}^n (\hat{Z}_{ij} - Z_{ij}) \eps_i e_i \biggr| \\
\remainder_{B,3}(e) & = \max_{j \in B} \biggl| \frac{1}{\sqrt{n}} \sum\limits_{i=1}^n \hat{Z}_{ij} (\eps_i - \uvec_i) e_i \biggr|. 
\end{align*} 
The terms $\remainder_{B,1}(\gamma,e)$, $\remainder_{B,2}(e)$ and $\remainder_{B,3}(e)$ have the following properties: on the event $S_\gamma^\prime \cap \mathcal{A}_n^\prime$, 
\begin{equation}\label{eq:remainderB1}
 \pr_e \biggl( \remainder_{B,1}(\gamma^\prime,e) > \frac{D (\log n)^2 \sqrt{\normone{\beta_B^*} \gamma^\prime}}{n^{1/4}} \biggr) \le C n^{-K} 
\end{equation}
for every $\gamma^\prime \ge \gamma$, where the constants $C$, $D$ and $K$ depend only on $\Theta^\prime$. Moreover, on the event $\mathcal{A}_n^\prime$, 
\begin{align}
\pr_e \biggl( \remainder_{B,2}(e) > \frac{D \log^{1/2}(n \lor p)}{n^{1/2-\rho}} \biggr) & \le C n^{-2\rho} \label{eq:remainderB2} \\
\pr_e \biggl( \remainder_{B,3}(e) > \frac{D \log(n \lor p)}{n^{1/2-\rho}} \biggr) & \le C n^{-K} \label{eq:remainderB3}, 
\end{align}
where $\rho \in (0,1/2)$ has been introduced in \eqref{eq:probbound:Delta1}--\eqref{eq:probbound:Delta2} and the constants $C$, $D$ and $K$ depend only on $\Theta^\prime$. The proofs of \eqref{eq:remainderB1}--\eqref{eq:remainderB3} are provided below. With \eqref{eq:propA1:B}--\eqref{eq:remainderB3} in place, we can now use the same arguments as in the proof of Theorem \ref{theo1} (with minor adjustments) to obtain that $\gamma_{\alpha + \nu_n^\prime,B}^* \le \hat{\gamma}_{\alpha,B} \le \gamma_{\alpha-\nu_n^\prime,B}^*$ on the event $\mathcal{S}_{\gamma_{\alpha+\nu_n^\prime,B}^*}^\prime \cap \mathcal{A}_n^\prime$.

\begin{proof}[\textnormal{\textbf{Proof of (\ref{eq:propA1:B}).}}] 
We prove that 
\begin{align}
\sup_{t \in \R} \big| \pr(\critstarstar_B \le t) - \pr(\critgauss_B \le t) \big| & \le C n^{-K} \label{eq:propA1:B:claim1} \\
\sup_{t \in \R} \big| \pr(\critstar_B \le t) - \pr(\critgauss_B \le t) \big| & \le C n^{-K}, \label{eq:propA1:B:claim2}
\end{align}
where $C$ and $K$ depend only on $\Theta^\prime$. Applying the same arguments as in the proof of Proposition \ref{propA1} to the statements \eqref{eq:propA1:B:claim1} and \eqref{eq:propA1:B:claim2} yields that
\begin{equation*}
\begin{split}
 & \gamma_{\alpha + Cn^{-K},B}^{**} \le \gamma_{\alpha,B}^\Wgauss \le \gamma_{\alpha - Cn^{-K},B}^{**} \\
 & \gamma_{\alpha + Cn^{-K},B}^\Wgauss \le \gamma_{\alpha,B}^{**} \le \gamma_{\alpha - Cn^{-K},B}^\Wgauss 
\end{split}
\qquad \text{and} \qquad 
\begin{split}
 & \gamma_{\alpha + Cn^{-K},B}^{*} \le \gamma_{\alpha,B}^\Wgauss \le \gamma_{\alpha - Cn^{-K},B}^{*} \\
 & \gamma_{\alpha + Cn^{-K},B}^\Wgauss \le \gamma_{\alpha,B}^{*} \le \gamma_{\alpha - Cn^{-K},B}^\Wgauss, 
\end{split}
\end{equation*}
from which \eqref{eq:propA1:B} follows immediately.

It remains to prove \eqref{eq:propA1:B:claim1} and \eqref{eq:propA1:B:claim2}. \eqref{eq:propA1:B:claim1} is a direct consequence of Lemma \ref{lemmaA7}, since $0 < c_\sigma^2 c_Z^2 \le n^{-1} \sum_{i=1}^n \ex[ (Z_{ij} \eps_i)^2] \le C_\sigma^2 C_Z^2 < \infty$ and $\max_{k=1,2} \{ n^{-1} \sum_{i=1}^n \linebreak \ex[|Z_{ij} \eps_i|^{2+k} / C^k] \} + \ex[\{ \max_{j \in B} |Z_{ij} \eps_i|/C \}^4] \le 4$ for $C$ large enough, where we have used \ref{C3} and the fact that $|Z_{ij}| \le C_Z < \infty$ and $\ex[Z_{ij}^2] \ge c_Z^2 > 0$ under the conditions of Proposition \ref{prop:testing:2}. For the proof of \eqref{eq:propA1:B:claim2}, it suffices to show that
\begin{equation}\label{eq:propA1:B:claim3}
\sup_{t \in \R} \big| \pr(\critstar_B \le t) - \pr(\critstarstar_B \le t) \big| \le C n^{-K}
\end{equation}
with $C$ and $K$ depending only on $\Theta^\prime$, since by \eqref{eq:propA1:B:claim1},
\begin{align*}
\sup_{t \in \R} \big| \pr(\critstar_B \le t) - \pr(\critgauss_B \le t) \big| 
 & \le \sup_{t \in \R} \big| \pr(\critstar_B \le t) - \pr(\critstarstar_B \le t) \big| \\
 & \quad + \sup_{t \in \R} \big| \pr(\critstarstar_B \le t) - \pr(\critgauss_B \le t) \big| \\
 & \le \sup_{t \in \R} \big| \pr(\critstar_B \le t) - \pr(\critstarstar_B \le t) \big| + Cn^{-K}. 
\end{align*}

To prove \eqref{eq:propA1:B:claim3}, we fix a constant $d \in (0,1/2)$ and let $c_n = D n^d \sqrt{\log(n \lor p) / n}$, where $D$ is a sufficiently large constant that depends only on $\Theta^\prime$. In the case that $\pr(\critstar_B \le t) \ge \pr(\critstarstar_B \le t)$, the difference $\pr(\critstar_B \le t) - \pr(\critstarstar_B \le t)$ can be bounded as follows:
\begin{align} 
 & \pr(\critstar_B \le t) - \pr(\critstarstar_B \le t) \nonumber \\
 & = \pr (\critstarstar_B \le t + \critstarstar_B - \critstar_B, |\critstarstar_B - \critstar_B| \le c_n) \nonumber \\
 & \quad + \pr (\critstarstar_B \le t + \critstarstar_B - \critstar_B, |\critstarstar_B - \critstar_B| > c_n) - \pr(\critstarstar_B \le t) \nonumber \\ 
 & \le \pr(\critstarstar_B \le t + c_n) - \pr(\critstarstar_B \le t) + \pr(|\critstarstar_B - \critstar_B| > c_n) \nonumber \\
 & \le \big| \pr(\critstarstar_B \le t + c_n) - \pr(\critgauss_B \le t + c_n) \big| + \big| \pr(\critstarstar_B \le t) - \pr(\critgauss_B \le t) \big| \nonumber \\
 & \quad + \big| \pr(\critgauss_B \le t + c_n) - \pr(\critgauss_B \le t) \big| + \pr(|\critstarstar_B - \critstar_B| > c_n) \nonumber \\
 & \le \big| \pr(\critstarstar_B \le t + c_n) - \pr(\critgauss_B \le t + c_n) \big| + \big| \pr(\critstarstar_B \le t) - \pr(\critgauss_B \le t) \big| \nonumber \\
 & \quad + \pr(|\critgauss_B - t| \le c_n) + \pr(|\critstarstar_B - \critstar_B| > c_n). \label{eq:propA1:B:eq1}
\end{align}
For the case that $\pr(\critstar_B \le t) < \pr(\critstarstar_B \le t)$, we similarly get that
\begin{align} 
 & \pr(\critstarstar_B \le t) - \pr(\critstar_B \le t) \nonumber \\
 & \le \big| \pr(\critstarstar_B \le t) - \pr(\critgauss_B \le t) \big| + \big| \pr(\critstarstar_B \le t - c_n) - \pr(\critgauss_B \le t - c_n) \big| \nonumber \\
 & \quad + \pr(|\critgauss_B - t| \le c_n) + \pr(|\critstarstar_B - \critstar_B| > c_n). \label{eq:propA1:B:eq2}
\end{align}
\eqref{eq:propA1:B:eq1} and \eqref{eq:propA1:B:eq2} immediately yield that 
\begin{align*}
\sup_{t \in \R} \big| \pr(\critstar_B \le t) - \pr(\critstarstar_B \le t) \big| 
 & \le 2 \sup_{t \in \R} \big| \pr(\critstarstar_B \le t) - \pr(\critgauss_B \le t) \big| \\
 & \quad + \sup_{t \in \R} \pr(|\critgauss_B - t| \le c_n) \\
 & \quad + \pr(|\critstarstar_B - \critstar_B| > c_n). 
\end{align*}
Since $\sup_{t \in \R} |\pr(\critstarstar_B \le t) - \pr(\critgauss_B \le t)| \le Cn^{-K}$ by \eqref{eq:propA1:B:claim1} and $\sup_{t \in \R} \pr(|\critgauss_B - t| \le c_n) \le Cn^{-K}$ by Lemma \ref{lemmaA5}, we further get that  
\[ \sup_{t \in \R} \big| \pr(\critstar_B \le t) - \pr(\critstarstar_B \le t) \big| \le  \pr(|\critstarstar_B - \critstar_B| > c_n) + Cn^{-K}, \]
where $C$ and $K$ depend only on $\Theta^\prime$. To complete the proof of \eqref{eq:propA1:B:claim3}, we thus need to show that 
\begin{equation}\label{eq:propA1:B:claim4}
\pr(|\critstarstar_B - \critstar_B| > c_n) \le C n^{-K} 
\end{equation}
with $C$ and $K$ depending only on $\Theta^\prime$. To do so, we bound the term $|\critstarstar_B - \critstar_B|$ by  
\begin{align}
| \critstarstar_B - \critstar_B | 
 & \le \max_{j \in B} \biggl| \frac{1}{\sqrt{n}} \sum\limits_{i=1}^n (\hat{Z}_{ij} - Z_{ij}) \eps_i \biggr| \nonumber \\
 & = \max_{j \in B} \biggl| \{ \psi_{j,A}^\top \Psi_A^{-1} - \hat{\psi}_{j,A}^\top \hat{\Psi}_A^{-1} \} \frac{1}{\sqrt{n}} \sum\limits_{i=1}^n \Xvec_{i,A} \eps_i \biggr| \nonumber \\
 & \le \biggl\{ \max_{j \in B} \normtwo{\psi_{j,A} - \hat{\psi}_{j,A}} \normtwo{\Psi_A^{-1}} \nonumber \\ & \qquad + \max_{j \in B} \normtwo{\hat{\psi}_{j,A}} \normtwo{\Psi_A^{-1} - \hat{\Psi}_A^{-1}} \Big\} \Big\| \frac{1}{\sqrt{n}} \sum\limits_{i=1}^n \Xvec_{i,A} \eps_i \Big\|_2. \label{eq:propA1:B:claim4:eq1}
\end{align}
From \eqref{eq:probbound:Delta3}, it immediately follows that 
\begin{align}
\pr \Big( \max_{j \in B} \normtwo{\psi_{j,A} - \hat{\psi}_{j,A}} > D \sqrt{\log(n \lor p)/n} \Big) & \le Cn^{-K} \label{eq:propA1:B:claim4:eq2} \\
\pr \Big( \normtwo{\Psi_A - \hat{\Psi}_A} > D \sqrt{\log(n \lor p)/n} \Big) & \le Cn^{-K} \label{eq:propA1:B:claim4:eq3}
\end{align}
with $C$, $D$ and $K$ depending only on $\Theta^\prime$. Moreover, it holds that
\begin{equation}\label{eq:propA1:B:claim4:eq4}
\pr \Big( \normtwo{\Psi_A^{-1} - \hat{\Psi}_A^{-1}} > D \sqrt{\log(n \lor p)/n} \Big) \le Cn^{-K}, 
\end{equation}
which is a consequence of \eqref{eq:propA1:B:claim4:eq3} and the fact that 
\begin{equation}\label{eq:inversematrix}
\normtwo{Q^{-1} - R^{-1}} \le \frac{\normtwos{R^{-1}} \normtwo{R - Q}}{1 - \normtwo{R - Q} \normtwo{R^{-1}}} 
\end{equation}
for every pair of invertible matrices $Q$ and $R$ that are close enough such that $\normtwo{R - Q} \normtwo{R^{-1}} < 1$. Finally, a simple application of Markov's inequality yields that 
\begin{equation} 
\pr \biggl( \Big\| \frac{1}{\sqrt{n}} \sum\limits_{i=1}^n \Xvec_{i,A} \eps_i \Big\|_2 > n^d \biggr) \le C n^{-2d}, \label{eq:propA1:B:claim4:eq5}
\end{equation}
where $C$ depends only on $\Theta^\prime$. The statement \eqref{eq:propA1:B:claim4} follows upon applying the results \eqref{eq:propA1:B:claim4:eq2}--\eqref{eq:propA1:B:claim4:eq4} and \eqref{eq:propA1:B:claim4:eq5} to the bound \eqref{eq:propA1:B:claim4:eq1}. 
\end{proof}

\begin{proof}[\textnormal{\textbf{Proof of (\ref{eq:remainderB1}).}}] 
To start with, we bound $\remainder_{B,1}(\gamma,e)$ by 
\begin{align}
\remainder_{B,1}(\gamma,e) & \le \Big\{ 1 + \sqrt{p_A} \max_{j \in B} \normtwo{\hat{\psi}_{j,A}^\top \hat{\Psi}_A^{-1}} \Big\} \nonumber \\* & \qquad \times \max_{1 \le j \le p} \biggl| \frac{1}{\sqrt{n}} \sum_{i=1}^n \Xvec_{ij} \big\{ \mathcal{P} \Xmat_B \big(\beta_B^* - \hat{\beta}_{B,\frac{2}{\sqrt{n}}\gamma}\big) \big\}_i \, e_i \biggr|. \label{eq:remainderB1:bound}
\end{align}
The same arguments as in the proof of Lemma \ref{lemmaA2} yield that on the event $S_\gamma^\prime$, 
\begin{align}
 \pr_e \biggl( \max_{1 \le j \le p} \biggl| \frac{1}{\sqrt{n}} \sum_{i=1}^n \Xvec_{ij} \big\{ \mathcal{P} \Xmat_B & \big(\beta_B^* - \hat{\beta}_{B,\frac{2}{\sqrt{n}}\gamma^\prime}\big) \big\}_i e_i \biggr| \nonumber \\* & > \frac{D (\log n)^2 \sqrt{\normone{\beta_B^*} \gamma^\prime}}{n^{1/4}} \biggr) \le C n^{-K} \label{eq:remainderB1:eq1}
\end{align}
for every $\gamma^\prime \ge \gamma$, where $C$, $D$ and $K$ depend only on $\Theta^\prime$. Moreover, on the event $\mathcal{A}_n^\prime$, 
\begin{align}
\max_{j \in B} \normtwo{\hat{\psi}_{j,A} - \psi_{j,A}} & \le C \sqrt{\log(n \lor p)/n} \label{eq:remainderB1:eq2} \\[-0.2cm]
\normtwo{\hat{\Psi}_A - \Psi_A} & \le C \sqrt{\log(n \lor p)/n} \label{eq:remainderB1:eq3} \\
\normtwo{\hat{\Psi}_A^{-1} - \Psi_A^{-1}} & \le C \sqrt{\log(n \lor p)/n}, \label{eq:remainderB1:eq4} 
\end{align}
where $C$ is a sufficiently large constant that depends only on $\Theta^\prime$, and \eqref{eq:remainderB1:eq4} is a simple consequence of \eqref{eq:remainderB1:eq3} and \eqref{eq:inversematrix}. To complete the proof, we apply \eqref{eq:remainderB1:eq1}--\eqref{eq:remainderB1:eq4} to the bound \eqref{eq:remainderB1:bound}, taking into account that $\normtwo{\Psi_A^{-1}} \le C < \infty$ and $\max_{j \in B} \normtwo{\psi_{j,A}} \le C < \infty$.
\end{proof}

\begin{proof}[\textnormal{\textbf{Proof of (\ref{eq:remainderB2}).}}]
We have the bound 
\begin{equation}\label{eq:remainderB2:bound}
\remainder_{B,2}(e) \le \max_{j \in B} \big\| \hat{\psi}_{j,A}^\top \hat{\Psi}_A^{-1} - \psi_{j,A}^\top \Psi_A^{-1} \big\|_2 \, \Big\| \frac{1}{\sqrt{n}} \sum\limits_{i=1}^n \Xvec_{i,A} \eps_i e_i \Big\|_2.
\end{equation}
On the event $\mathcal{A}_n^\prime$, 
\begin{align*} 
\pr_e \biggl( \Big| \frac{1}{\sqrt{n}} \sum\limits_{i=1}^n \Xvec_{ij} \eps_i e_i \Big| > n^\rho \biggr) 
 & \le n^{-2\rho} \Big\{ \frac{1}{n} \sum\limits_{i=1}^n \Xvec_{ij}^2 \eps_i^2 \Big\} \le n^{-2\rho} \big\{ \ex[\Xvec_{ij}^2 \eps_i^2] + \Delta_2 \big\} \\
 & \le n^{-2\rho} \big\{ C_X^2 C_\sigma^2 + n^{-\frac{1}{2}+\rho} \big\}
\end{align*}
for every $j \in A$, which implies that $\pr_e(\normtwo{n^{-1/2} \sum_{i=1}^n \Xvec_{i,A} \eps_i e_i} > n^\rho) \le C n^{-2\rho}$ with $C$ depending only on $\Theta^\prime$. To complete the proof, we apply this, \eqref{eq:remainderB1:eq2}--\eqref{eq:remainderB1:eq4} and the fact that $\normtwo{\Psi_A^{-1}} \le C < \infty$ and $\max_{j \in B} \normtwo{\psi_{j,A}} \le C < \infty$ to the bound \eqref{eq:remainderB2:bound}.
\end{proof}

\begin{proof}[\textnormal{\textbf{Proof of (\ref{eq:remainderB3}).}}] Let $d_n = D \log(n \lor p) / n^{1/2-\rho}$ and define  
\begin{align*}
e_i^\le & = e_i \, \ind(|e_i| \le \log n) - \ex[e_i \, \ind(|e_i| \le \log n)] \\
e_i^> & = e_i \, \ind(|e_i| > \log n) - \ex[e_i \, \ind(|e_i| > \log n)]. 
\end{align*}
It holds that 
\begin{align}
\pr_e \big( \remainder_{B,3}(e) > d_n \big) 
 & \le \sum\limits_{j \in B} \pr_e \biggl( \Big| \frac{1}{\sqrt{n}} \sum\limits_{i=1}^n \hat{Z}_{ij} (\eps_i - \uvec_i) e_i \Big| > d_n \biggr) \nonumber \\
 & \le \sum\limits_{j \in B} \big\{ P_{e,j}^\le + P_{e,j}^> \big\}, \label{eq:remainderB3:unionbound}
\end{align}
where 
\begin{align*}
P_{e,j}^\le & = \pr_e \biggl( \Big| \frac{1}{\sqrt{n}} \sum\limits_{i=1}^n \hat{Z}_{ij} (\eps_i - \uvec_i) e_i^\le \Big| > \frac{d_n}{2} \biggr) \\ 
P_{e,j}^> & = \pr_e \biggl( \Big| \frac{1}{\sqrt{n}} \sum\limits_{i=1}^n \hat{Z}_{ij} (\eps_i - \uvec_i) e_i^> \Big| > \frac{d_n}{2} \biggr).
\end{align*}

We first analyze the term $P_{e,j}^\le$. With the help of \eqref{eq:remainderB1:eq2}--\eqref{eq:remainderB1:eq4}, we obtain that on the event $\mathcal{A}_n^\prime$, 
\begin{align*}
|\hat{Z}_{ij}| & = |\Xvec_{ij} - \hat{\psi}_{j,A}^\top \hat{\Psi}_A^{-1} \Xvec_{i,A}| \le \{ 1 + \normtwo{\hat{\psi}_{j,A}} \normtwo{\hat{\Psi}_A^{-1}} \sqrt{p_A} \} C_X \le C \\
|\eps_i - \uvec_i| & = \big|\{ \Xmat_A (\Xmat_A^\top \Xmat_A)^{-1} \Xmat_A^\top \eps \}_i\big| = \biggl|\Xmat_{i,A}^\top \hat{\Psi}_A^{-1} \Big\{ \frac{1}{n} \sum\limits_{\ell=1}^n \Xmat_{\ell,A} \eps_\ell \Big\} \biggr| \\[-0.25cm] & \le \sqrt{p_A} C_X \normtwo{\hat{\Psi}_A^{-1}} \Big\|\frac{1}{n} \sum\limits_{\ell=1}^n \Xmat_{\ell,A} \eps_\ell \Big\|_2 \le \frac{C}{n^{1/2 - \rho}},
\end{align*}
which implies that $| \hat{Z}_{ij} (\eps_i - \uvec_i) e_i^\le | \le C \log n / n^{1/2-\rho}$. Using Markov's inequality, $P_{e,j}^\le$ can be bounded by 
\begin{align}
P_{e,j}^\le 
 & \le \ex_e \exp \biggl( \mu \Big| \frac{1}{\sqrt{n}} \sum\limits_{i=1}^n \hat{Z}_{ij} (\eps_i - \uvec_i) e_i^\le \Big| \biggr) \Big/ \exp\Big( \frac{\mu d_n}{2} \Big) \nonumber \\
 & \le \biggl\{ \ex_e \exp \biggl( \frac{\mu}{\sqrt{n}} \sum\limits_{i=1}^n \hat{Z}_{ij} (\eps_i - \uvec_i) e_i^\le \biggr) \nonumber \\*
 & \qquad + \ex_e \exp \biggl( -\frac{\mu}{\sqrt{n}} \sum\limits_{i=1}^n \hat{Z}_{ij} (\eps_i - \uvec_i) e_i^\le \biggr) \biggr\} \Big/ \exp\Big( \frac{\mu d_n}{2} \Big), \label{eq:remainderB3:expinequ}
\end{align}
where we choose $\mu = c_\mu n^{1/2-\rho}$ with $c_\mu > 0$ so small that $\mu | \hat{Z}_{ij} (\eps_i - \uvec_i) e_i^\le | / \sqrt{n} \le 1/2$. Since $\exp(x) \le 1 + x + x^2$ for $|x| \le 1/2$, we further get that 
\begin{align*}
\ex_e \exp \biggl( \pm \frac{\mu}{\sqrt{n}} \sum\limits_{i=1}^n \hat{Z}_{ij} (\eps_i - \uvec_i) e_i^\le \biggr) 
 & = \prod\limits_{i=1}^n \ex_e \exp \biggl( \pm \frac{\mu \hat{Z}_{ij} (\eps_i - \uvec_i) e_i^\le}{\sqrt{n}} \biggr) \\
 & \le \prod\limits_{i=1}^n \biggl\{ 1 + \frac{\mu^2 \hat{Z}_{ij}^2 (\eps_i - \uvec_i)^2 \ex(e_i^\le)^2}{n} \biggr\} \\
 & \le \prod\limits_{i=1}^n \exp \biggl( \frac{\mu^2 \hat{Z}_{ij}^2 (\eps_i - \uvec_i)^2 \ex(e_i^\le)^2}{n} \biggr) \\
 & = \exp \biggl( \frac{\mu^2}{n} \sum\limits_{i=1}^n \hat{Z}_{ij}^2 (\eps_i - \uvec_i)^2 \ex(e_i^\le)^2 \biggr) \le \exp(c)
\end{align*}
with a sufficiently large constant $c$ that depends only on $\Theta^\prime$. Plugging this into \eqref{eq:remainderB3:expinequ} yields that 
\begin{equation}\label{eq:remainderB3:Pjle}
P_{e,j}^\le \le 2 \exp \Big( c - \frac{c_\mu D \log(n \lor p)}{2} \Big) \le C n^{-K}, 
\end{equation}
where $K$ can be made as large as desired.

We next have a closer look at the term $P_{e,j}^>$. Since $\max_{j \in B} |\sum_{i=1}^n \hat{Z}_{ij} (\eps_i - \uvec_i)| = \normsup{(\mathcal{P} \Xmat_B)^\top (\eps - \mathcal{P} \eps)} = 0$, it holds that 
\[ \frac{1}{\sqrt{n}} \sum\limits_{i=1}^n \hat{Z}_{ij} (\eps_i - \uvec_i) e_i^> =  \frac{1}{\sqrt{n}} \sum\limits_{i=1}^n \hat{Z}_{ij} (\eps_i - \uvec_i) e_i \, \ind(|e_i| > \log n), \]
and thus, as already proven in \eqref{eq:lemmaA2:tailbound1},  
\begin{equation}\label{eq:remainderB3:Pjge} 
P_{e,j}^> \le \pr_e \big( |e_i| > \log n \text{ for some } 1 \le i \le n \big) \le Cn^{-K},
\end{equation}
where $K$ can be made as large as desired. To complete the proof, we insert equations \eqref{eq:remainderB3:Pjle} and \eqref{eq:remainderB3:Pjge} into \eqref{eq:remainderB3:unionbound} and invoke condition \ref{C4}.
\end{proof}

\subsection*{Proof of (\ref{bound-EVT})}

Let $\Wgauss_1,\ldots,\Wgauss_p$ be independent normal random variables with $\ex[\Wgauss_j] = 0$ for all $j$ and suppose w.l.o.g.\ that $\ex[\Wgauss_j^2] = 1$ for all $j$. By Lemma \ref{lemmaA8},
\begin{equation}\label{eq1:bound-EVT}
\pr \Big( \max_{1 \le j \le p} |\Wgauss_j| \le \gamma_\alpha^{\Wgauss} \Big) = 1 - \alpha. 
\end{equation}
Moreover, standard arguments from classic extreme value theory show that 
\[ \pr \Big( \max_{1 \le j \le p} |\Wgauss_j| \le \frac{x}{a_p} + b_p \Big) \to e^{-2e^{-x}} \]
as $p \to \infty$ with $a_p = \sqrt{2 \log p}$ and $b_p = \sqrt{2\log p} - \{\log \log p + \log (4\pi)\}/\{2 \sqrt{2\log p}\}$, which in particular implies that for any fixed $\delta > 0$, 
\begin{equation}\label{eq2:bound-EVT}
\pr \Big( \max_{1 \le j \le p} |\Wgauss_j| \le \frac{x_{\alpha \pm \delta}}{a_p} + b_p \Big) \to 1 - \{ \alpha \pm \delta \} 
\end{equation}
with $x_{\alpha \pm \delta} = - \log( - \log(1 - \{\alpha \pm \delta\})/2 )$. From \eqref{eq1:bound-EVT} and \eqref{eq2:bound-EVT}, it follows that for any null sequence of positive numbers $\eta_p$, 
\[ \frac{x_{\alpha+\delta}}{a_p} + b_p \le \gamma_{\alpha+\eta_p}^{\Wgauss} \le \gamma_{\alpha-\eta_p}^{\Wgauss} \le \frac{x_{\alpha-\delta}}{a_p} + b_p \]
for $p$ sufficiently large. We thus arrive at 
\[ | \gamma_{\alpha-\eta_p}^{\Wgauss} - \gamma_{\alpha+\eta_p}^{\Wgauss}| \le \frac{x_{\alpha-\delta} - x_{\alpha+\delta}}{\sqrt{2\log p}} \le \frac{C}{\sqrt{2\log p}} \]
with some sufficiently large constant $C$.

\newpage

\phantom{Upper boundary}
\vspace{-2.25cm}

\section{Robustness checks}\label{sec:supp:robustness}

\subsection*{Choice of $\boldsymbol{\alpha}$ for tuning parameter calibration}

Our estimates of the quantiles of the effective noise can be used for different tasks,
with inference and tuning parameter calibration as two examples.
In inference, the choice of $\alpha$ is determined by the significance level.
In tuning parameter calibration, in contrast, $\alpha$ can be chosen freely.
In what follows, we examine how our tuning parameter calibration is influenced by the choice of $\alpha$. 
To do so, we repeat the simulation exercises from Section \ref{subsec:sim:2} (with $\kappa = 0.25$) 
for three different values of $\alpha$, namely $\alpha = 0.01, 0.05, 0.1$. 
Choosing $\alpha$ in the range between $0.01$ and $0.1$ in practice is sensible for the following reasons: The constraint $\alpha \le 0.1$ makes sure that the finite sample guarantees for tuning parameter calibration from Section \ref{subsec:tuning} hold with reasonably high probability ($\approx 90\%$ or higher). The constraint $\alpha \ge 0.01$, on the other hand, ensures that the bias of the lasso does not get overly strong. We thus restrict attention to $\alpha \in [0.01,0.1]$, which is also the range of typical significance levels in testing.

Figure \ref{fig:sim:tuning:alphas:Hamming} reports the results for the Hamming loss. 
The grey-shaded area in each panel depicts the histogram of the Hamming distances $\Delta_H(\hat{\beta},\beta^*)$ 
that are produced by our estimator $\hat{\beta}$ over the $N=1000$~simulation runs when $\alpha$ is set to $0.05$, the red line depicts the histogram for $\alpha = 0.01$, and the blue line the histogram for $\alpha = 0.1$. In addition, the histogram of the cross-validated estimator is shown as a dotted line. Notice that the grey-shaded histograms in Figure \ref{fig:sim:tuning:alphas:Hamming} are the same as those in Figure \ref{subfig:a:sim:Hamming}. Figures \ref{fig:sim:tuning:alphas:L1}--\ref{fig:sim:tuning:alphas:prediction} present the results for the $\ell_1$-, $\ell_\infty$- and prediction loss in an analogous way. Inspecting the plots, we conclude that the precise choice of $\alpha$ only has a minor effect on tuning parameter calibration with our method.

\begin{figure}[h!]
\vspace{0.25cm}

\centering
\includegraphics[width=\textwidth]{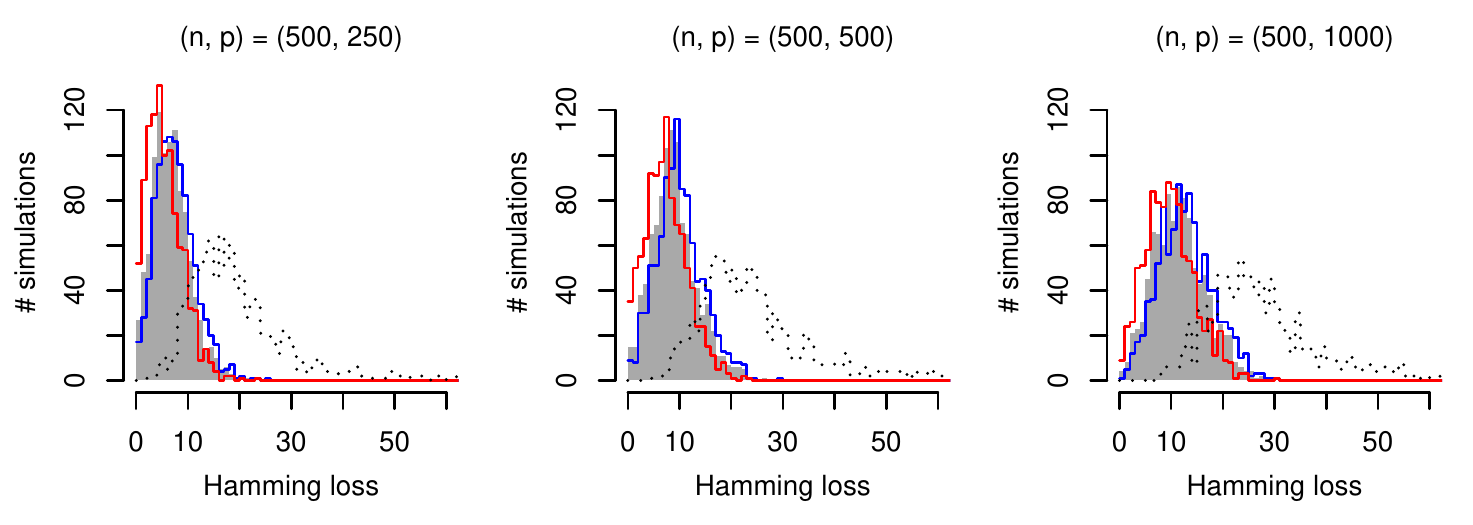}
\caption{Histograms of the Hamming loss for different values of $\alpha$.
}\label{fig:sim:tuning:alphas:Hamming}
\end{figure}

\begin{figure}[p]
\centering
\includegraphics[width=\textwidth]{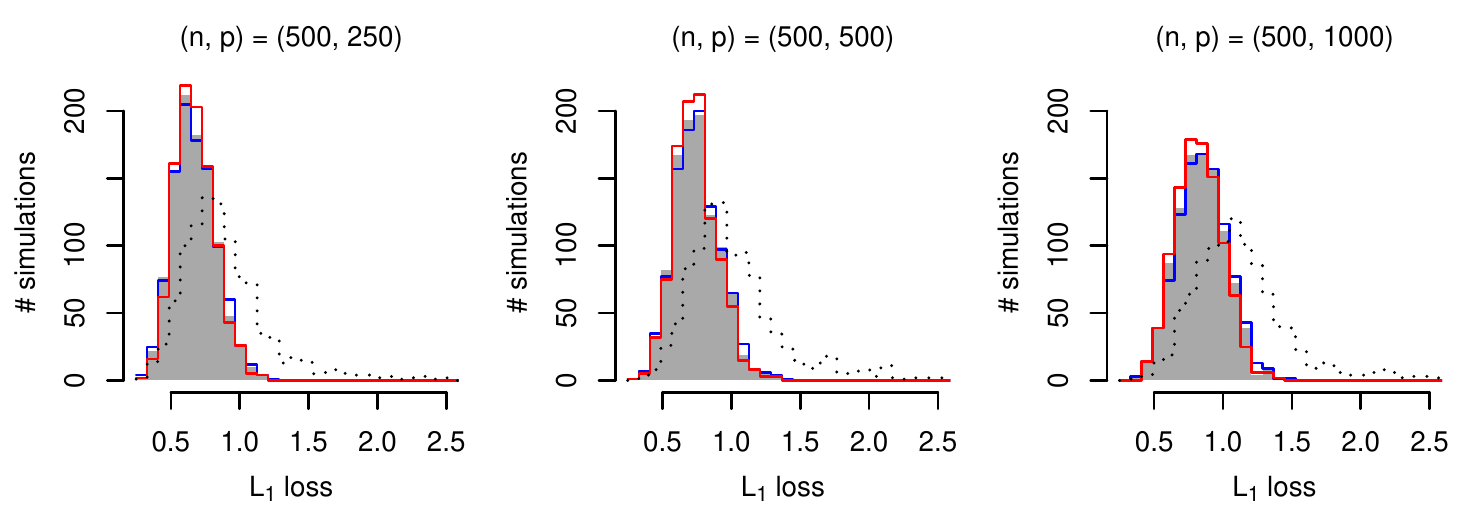}
\caption{Histograms of the $\ell_1$-loss for different values of $\alpha$.}\label{fig:sim:tuning:alphas:L1}
\vspace{0.3cm}

\includegraphics[width=\textwidth]{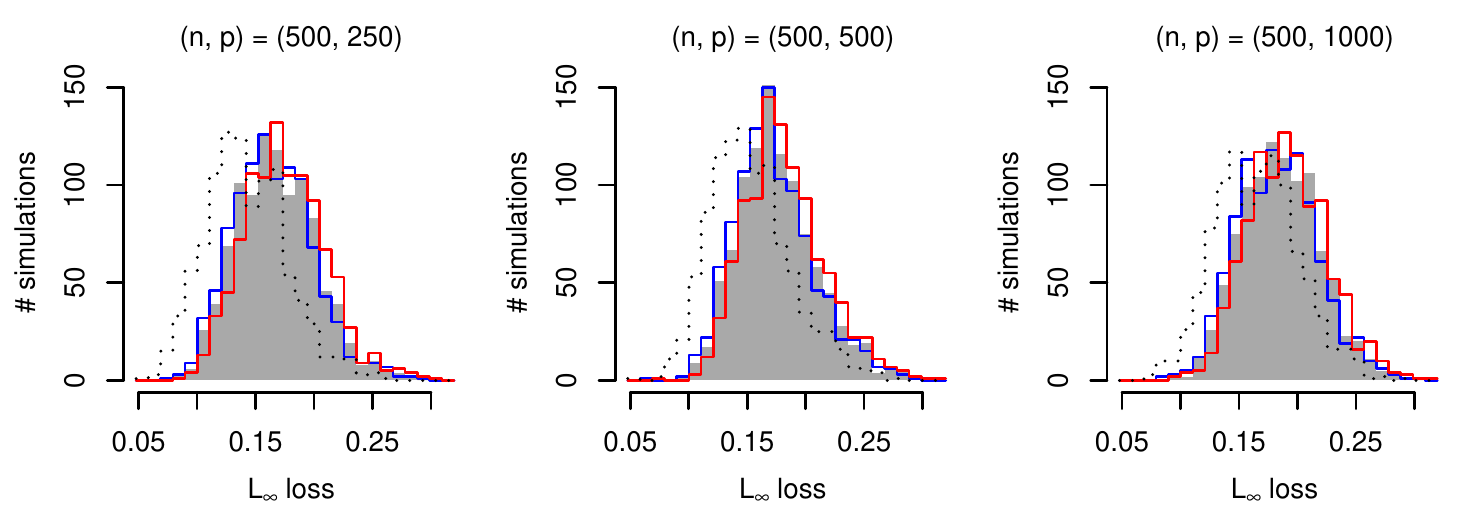}
\caption{Histograms of  the $\ell_\infty$-loss for different values of $\alpha$.}\label{fig:sim:tuning:alphas:sup}
\vspace{0.3cm}

\includegraphics[width=\textwidth]{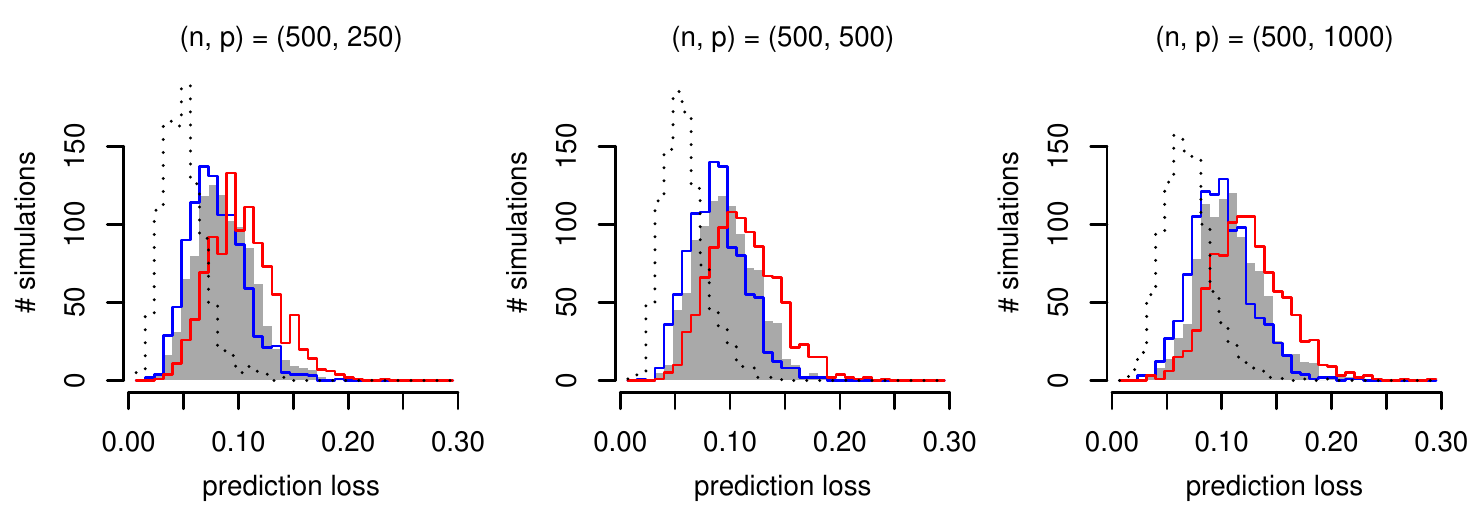}
\caption{Histograms of the prediction loss for different values of $\alpha$.}\label{fig:sim:tuning:alphas:prediction}
\end{figure}

\clearpage

\phantom{Upper boundary}
\vspace{-3.4cm}

\subsection*{Different distributions of the noise and the design}

In this section, we investigate how our simulation results are influenced by the distribution of the noise $\eps_i$ and the design $\Xvec_i$. In order to do so, we repeat the simulation exercises from Sections \ref{subsec:sim:1}--\ref{subsec:sim:3} with non-normal noise variables $\eps_i$ and design vectors $\Xvec_i$. Specifically, we sample $\eps_i$ independently from a $t$-distribution with $d$ degrees of freedom and variance normalized to $1$. Moreover, $\Xvec_i$ is drawn from a multivariate $t$-distribution with the same number of degrees of freedom, where the covariance matrix is the same as in Section \ref{sec:sim} (in particular, it is given by $(1-\kappa) \idmat + \kappa \boldsymbol{E}$ with $\kappa = 0.25$). We consider three different choices of $d$, namely $d \in \{5,10,30\}$. For small $d$, the $t$-distribution differs substantially from the standard normal law, having much heavier tails. (Note in particular that $d=5$ is the smallest integer for which the $t$-distribution has $\theta > 4$ moments as required by condition \ref{C3}.) As $d$ increases, the $t$-distribution becomes less heavy-tailed and more akin to a standard normal law.

We start with the simulations from Section \ref{subsec:sim:1}, 
which concern the approximation quality of our estimator $\hat{\lambda}_\alpha$. To see how the quality of $\hat{\lambda}_\alpha$ depends on the distribution of the noise and the design, we reproduce Figure \ref{fig:sim:histlambda} for the case of $t$-distributed errors and design vectors with $d \in \{5,10,30\}$. The results are reported in Figure \ref{fig:robust:dist:tuning:lambda:t}. As can be seen, the precision of our estimator diminishes somewhat as $d$ gets smaller. Nevertheless, even for the case $d=5$, we obtain quite precise results.

We now turn to the simulations on tuning parameter calibration from Section~\ref{subsec:sim:2}. We reproduce Figures \ref{subfig:a:sim:Hamming}, \ref{fig:sim:L1}, \ref{fig:sim:sup} and \ref{fig:sim:prediction}, which correspond to the four different losses under consideration, for the case of $t$-distributed noise terms and design vectors with $d \in \{5,10,30\}$. The results are presented in Figures \ref{fig:robust:dist:tuning:Hamming:t}--\ref{fig:robust:dist:tuning:prediction:t}. The format is the same as in Figures \ref{subfig:a:sim:Hamming}, \ref{fig:sim:L1}, \ref{fig:sim:sup} and \ref{fig:sim:prediction}: the grey-shaded areas correspond to the histograms produced by our estimator, the black lines correspond to the histograms of the oracle method, and the dotted lines correspond to the histograms of the cross-validated lasso. In all of the considered cases, the histograms of our estimator are extremely close to those of the oracle. Moreover, the histograms are very similar to those of Figures \ref{subfig:a:sim:Hamming}, \ref{fig:sim:L1}, \ref{fig:sim:sup} and \ref{fig:sim:prediction} for the Gaussian case.

We finally revisit the inference results from Section \ref{subsec:sim:3}. As before, we repeat the simulations with $t$-distributed noise and design vectors for $d \in \{5,10,30\}$. The results are given in Tables \ref{table:sim:test:t5}--\ref{table:sim:test:t30}. For all considered values of $d$, the size of the test under the null is close to the target $\alpha$. Moreover, the power of the test is comparable to that in the Gaussian case, even though it gets a bit lower for smaller $d$.

To summarize, the results demonstrate that our method does not require normally distributed noise and design, which supports our general theory in the main part of the paper.

\begin{figure}[h!]
\begin{subfigure}[b]{\textwidth}
\centering
\includegraphics[width=\textwidth]{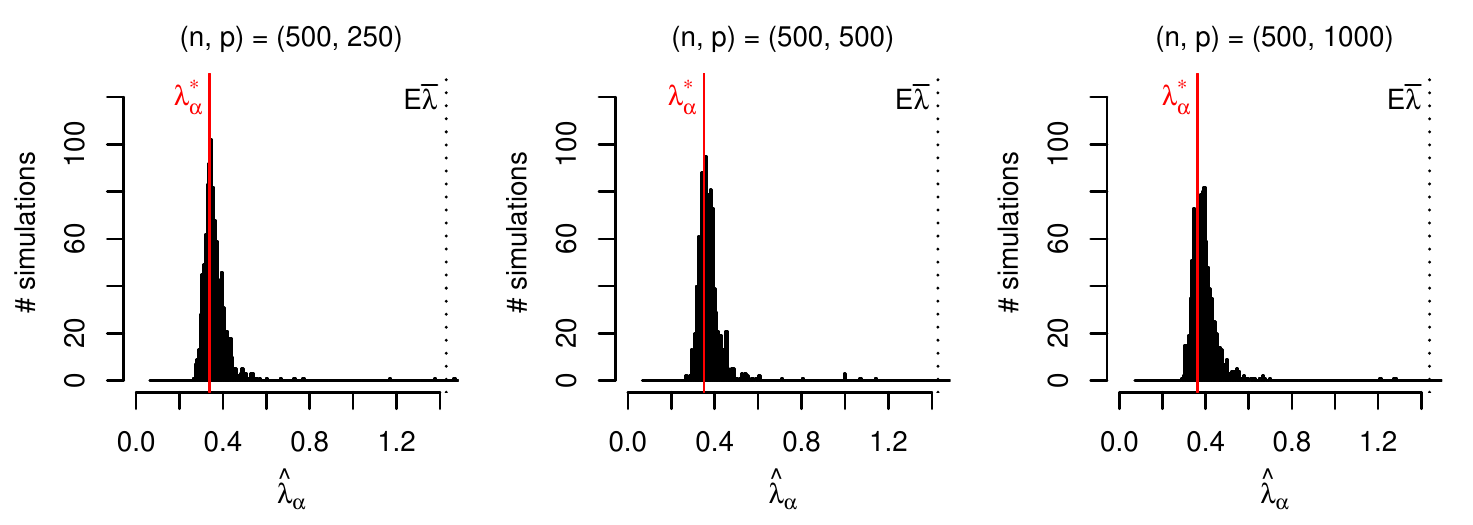}
\caption{histograms for $t$-distributed noise and design with $d=5$}\label{subfig:robust:dist:tuning:lambda:t5}
\vspace{0.3cm}

\end{subfigure}

\begin{subfigure}[b]{\textwidth}
\centering
\includegraphics[width=\textwidth]{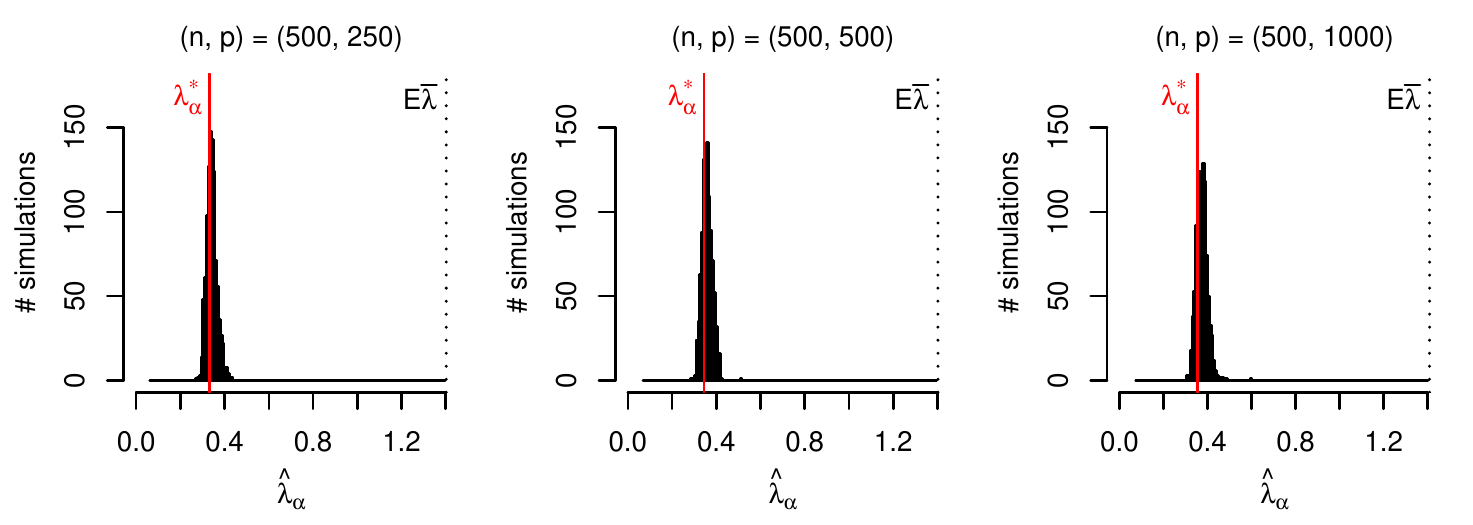}
\caption{histograms for $t$-distributed noise and design with $d=10$}\label{subfig:robust:dist:tuning:lambda:t10}
\vspace{0.3cm}

\end{subfigure}

\begin{subfigure}[b]{\textwidth}
\centering
\includegraphics[width=\textwidth]{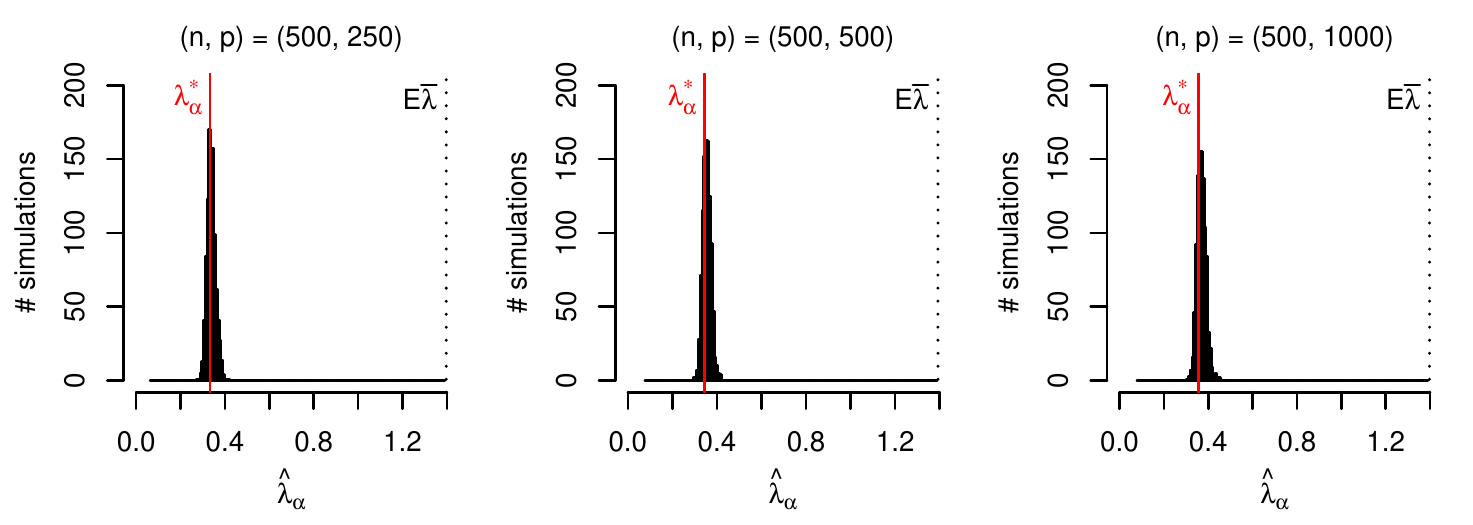}
\caption{histograms for $t$-distributed noise and design with $d=30$}\label{subfig:robust:dist:tuning:lambda:t30}
\end{subfigure}
\caption{Histograms of the estimates $\hat{\lambda}_\alpha$ for $t$-distributed noise variables and design vectors with $d \in \{5,10,30\}$.}
\label{fig:robust:dist:tuning:lambda:t}
\end{figure}

\begin{figure}[p!]
\centering

\begin{subfigure}[b]{\textwidth}
\centering
\includegraphics[width=\textwidth]{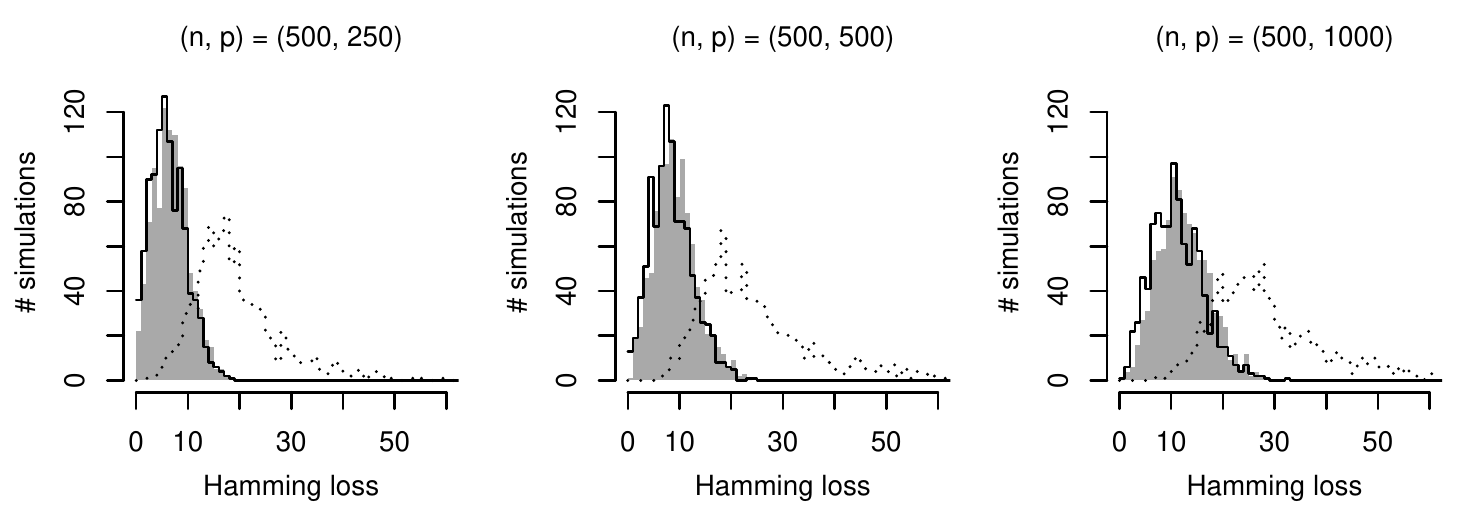}
\caption{histograms for $t$-distributed noise and design with $d=5$}\label{subfig:robust:dist:tuning:Hamming:t5}
\vspace{0.25cm}
\end{subfigure} 

\begin{subfigure}[b]{\textwidth}
\centering
\includegraphics[width=\textwidth]{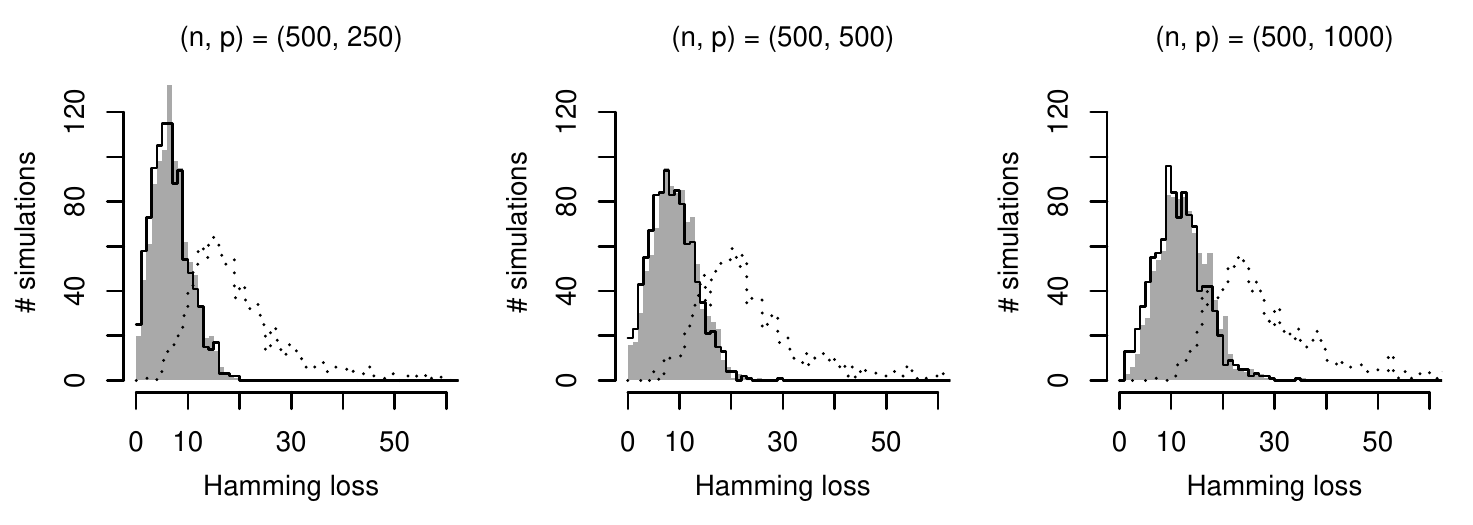}
\caption{histograms for $t$-distributed noise and design with $d=10$}\label{subfig:robust:dist:tuning:Hamming:t10}
\vspace{0.25cm}
\end{subfigure}

\begin{subfigure}[b]{\textwidth}
\centering
\includegraphics[width=\textwidth]{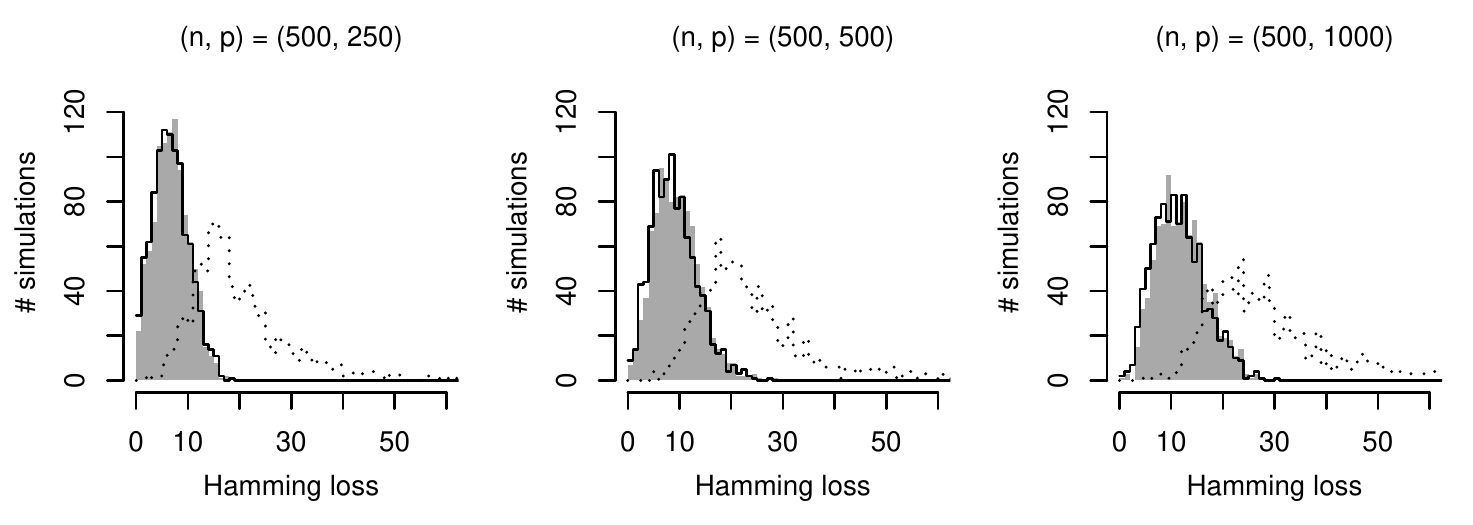}
\caption{histograms for $t$-distributed noise and design with $d=30$}\label{subfig:robust:dist:tuning:Hamming:t30}
\end{subfigure}
\caption{Histograms of the Hamming loss for $t$-distributed noise variables and design vectors with $d \in \{5,10,30\}$.}
\label{fig:robust:dist:tuning:Hamming:t}
\end{figure}

\begin{figure}[p!]
\centering

\begin{subfigure}[b]{\textwidth}
\centering
\includegraphics[width=\textwidth]{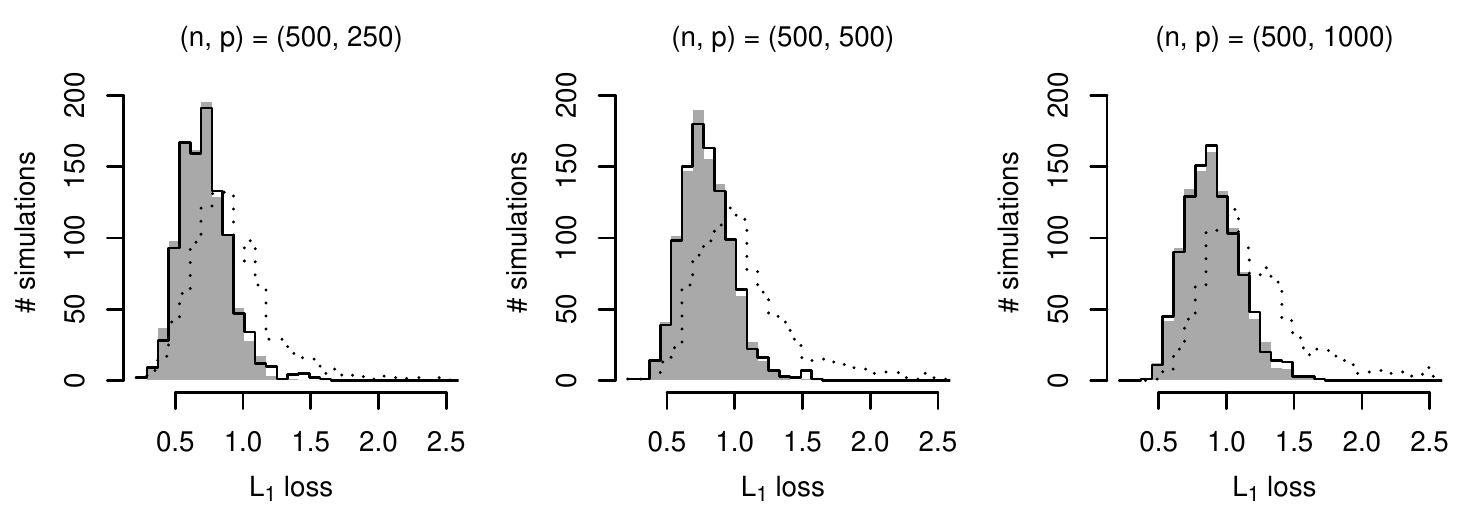}
\caption{histograms for $t$-distributed noise and design with $d=5$}\label{subfig:robust:dist:tuning:L1:t5}
\vspace{0.25cm}
\end{subfigure} 

\begin{subfigure}[b]{\textwidth}
\centering
\includegraphics[width=\textwidth]{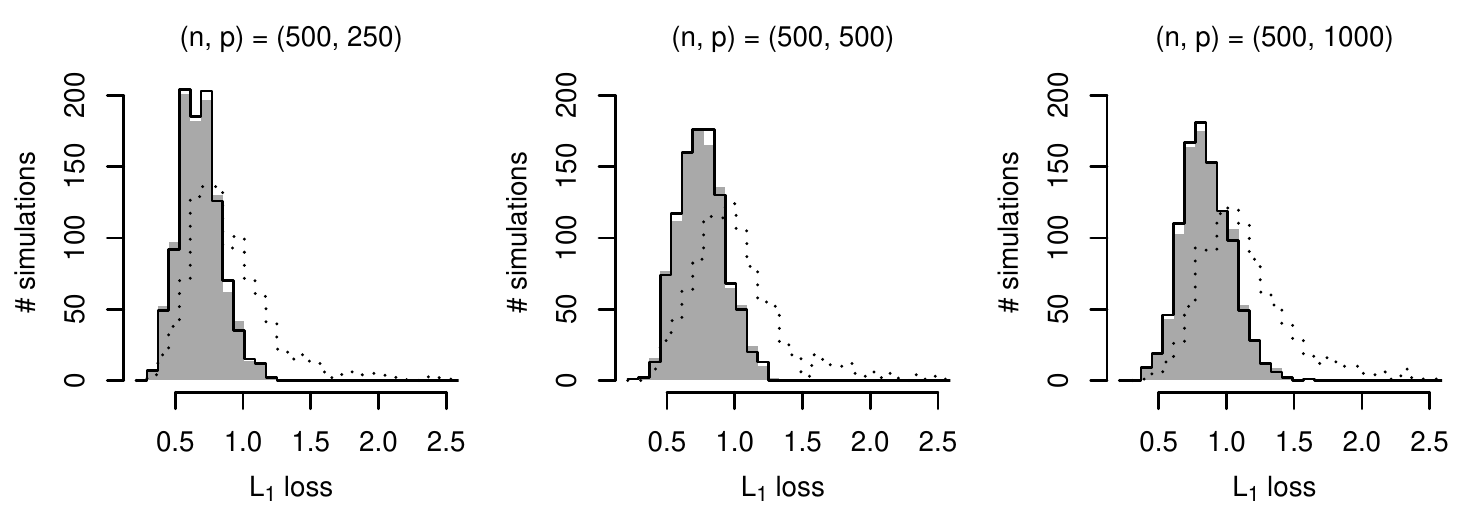}
\caption{histograms for $t$-distributed noise and design with $d=10$}\label{subfig:robust:dist:tuning:L1:t10}
\vspace{0.25cm}
\end{subfigure}

\begin{subfigure}[b]{\textwidth}
\centering
\includegraphics[width=\textwidth]{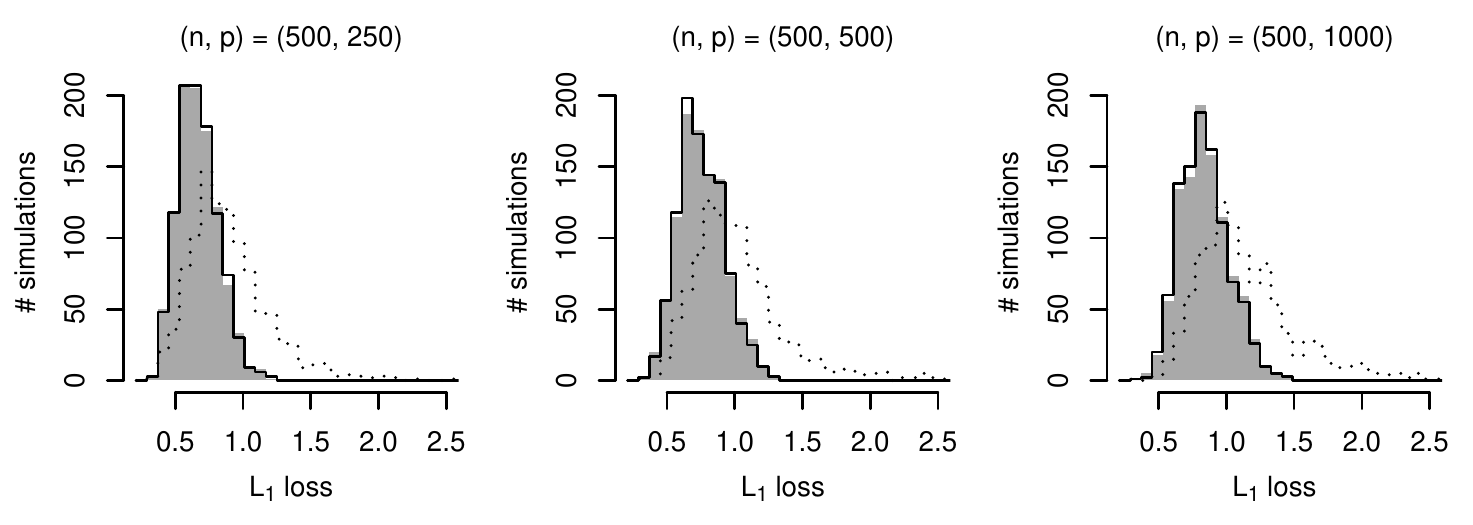}
\caption{histograms for $t$-distributed noise and design with $d=30$}\label{subfig:robust:dist:tuning:L1:t30}
\end{subfigure}
\caption{Histograms of the $\ell_1$-loss for $t$-distributed noise variables and design vectors with $d \in \{5,10,30\}$.}
\label{fig:robust:dist:tuning:L1:t}
\end{figure}

\begin{figure}[p!]
\centering

\begin{subfigure}[b]{\textwidth}
\centering
\includegraphics[width=\textwidth]{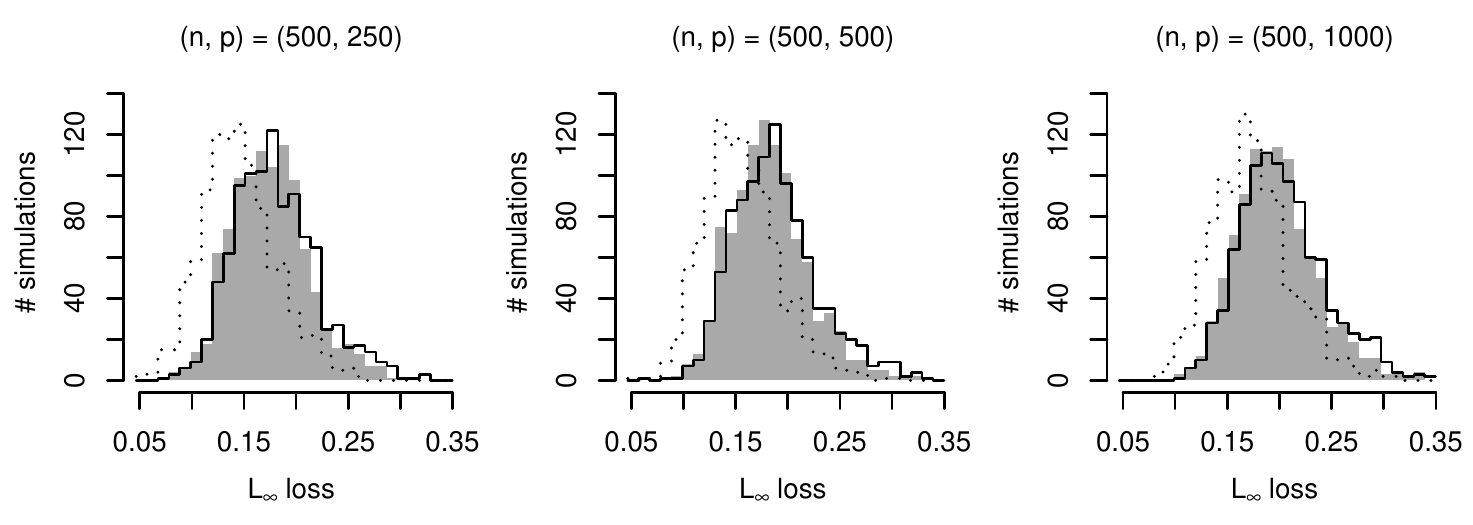}
\caption{histograms for $t$-distributed noise and design with $d=5$}\label{subfig:robust:dist:tuning:sup:t5}
\vspace{0.25cm}
\end{subfigure} 

\begin{subfigure}[b]{\textwidth}
\centering
\includegraphics[width=\textwidth]{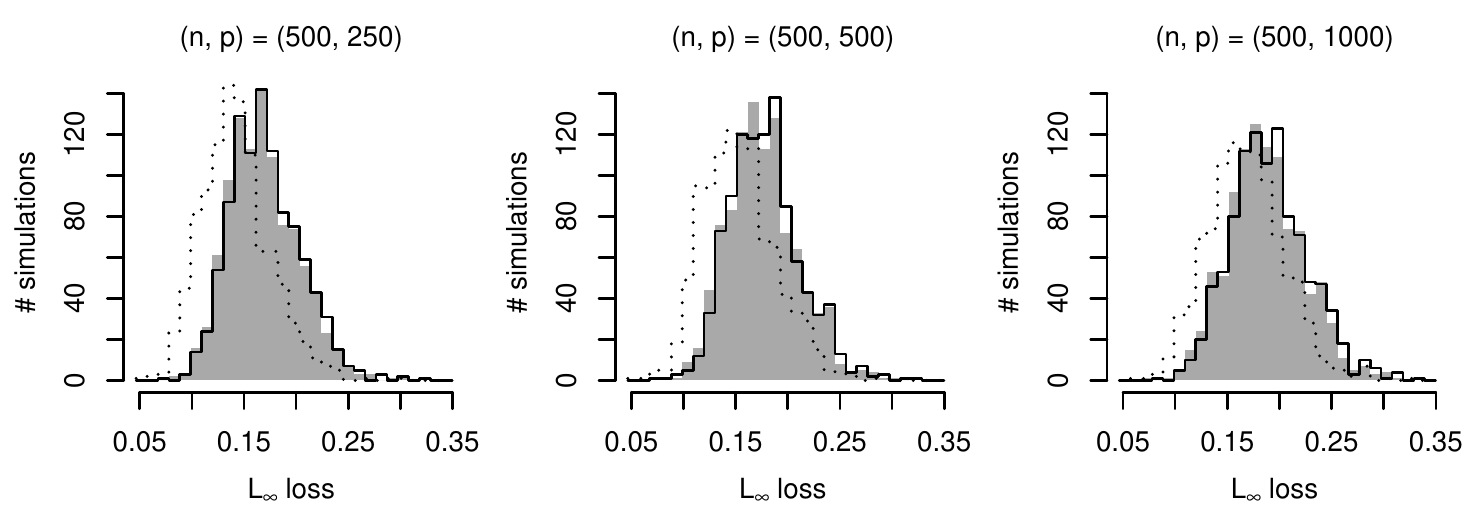}
\caption{histograms for $t$-distributed noise and design with $d=10$}\label{subfig:robust:dist:tuning:sup:t10}
\vspace{0.25cm}
\end{subfigure}

\begin{subfigure}[b]{\textwidth}
\centering
\includegraphics[width=\textwidth]{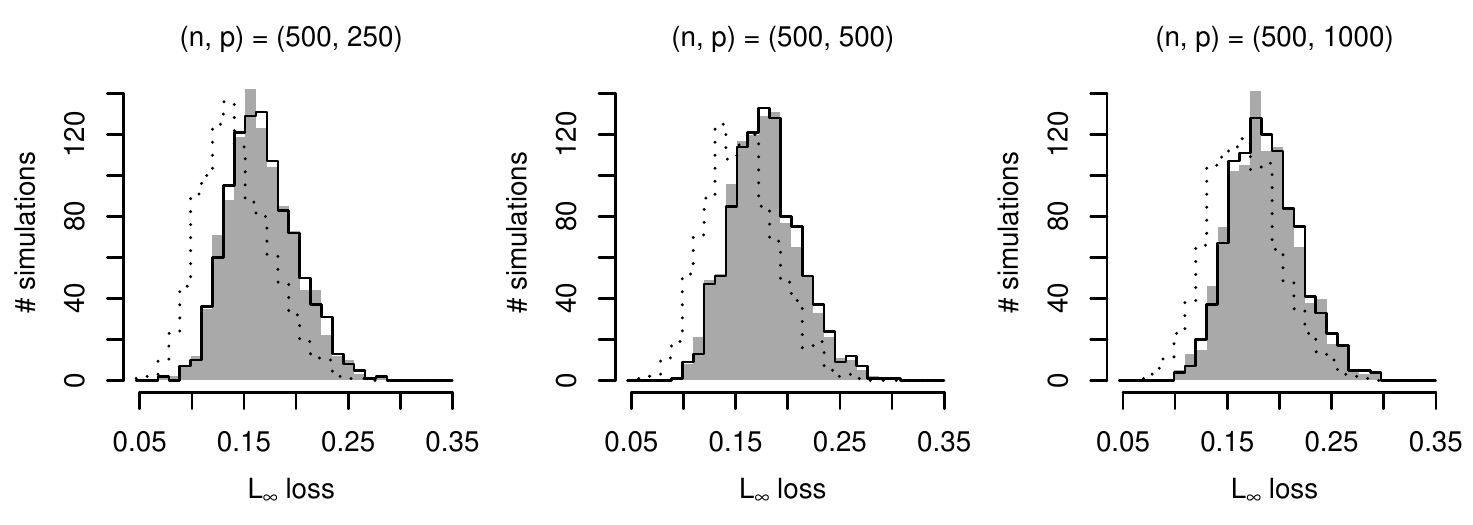}
\caption{histograms for $t$-distributed noise and design with $d=30$}\label{subfig:robust:dist:tuning:sup:t30}
\end{subfigure}
\caption{Histograms of the $\ell_\infty$-loss for $t$-distributed noise variables and design vectors with $d \in \{5,10,30\}$.}
\label{fig:robust:dist:tuning:sup:t}
\end{figure}

\begin{figure}[p!]
\centering

\begin{subfigure}[b]{\textwidth}
\centering
\includegraphics[width=\textwidth]{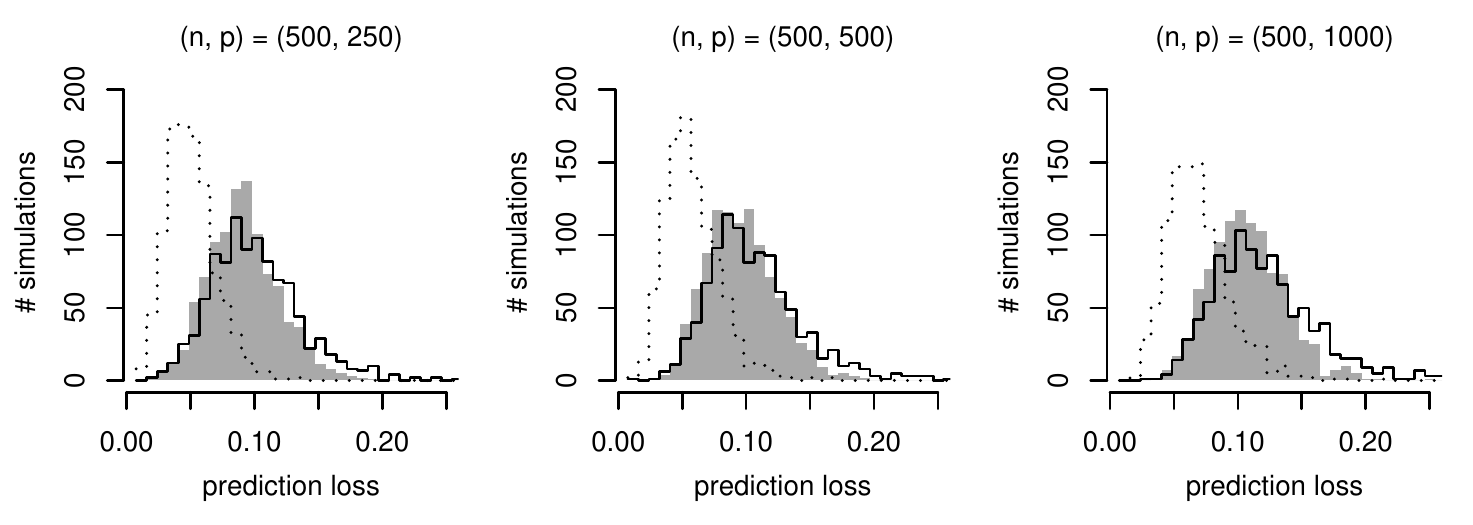}
\caption{histograms for $t$-distributed noise and design with $d=5$}\label{subfig:robust:dist:tuning:prediction:t5}
\vspace{0.25cm}
\end{subfigure} 

\begin{subfigure}[b]{\textwidth}
\centering
\includegraphics[width=\textwidth]{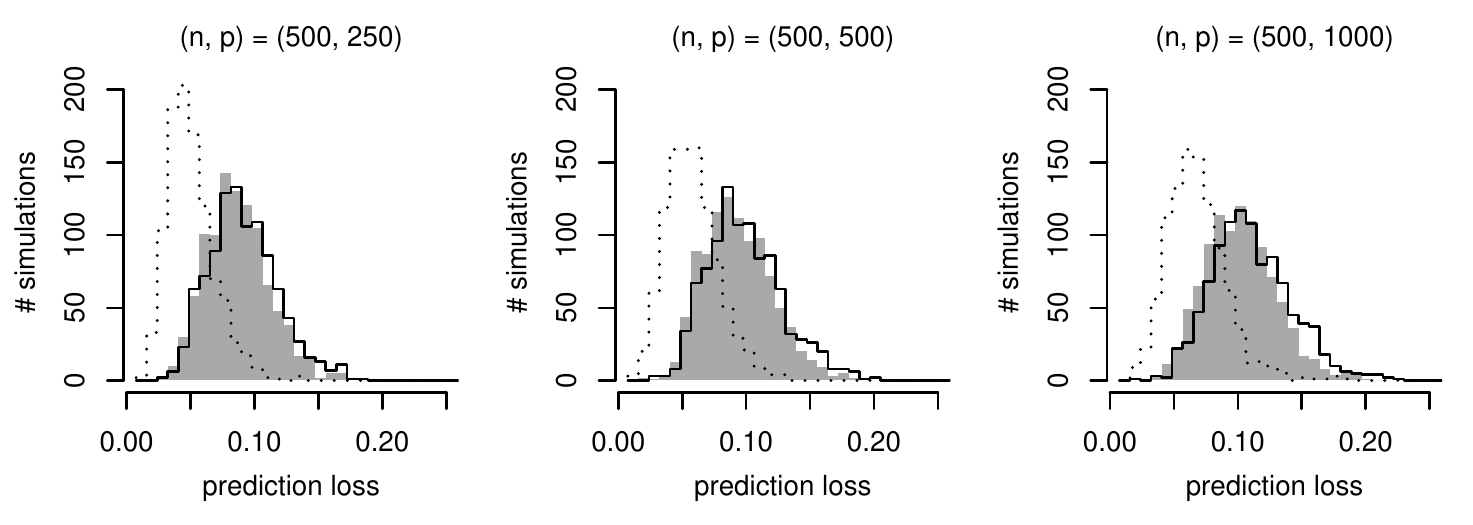}
\caption{histograms for $t$-distributed noise and design with $d=10$}\label{subfig:robust:dist:tuning:prediction:t10}
\vspace{0.25cm}
\end{subfigure}

\begin{subfigure}[b]{\textwidth}
\centering
\includegraphics[width=\textwidth]{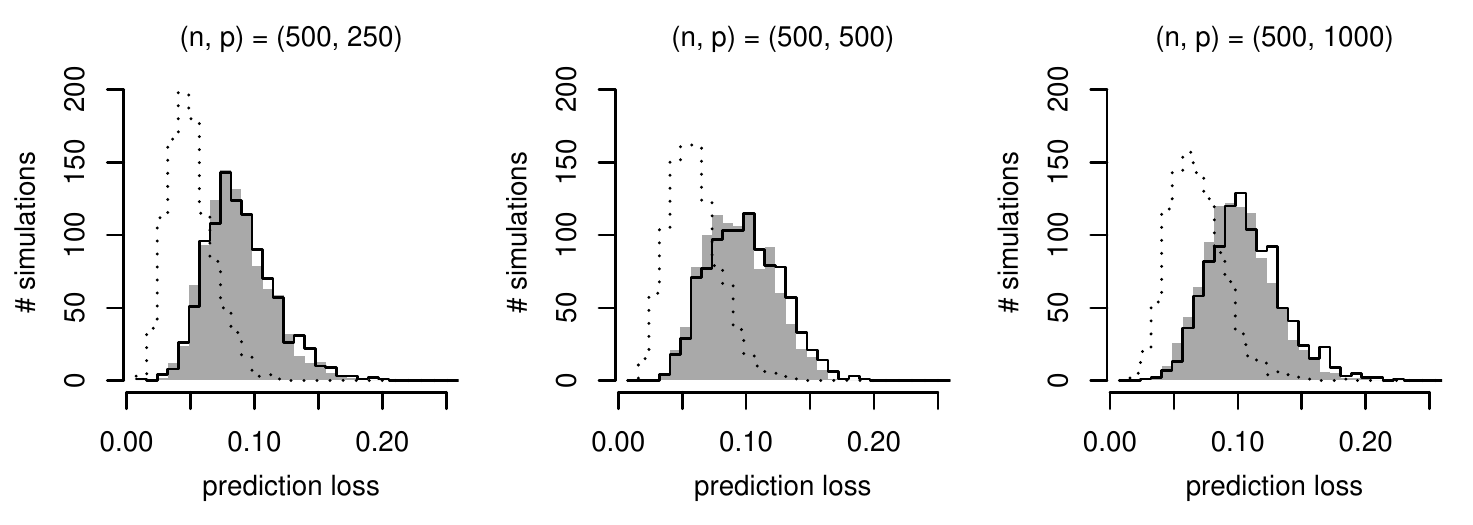}
\caption{histograms for $t$-distributed noise and design with $d=30$}\label{subfig:robust:dist:tuning:prediction:t30}
\end{subfigure}
\caption{Histograms of the prediction loss for $t$-distributed noise variables and design vectors with $d \in \{5,10,30\}$.}
\label{fig:robust:dist:tuning:prediction:t}
\end{figure}

\FloatBarrier

\begin{table}[ph!] 
\vspace{3.25cm}

\setlength{\tabcolsep}{2pt}
\centering 
\caption{Empirical size under the null and power against different alternatives for $t$-distributed noise variables and design vectors with $d=5$.}\label{table:sim:test:t5}

{\small 
\begin{subtable}[b]{\textwidth}
\centering 
\caption{empirical size under $H_0: \beta^* = 0$}\label{subtable:sim:test:null:t5} 
\begin{tabular}{@{\extracolsep{5pt}} lcccccccc} 
\\[-1.8ex]\hline 
\hline \\[-1.8ex] 
 & & \multicolumn{3}{c}{feasible test} & & \multicolumn{3}{c}{oracle test} \\
 & & $\alpha=0.01$ & $\alpha=0.05$ & $\alpha=0.1$ & & $\alpha=0.01$ & $\alpha=0.05$ & $\alpha=0.1$ \\[0.1cm]
\hline \\[-1.8ex] 
$(n, p) = (500, 250)$ & & $0.011$ & $0.033$ & $0.080$ & & $0.009$ & $0.058$ & $0.105$ \\ 
$(n, p) = (500, 500)$ & & $0.009$ & $0.036$ & $0.078$ & & $0.013$ & $0.054$ & $0.094$ \\ 
$(n, p) = (500, 1000)$ & & $0.007$ & $0.028$ & $0.067$ & & $0.018$ & $0.061$ & $0.095$ \\[0.1cm] 
\hline\\ 
\end{tabular} 
\end{subtable}

\begin{subtable}[b]{\textwidth}
\centering 
\caption{empirical power under the alternative with $\text{SNR} = 0.1$}\label{subtable:sim:test:alt01:t5} 
\begin{tabular}{@{\extracolsep{5pt}} lcccccccc} 
\\[-1.8ex]\hline 
\hline \\[-1.8ex] 
 & & \multicolumn{3}{c}{feasible test} & & \multicolumn{3}{c}{oracle test} \\
 & & $\alpha=0.01$ & $\alpha=0.05$ & $\alpha=0.1$ & & $\alpha=0.01$ & $\alpha=0.05$ & $\alpha=0.1$ \\[0.1cm]
\hline \\[-1.8ex] 
$(n, p) = (500, 250)$ & & $0.092$ & $0.220$ & $0.346$ & & $0.060$ & $0.247$ & $0.360$ \\ 
$(n, p) = (500, 500)$ & & $0.100$ & $0.247$ & $0.401$ & & $0.107$ & $0.298$ & $0.419$ \\ 
$(n, p) = (500, 1000)$ & & $0.085$ & $0.223$ & $0.365$ & & $0.139$ & $0.309$ & $0.397$ \\[0.1cm] 
\hline \\
\end{tabular} 
\end{subtable}

\begin{subtable}[b]{\textwidth}
\centering 
\caption{empirical power under the alternative with $\text{SNR} = 0.2$}\label{subtable:sim:test:alt02:t5}
\begin{tabular}{@{\extracolsep{5pt}} lcccccccc} 
\\[-1.8ex]\hline 
\hline \\[-1.8ex] 
 & & \multicolumn{3}{c}{feasible test} & & \multicolumn{3}{c}{oracle test} \\
 & & $\alpha=0.01$ & $\alpha=0.05$ & $\alpha=0.1$ & & $\alpha=0.01$ & $\alpha=0.05$ & $\alpha=0.1$ \\[0.1cm]
\hline \\[-1.8ex] 
$(n, p) = (500, 250)$ & & $0.471$ & $0.741$ & $0.856$ & & $0.602$ & $0.836$ & $0.917$ \\ 
$(n, p) = (500, 500)$ & & $0.510$ & $0.762$ & $0.874$ & & $0.617$ & $0.865$ & $0.929$ \\ 
$(n, p) = (500, 1000)$ & & $0.453$ & $0.725$ & $0.852$ & & $0.656$ & $0.843$ & $0.908$ \\[0.1cm] 
\hline \\[-1.8ex] 
\end{tabular} 
\end{subtable}}

\vspace{-0.45cm}
\end{table}

\begin{table}[hp!] 
\setlength{\tabcolsep}{2pt}
\centering 
\caption{Empirical size under the null and power against different alternatives for $t$-distributed noise variables and design vectors with $d=10$.}\label{table:sim:test:t10}

{\small 
\begin{subtable}[b]{\textwidth}
\centering 
\caption{empirical size under $H_0: \beta^* = 0$}\label{subtable:sim:test:null:t10} 
\begin{tabular}{@{\extracolsep{5pt}} lcccccccc} 
\\[-1.8ex]\hline 
\hline \\[-1.8ex] 
 & & \multicolumn{3}{c}{feasible test} & & \multicolumn{3}{c}{oracle test} \\
 & & $\alpha=0.01$ & $\alpha=0.05$ & $\alpha=0.1$ & & $\alpha=0.01$ & $\alpha=0.05$ & $\alpha=0.1$ \\[0.1cm]
\hline \\[-1.8ex] 
$(n, p) = (500, 250)$ & & $0.012$ & $0.054$ & $0.087$ & & $0.011$ & $0.045$ & $0.095$ \\ 
$(n, p) = (500, 500)$ & & $0.012$ & $0.047$ & $0.100$ & & $0.004$ & $0.053$ & $0.102$ \\ 
$(n, p) = (500, 1000)$ & & $0.005$ & $0.033$ & $0.080$ & & $0.005$ & $0.041$ & $0.084$ \\[0.1cm] 
\hline\\ 
\end{tabular} 
\end{subtable}

\begin{subtable}[b]{\textwidth}
\centering 
\caption{empirical power under the alternative with $\text{SNR} = 0.1$}\label{subtable:sim:test:alt01:t10} 
\begin{tabular}{@{\extracolsep{5pt}} lcccccccc} 
\\[-1.8ex]\hline 
\hline \\[-1.8ex] 
 & & \multicolumn{3}{c}{feasible test} & & \multicolumn{3}{c}{oracle test} \\
 & & $\alpha=0.01$ & $\alpha=0.05$ & $\alpha=0.1$ & & $\alpha=0.01$ & $\alpha=0.05$ & $\alpha=0.1$ \\[0.1cm]
\hline \\[-1.8ex] 
$(n, p) = (500, 250)$ & & $0.111$ & $0.260$ & $0.403$ & & $0.109$ & $0.279$ & $0.399$ \\ 
$(n, p) = (500, 500)$ & & $0.119$ & $0.270$ & $0.393$ & & $0.096$ & $0.297$ & $0.418$ \\ 
$(n, p) = (500, 1000)$ & & $0.106$ & $0.247$ & $0.374$ & & $0.088$ & $0.262$ & $0.376$ \\[0.1cm] 
\hline \\
\end{tabular} 
\end{subtable}

\begin{subtable}[b]{\textwidth}
\centering 
\caption{empirical power under the alternative with $\text{SNR} = 0.2$}\label{subtable:sim:test:alt02:t10}
\begin{tabular}{@{\extracolsep{5pt}} lcccccccc} 
\\[-1.8ex]\hline 
\hline \\[-1.8ex] 
 & & \multicolumn{3}{c}{feasible test} & & \multicolumn{3}{c}{oracle test} \\
 & & $\alpha=0.01$ & $\alpha=0.05$ & $\alpha=0.1$ & & $\alpha=0.01$ & $\alpha=0.05$ & $\alpha=0.1$ \\[0.1cm]
\hline \\[-1.8ex] 
$(n, p) = (500, 250)$ & & $0.607$ & $0.822$ & $0.922$ & & $0.649$ & $0.860$ & $0.938$ \\ 
$(n, p) = (500, 500)$ & & $0.578$ & $0.790$ & $0.891$ & & $0.592$ & $0.832$ & $0.910$ \\ 
$(n, p) = (500, 1000)$ & & $0.556$ & $0.806$ & $0.895$ & & $0.567$ & $0.851$ & $0.909$ \\[0.1cm] 
\hline \\[-1.8ex] 
\end{tabular} 
\end{subtable}}

\vspace{-0.45cm}
\end{table}

\begin{table}[hp!] 
\setlength{\tabcolsep}{2pt}
\centering 
\caption{Empirical size under the null and power against different alternatives for $t$-distributed noise variables and design vectors with $d=30$.}\label{table:sim:test:t30}

{\small 
\begin{subtable}[b]{\textwidth}
\centering 
\caption{empirical size under $H_0: \beta^* = 0$}\label{subtable:sim:test:null:t30} 
\begin{tabular}{@{\extracolsep{5pt}} lcccccccc} 
\\[-1.8ex]\hline 
\hline \\[-1.8ex] 
 & & \multicolumn{3}{c}{feasible test} & & \multicolumn{3}{c}{oracle test} \\
 & & $\alpha=0.01$ & $\alpha=0.05$ & $\alpha=0.1$ & & $\alpha=0.01$ & $\alpha=0.05$ & $\alpha=0.1$ \\[0.1cm]
\hline \\[-1.8ex] 
$(n, p) = (500, 250)$ & & $0.012$ & $0.057$ & $0.111$ & & $0.007$ & $0.056$ & $0.105$ \\ 
$(n, p) = (500, 500)$ & & $0.019$ & $0.057$ & $0.099$ & & $0.012$ & $0.070$ & $0.115$ \\ 
$(n, p) = (500, 1000)$ & & $0.011$ & $0.044$ & $0.082$ & & $0.004$ & $0.055$ & $0.101$ \\[0.1cm] 
\hline\\ 
\end{tabular} 
\end{subtable}

\begin{subtable}[b]{\textwidth}
\centering 
\caption{empirical power under the alternative with $\text{SNR} = 0.1$}\label{subtable:sim:test:alt01:t30} 
\begin{tabular}{@{\extracolsep{5pt}} lcccccccc} 
\\[-1.8ex]\hline 
\hline \\[-1.8ex] 
 & & \multicolumn{3}{c}{feasible test} & & \multicolumn{3}{c}{oracle test} \\
 & & $\alpha=0.01$ & $\alpha=0.05$ & $\alpha=0.1$ & & $\alpha=0.01$ & $\alpha=0.05$ & $\alpha=0.1$ \\[0.1cm]
\hline \\[-1.8ex] 
$(n, p) = (500, 250)$ & & $0.120$ & $0.275$ & $0.421$ & & $0.113$ & $0.288$ & $0.417$ \\ 
$(n, p) = (500, 500)$ & & $0.130$ & $0.278$ & $0.391$ & & $0.129$ & $0.335$ & $0.437$ \\ 
$(n, p) = (500, 1000)$ & & $0.140$ & $0.287$ & $0.407$ & & $0.119$ & $0.327$ & $0.446$ \\[0.1cm] 
\hline \\
\end{tabular} 
\end{subtable}

\begin{subtable}[b]{\textwidth}
\centering 
\caption{empirical power under the alternative with $\text{SNR} = 0.2$}\label{subtable:sim:test:alt02:t30}
\begin{tabular}{@{\extracolsep{5pt}} lcccccccc} 
\\[-1.8ex]\hline 
\hline \\[-1.8ex] 
 & & \multicolumn{3}{c}{feasible test} & & \multicolumn{3}{c}{oracle test} \\
 & & $\alpha=0.01$ & $\alpha=0.05$ & $\alpha=0.1$ & & $\alpha=0.01$ & $\alpha=0.05$ & $\alpha=0.1$ \\[0.1cm]
\hline \\[-1.8ex]
$(n, p) = (500, 250)$ & & $0.629$ & $0.824$ & $0.908$ & & $0.653$ & $0.855$ & $0.922$ \\ 
$(n, p) = (500, 500)$ & & $0.605$ & $0.823$ & $0.905$ & & $0.660$ & $0.892$ & $0.935$ \\ 
$(n, p) = (500, 1000)$ & & $0.619$ & $0.831$ & $0.906$ & & $0.620$ & $0.854$ & $0.935$ \\[0.1cm] 
\hline \\[-1.8ex] 
\end{tabular} 
\end{subtable}}

\vspace{-0.45cm}
\end{table}

\clearpage 
\phantom{Upper boundary}
\vspace{-1.85cm}

\subsection*{Choice of $\boldsymbol{L}$ and $\boldsymbol{M}$}

When implementing our method, we need to choose the number of bootstrap iterations $L$ as well as the grid size $M$ for computing the lasso estimates. We have experimented with different choices of $L$ and $M$ and found that they have little effect on the simulation results. To illustrate this, we consider the same simulation setting as in Section \ref{subsec:sim:1} and produce $N=1000$ estimates of $\hat{\lambda}_\alpha$ for different choices of $(L,M)$. In addition to the choice $(L,M) = (100,100)$ which is used in Section~\ref{sec:sim}, we consider the choices $(L,M) = (200,200)$ and $(L,M)=(300,300)$. Figure \ref{fig:robust:LM} reports the results. In each panel, the grey-shaded area is the histogram of the $N=1000$ estimates of $\hat{\lambda}_\alpha$ for the choice $(L,M)=(100,100)$, the blue line is the histogram for $(L,M) = (200,200)$, and the red line is the histogram for $(L,M) = (300,300)$. As one can see, the histograms are very similar across the different choices of $(L,M)$, which suggests that the precise choice of $L$ and $M$ has little effect on our method.

\begin{figure}[h!]
\vspace{0.25cm}

\centering
\includegraphics[width=\textwidth]{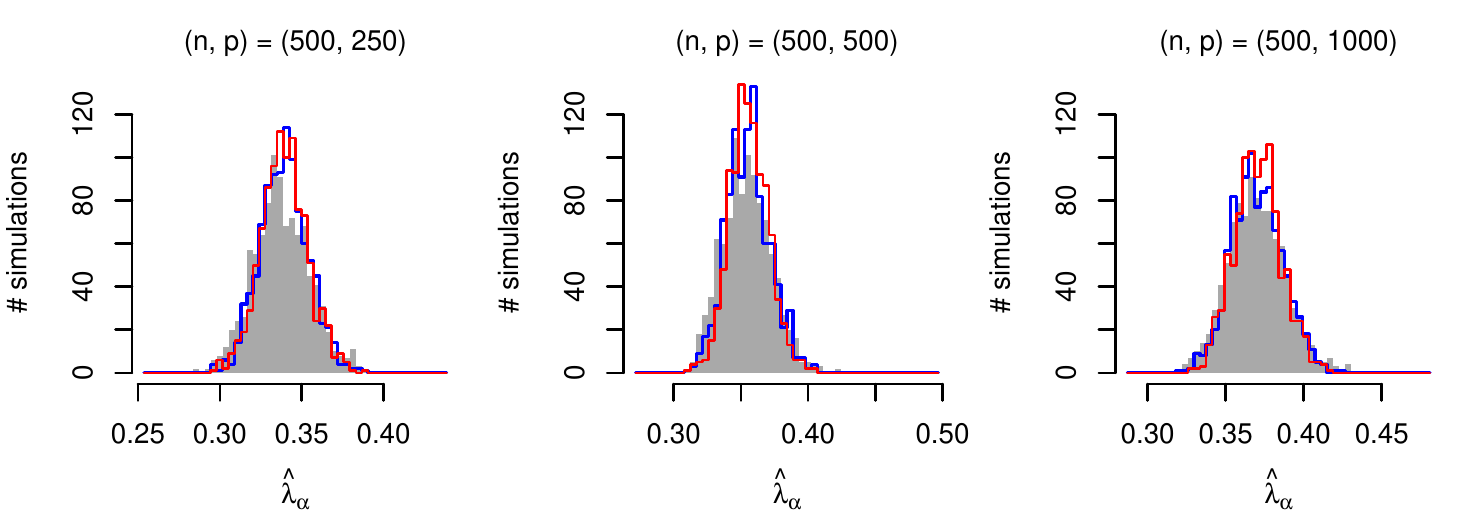}
\caption{Histograms of the estimates $\hat{\lambda}_\alpha$ for different choices of $(L,M)$.}\label{fig:robust:LM}
\end{figure}

\end{document}